\documentclass[runningheads,orivec]{llncs}
\usepackage{hyperref}
\usepackage{breakurl}
\usepackage{fancyvrb}
\usepackage{xcolor}
\usepackage{xspace}
\usepackage{url}
\usepackage{aodv}
\usepackage{amsfonts,amsmath,amssymb}
\usepackage{algorithm,algorithmic}
\usepackage{wrapfig}
\usepackage{enumerate}
\usepackage{etex,etoolbox} 

\DeclareFontFamily{U}{mathx}{\hyphenchar\font45}
\DeclareFontShape{U}{mathx}{m}{n}{
      <5> <6> <7> <8> <9> <10>
      <10.95> <12> <14.4> <17.28> <20.74> <24.88>
      mathx10
      }{}
\DeclareSymbolFont{mathx}{U}{mathx}{m}{n}
\DeclareMathAccent{\wideparen}{0}{mathx}{"75}

\makeatletter
\spn@wtheorem{observation}{Observation}{\bfseries}{\itshape}
\def\theHALC@line{\thealgorithm-\theALC@line}
\def\theHALC@rem{\thealgorithm-\theALC@rem}
\makeatother
\floatname{algorithm}{Process}

\DeclareSymbolFont{frenchscript}{OMS}{ztmcm}{m}{n}
\DeclareMathSymbol{\T}{\mathord}{frenchscript}{84}   
\DeclareMathSymbol{\pow}{\mathord}{frenchscript}{80} 
\newcommand{\journalonly}[1]{}
\makeatletter
\def\comesfrom{\@transition\leftarrowfill}
\def\goesto{\@transition\rightarrowfill}
\def\ngoesto{\@transition\nrightarrowfill}
\def\Goesto{\@transition\Rightarrowfill}
\def\nGoesto{\@transition\nRightarrowfill}
\def\xmapsto{\@transition\mapstofill}
\def\nxmapsto{\@transition\nmapstofill}
\def\@transition#1{\@@transition{#1}}
\newbox\@transbox
\newbox\@arrowbox
\newbox\@downbox
\def\@@transition#1#2%
   {\setbox\@transbox\hbox
      {\vrule height 1.5ex depth .9ex width 0ex\hskip0.25em$\scriptstyle#2$\hskip0.25em}
   \ifdim\wd\@transbox<1.5em
      \setbox\@transbox\hbox to 1.5em{\hfil\box\@transbox\hfil}\fi
   \setbox\@arrowbox\hbox to \wd\@transbox{#1}
   \ht\@arrowbox\z@\dp\@arrowbox\z@
   \setbox\@transbox\hbox{$\mathop{\box\@arrowbox}\limits^{\box\@transbox}$}
   \dp\@transbox\z@\ht\@transbox 10pt
   \mathrel{\box\@transbox}}
\def\nrightarrowfill{$\m@th\mathord-\mkern-6mu%
  \cleaders\hbox{$\mkern-2mu\mathord-\mkern-2mu$}\hfill
  \mkern-6mu\mathord\not\mkern-2mu\mathord\rightarrow$}
\def\Rightarrowfill{$\m@th\mathord=\mkern-6mu%
  \cleaders\hbox{$\mkern-2mu\mathord=\mkern-2mu$}\hfill
  \mkern-6mu\mathord\Rightarrow$}
\def\nRightarrowfill{$\m@th\mathord=\mkern-6mu%
  \cleaders\hbox{$\mkern-2mu\mathord=\mkern-2mu$}\hfill
  \mkern-6mu\mathord\not\mathord\Rightarrow$}
\def\mapstofill{$\m@th\mathord\mapstochar\mathord-\mkern-6mu%
  \cleaders\hbox{$\mkern-2mu\mathord-\mkern-2mu$}\hfill
  \mkern-6mu\mathord\rightarrow$}
\def\nmapstofill{$\m@th\mathord\mapstochar\mathord-\mkern-6mu%
  \cleaders\hbox{$\mkern-2mu\mathord-\mkern-2mu$}\hfill
  \mkern-6mu\mathord\not\mkern-2mu\mathord\rightarrow$}
\makeatother 
\newcommand{\ar}[1]{\mathrel{\goesto{#1}}}            
\newcommand{\nar}[1]{\mathrel{\ngoesto{#1\;}}}        
\newcommand{\dar}[1]{\mathrel{\Goesto{\raisebox{.08em}{\scriptsize$#1$}}}}
\makeatletter
\def\moverlay{\mathpalette\mov@rlay}
\def\mov@rlay#1#2{\leavevmode\vtop{%
   \baselineskip\z@skip \lineskiplimit-\maxdimen
   \ialign{\hfil$\m@th#1##$\hfil\cr#2\crcr}}}
\newcommand{\charfusion}[3][\mathord]{
    #1{\ifx#1\mathop\vphantom{#2}\fi
        \mathpalette\mov@rlay{#2\cr#3}
      }
    \ifx#1\mathop\expandafter\displaylimits\fi}
\makeatother
\newcommand{\dcup}{\charfusion[\mathbin]{\cup}{\mbox{\Large$\cdot$}}}
\renewcommand{\comma}{,}
\newcommand{\tawn}{T-AWN\xspace}
\newcommand{\awn}{AWN\xspace}
\newcommand{\otawn}{(T-)AWN\xspace}
\newcommand{\bis}{\raisebox{.3ex}{$\underline{\makebox[.7em]{$\leftrightarrow$}}$}}
\def\squareforqed{\hbox{\rlap{$\sqcap$}$\sqcup$}}
\def\endbox{\ifmmode\squareforqed\else{\unskip\nobreak\hfil
\penalty50\hskip1em\null\nobreak\hfil\squareforqed
\parfillskip=0pt\finalhyphendemerits=0\endgraf}\fi}

\newcommand{\highlightspec}[1]{{\color{red}#1}}
\newcommand{\highlightUPD}[1]{\STATE{\textcolor{red}{\ensuremath{\mbox{\bf [\![}#1\mbox{\bf ]\!]}}}}}%
\newcommand{\plat}[1]{\raisebox{0pt}[0pt][0pt]{#1}} 
\newcommand{\spaces}[1]{\ #1\ }
\newcommand{\ans}{\spaces{\wedge}}
\newcommand{\ors}{\spaces{\vee}}
\newcommand{\ims}{\spaces{\Rightarrow}}

\newcommand{\now}{\keyw{now}\xspace}

\newcommand{\NN}{
    \ensuremath{%
        \mathop{\rm I\mkern-2.5mu N}%
        \nolimits%
    }%
}
\newcommand{\B}{\mathrel{\mathcal{B}}} 
\newcommand{\Sim}{\mathrel{\mathcal{S}}} 
\newcommand{\U}{\mathrel{\mathcal{U}}} 

\newcommand{\hopsc}{\dval{hops}_c}
\newcommand{\dipc}{\dval{dip}_{\hspace{-1pt}c}}
\newcommand{\ripc}{\dval{rip}_{\hspace{-1pt}c}}
\newcommand{\oipc}{\dval{oip}_{\hspace{-1pt}c}}

\newcommand{\dsnc}{\dval{dsn}_c}
\newcommand{\rsnc}{\dval{rsn}_c}
\newcommand{\destsc}{\dval{dests}_c}
\newcommand{\osnc}{\dval{osn}_c}
\newcommand{\ipc}{\dval{ip}_{\hspace{-1pt}c}}
\newcommand{\xiN}[2][N]{\xi_{#1}^{#2}}

\newcommand{\rtord}[1][\dval{dip}]{\ensuremath{\sqsubseteq_{#1}}}
\newcommand{\rtequiv}[1][\dval{dip}]{\ensuremath{\approx_{#1}}}
\newcommand{\rtsord}[1][\dval{dip}]{\ensuremath{\sqsubset_{#1}}}

\newcommand{\rte}{routing table entry\xspace}
\newcommand{\rtes}{routing table entries\xspace}

\newcommand{\W}{\mathcal{W}}
\newcommand{\R}{\mathcal{R}}

\newcommand{\nil}{\mathbf{0}}
\newcommand{\RW}{\R\mathop{:}\W}
\newcommand{\Rw}[2][\mathop]{R#1{:}#2}

\newcommand{\exclnow}[1][\xi]{#1_{\backslash\now}}

\makeatletter
\providecommand{\@thirdoffive}[5]{#3}

\def\fixstatement#1{%
  \AtEndEnvironment{#1}{%
    \xdef\pat@label{\expandafter\expandafter\expandafter
      \@thirdoffive\csname#1\endcsname\space\@currentlabel}}}
      
\globtoksblk\prooftoks{1000}
\newcounter{proofcount}

\newboolean{qedatend}
\newcommand{\prf}{\setboolean{qedatend}{true}\begin{proof}}
\newcommand{\prfnobox}{\setboolean{qedatend}{false}\begin{proof}}
\newcommand{\eprf}{\ifqedatend\qed\fi\end{proof}}

\long\def\prfatend#1\eprf{%
  \edef\next{\noexpand\begin{theopargself}\noexpand\prf[\noexpand\hspace{-1.5pt}of \pat@label\noexpand\hspace{1.75pt}]}%
  \toks\numexpr\prooftoks+\value{proofcount}\relax=\expandafter{\next#1\eprf\end{theopargself}}
  \stepcounter{proofcount}}
  
  \long\def\prfnoboxatend#1\eprf{%
  \edef\next{\noexpand\begin{theopargself}\noexpand\prfnobox[\noexpand\hspace{-1.5pt}of \pat@label\noexpand\hspace{1.5pt}]}%
  \toks\numexpr\prooftoks+\value{proofcount}\relax=\expandafter{\next#1\eprf\end{theopargself}}
  \stepcounter{proofcount}}

\long\def\toApp#1\etoApp{%
  \edef\next{\noexpand\begin{theopargself}}%
  \toks\numexpr\prooftoks+\value{proofcount}\relax=\expandafter{\next#1\end{theopargself}}
  \stepcounter{proofcount}}

\def\printproofs{%
  \count@=\z@
  \loop
    \the\toks\numexpr\prooftoks+\count@\relax
    \ifnum\count@<\value{proofcount}%
    \advance\count@\@ne
  \repeat}
\makeatother

\fixstatement{lemma}
\fixstatement{theorem}
\fixstatement{corollary}
\fixstatement{proposition}
\fixstatement{emp}

\spnewtheorem*{proofsketch}{Proof Sketch}{\itshape}{\rm}
\spnewtheorem*{toAppendix}{}{}{\rm}

\newcounter{applemmacount}
\newcounter{apptheoremcount}
\newcounter{appdefcount}
\newcommand{\highlight}[2][red]{\textcolor{#1}{#2}}
\newcommand{\deltime}{\keyw{DELETE\_PERIOD}}
\newcommand{\valtime}{\keyw{ACTIVE\_ROUTE\_TIMEOUT}}
\newcommand{\traversal}{\keyw{NET\_TRAVERSAL\_TIME}}
\newcommand{\nodetraversal}{\keyw{NODE\_TRAVERSAL\_TIME}}
\newcommand{\retries}{\keyw{RREQ\_RETRIES}}
\newcommand{\pathdiscoverytime}{\keyw{PATH\_DISCOVERY\_TIME}}
\newcommand{\myroute}{\keyw{MY\_ROUTE\_TIMEOUT}}
\newcommand{\fnexprt}{\keyw{exp\_rt}}
\newcommand{\fnexpq}{\keyw{exp\_store}}
\newcommand{\exprt}[3]{\fnexprt(#1\comma #2\comma #3)}
\newcommand{\expq}[3]{\fnexpq(#1\comma #2)}
\newcommand{\fnsettimert}{\keyw{setTime\_rt}}
\newcommand{\settimert}[3]{\fnsettimert(#1\comma #2\comma #3)}
\newcommand{\fnsettimequeues}{\keyw{setTime\_store}}
\newcommand{\settimequeues}[3]{\fnsettimequeues(#1\comma #2\comma #3)}
\newcommand{\fnsetretries}{\keyw{resetRetries}}
\newcommand{\setretries}[2]{\fnsetretries(#1\comma #2)}
\newcommand{\fnincretries}{\keyw{incRetries}}
\newcommand{\incretries}[2]{\fnincretries(#1\comma #2)}
\newcommand{\fnexprreqs}{\keyw{exp\_rreqs}}
\newcommand{\exprreqs}[2]{\fnexprreqs(#1\comma #2)}
\newcommand{\lifetime}{\keyw{ltime}}
\newcommand{\tTIME}{\keyw{TIME}}

\newcommand\tablelabel[2][]{%
  \ifstrempty{#1}{%
    \hypertarget{#2}{\mbox{\color{gray}\scriptsize({\sf #2})}}%
  }{%
    \hypertarget{#1}{\mbox{\color{gray}\scriptsize({\sf #2})}}%
  }%
}
\newcommand\tableref[2][]{%
  \ifstrempty{#1}{%
    \hyperlink{#2}{\mbox{({\sf #2})}}\xspace%
  }{%
    \hyperlink{#1}{\mbox{({\sf #2})}}\xspace%
  }%
}

\newcounter{Hequation}

\makeatletter
\g@addto@macro\equation{\stepcounter{Hequation}}
\makeatother

\begin{document}%
\titlerunning{A Timed Process Algebra for Wireless Networks}
\title{A Timed Process Algebra for Wireless Networks with an Application in Routing\texorpdfstring{\thanks{An
    extended abstract of this paper---everything but the appendices---appeared as~\cite{ESOP16}.}}{}}
\authorrunning{E.\ Bres, R.J.\ van Glabbeek and P.\ H\"ofner}
  \author{
  Emile Bres\inst{1,3}\and
  Rob van Glabbeek\inst{1,2}\and
  Peter H\"ofner\inst{1,2} 
  }
\institute{
  NICTA, Australia\and
  Computer Science and Engineering, University of New South Wales, Australia\and
  \'Ecole Polytechnique, Paris, France
}

\maketitle 
\setcounter{footnote}{0}
\begin{abstract}
This paper proposes a timed process algebra for wireless networks,
an extension of the Algebra for Wireless Networks. It combines
treatments of local broadcast, conditional unicast and data structures,
which are essential features for the modelling of network protocols.
In this framework we model and analyse the Ad hoc On-Demand Distance
Vector routing protocol, and show that, contrary to claims in the
literature, it fails to be loop free. We also present boundary
conditions for a fix ensuring
that the resulting protocol is indeed loop free.

\end{abstract}
\section{Introduction}\label{sec:intro} 
In 2011 we developed the \emph{Algebra for Wireless Networks} (\awn)~\cite{ESOP12},
a process algebra particularly tailored for Wireless Mesh Networks (WMNs) and Mobile Ad Hoc Networks (MANETs).
Such networks are currently being used in a wide range of application areas, such as public safety and mining. 
They are self-organising wireless multi-hop networks that provide network communication without relying on a wired backhaul infrastructure.
A significant characteristic of such networks is that they allow highly dynamic network topologies, meaning that network nodes can 
join, leave, or move within the network at any moment. 
As a consequence routing protocols have constantly to check for broken links, and to replace invalid routes by better ones. 

To capture the typical characteristics of WMNs and MANETs, \awn offers
a unique set of features: 
\emph{conditional unicast} (a message transmission attempt with different follow-up behaviour depending on its success), 
\emph{groupcast} (communication to a specific set of nodes),
\emph{local broadcast} (messages are received only by nodes within transmission range of the sender), 
and \emph{data structure}. 
We are not aware of any other process algebra that provides all these features, and hence
could not use any other algebra to model certain protocols for WMNs or MANETs in a straightforward fashion.%
\footnote{A comparison between \awn and other process algebras can be found in \cite[Sect.~11]{TR13}.}
Case studies~\cite{ESOP12,TR13,GHPT16,EHWripe12} have shown that \awn provides the right level of abstraction to model 
full IETF protocols, such as the Ad hoc On-Demand Distance Vector {(AODV)} routing protocol~\cite{rfc3561}.
\awn has been employed to formally model this
  protocol---thereby eliminating ambiguities\pagebreak[3] and contradictions from the official specification, written in English Prose---and
 to reason about protocol behaviour and provide rigorous proofs of key protocol
properties such as loop freedom and route correctness.

However, \awn\ abstracts from time. 
Analysing routing protocols without considering timing issues is useful in its own right; for AODV it has revealed many shortcomings in drafts as well as in the standard  (e.g.,~\cite{BOG02,MSWIM12,AODVloop}).
Including time in a formal analysis, however, will pave the way to analyse protocols that repeat some procedures every couple of time units; 
examples are OLSR~\cite{rfc3626} and B.A.T.M.A.N.~\cite{batman}.
Even for a reactive protocol such as AODV, which does not schedule tasks regularly, 
it has been shown that timing aspects are important: 
if timing parameters are chosen poorly, 
some routes are not established since data that is stored locally at network nodes 
expires too soon and is erased~\cite{CK05}.
Besides such shortcomings in ``performance'', also fundamental
  correctness properties like loop freedom can be affected by the treatment of time---as we will
  illustrate.

To enable time analyses of WMNs and MANETs, 
this paper proposes a \emph{Timed (process) Algebra for Wireless Networks} (\tawn),
an extension of \awn. It combines \awn's unique set of 
features, such as local broadcast, with~time. 

In this framework we model and analyse the AODV routing protocol, and show that, contrary to claims in the
literature, e.g., \cite{AODV99}, it fails to be loop free, as data required for routing can expire. 
We also present boundary conditions for a fix ensuring
that the resulting protocol is loop free.

\subsection*{Design Decisions}
Prior to the development of \tawn\ we had to make a couple of decisions.

\paragraph{Intranode computations.}
In wireless networks sending a packet from one node to another takes multiple microseconds.
Compared to these ``slow'' actions, time spent for internal (intranode) computations, such as variable assignments or evaluations of expressions,
is negligible. We therefore postulate that only transmissions from one node to another take time.

This decision is debatable for processes that can perform infinite
sequences of intranode computations without ever performing a
durational action.  In this paper (and in all applications), we
restrict ourselves to \emph{well-timed} processes in the spirit of
\cite{NS94}, i.e., to processes where any infinite sequence of actions
contains infinitely many time steps or infinitely many input actions,
such as receiving an incoming packet.

But, in the same spirit as \tawn assigns time to internode communications, 
it is more or less straightforward to assign times to other operations as well.

\paragraph{Guaranteed Message Receipt and Input Enabledness.}
A fundamental assumption underlying the semantics of \otawn is that any broadcast message \emph{is}
received by all nodes within transmission range \cite[\textsection1]{TR13}.%
\footnote{\label{guaranteed receipt}In reality, communication is only
half-duplex: a single-interface network node cannot receive messages while sending and hence
messages can be lost.
However, the CSMA protocol used at the link layer---not 
modelled by {\otawn}---keeps the probability of packet loss due to two nodes (within
range) sending at the same time rather low.
\journalonly{Since we are examining imperfect protocols, we first of all want to
establish how they behave under optimal conditions. For this reason we
abstract from probabilistic reasoning by assuming no message loss at
all, rather than working with a lossy broadcast formalism that offers
no guarantees that any message will ever arrive.}
}
This abstraction enables us to interpret a failure of route discovery (as documented for AODV in
\cite[\textsection9]{TR13}) as an imperfection in the protocol, rather than as a result of a chosen formalism not
ensuring guaranteed receipt. 

A consequence of this design decision is that in the operational semantics of \otawn a broadcast of
one node in a network needs to synchronise with some (in)activity of all other nodes in the network~\cite[\textsection11]{TR13}.
If another node is within transmission range of the broadcast, the broadcast synchronises with a
receive action of that node, and otherwise with a non-arrive transition, which signals that the node
is out of range for this broadcast~\cite[\textsection4.3]{TR13}.

A further consequence is that we need to specify our nodes in such a way that they are
\emph{input-enabled}, meaning that in any state they are able to receive messages from any other node
within transmission range.
\journalonly{\outP{For allowing a node to be in a state where it cannot receive a new
message of another node within transmission range, in combination with the requirement of guaranteed
receipt, leads to \emph{blocking}, the situation where a node is (temporarily) unable to to perform
a broadcast due to the fact that some node within transmission range is not ready to receive it~\cite[\textsection4.5]{TR13}.
This is a non-realistic model of WMNs.} \comsP{move backwards to the model of AODV or skip for conference version}}

Since a transmission (broadcast, groupcast, or unicast) takes 
multiple units of time, we postulate 
that another node can only receive a message if it remains
within transmission range during the whole period of sending.%
\footnote{To be precise, we forgive very short
  interruptions in the connection between two nodes---those that begin
  and end within the same unit of time.} 
A possible way to model the receive action that synchronises with a
transmission such as a broadcast is to let it take
the same amount of time as the broadcast action. However, a process that is busy
executing a durational receive action would fail to be input-enabled, for it would not be able to start
receiving another message before the ongoing message receipt is finished.
For this reason, we
model the receipt of a message as an instantaneous action that synchronises with the very end of a broadcast action.%
\footnote{Another solution would be to assume that a broadcast-receiving process can receive multiple messages in parallel.
In case the process is meant to add incoming messages to a message queue (as happens in our application
to AODV), one can assume that a message that is being received in parallel is added to that queue as
soon as its receipt is complete. However, such a model is equivalent to one in which only the very
last stage of the receipt action is modelled.}

\paragraph{\tawn Syntax.} 
When designing or formalising a protocol in \tawn, an engineer should
not be bothered with timing aspects; except for functions and procedures that schedule tasks depending on the current time. 
Because of this, we use the syntax of \awn also for \tawn;
``extended'' by a local timer \now.
Hence we can perform a 
timed analysis of any specification written in \awn, since they are
also \tawn specifications.

\section{A Timed Process Algebra for Wireless Networks}
\label{sec:process_algebra}
In this section we propose \tawn (Timed Algebra for Wireless Networks), 
an extension of the process algebra AWN \cite{ESOP12,TR13} with time.
AWN itself is a variant of standard process algebras \cite{Mi89,Ho85,BK86,BB87}, 
tailored to protocols in wireless mesh networks, such as the 
Ad-hoc on Demand Distance Vector (AODV) routing protocol.
In \otawn, a WMN is modelled as an encapsulated
parallel composition of network nodes.  On each node several
sequential processes may be running in parallel.  Network nodes
communicate with their direct neighbours---those nodes that are in
transmission range---using either broadcast, groupcast or unicast.  Our
formalism maintains for each node the set of nodes that are currently
in transmission range.  Due to mobility of nodes and variability of
wireless links, nodes can move in or out of transmission range.  The
encapsulation of the entire network inhibits communications between
network nodes and the outside world, with the exception of the receipt
and delivery of data packets from or to clients\footnote{The
application layer that initiates packet sending and/or awaits receipt of
a packet.}  of the modelled protocol that may be hooked up to various
nodes.

In \tawn we apply a discrete model of time, where each sequential process maintains a local variable
\now holding its local clock value---an integer. We employ only one clock for each sequential
process. All sequential processes in a network synchronise in taking time steps, and at each time
step all local clocks advance by one unit.  For the rest, the variable \now behaves as
any other variable maintained by a process: its value can be read when evaluating guards, thereby
making progress time-dependant, and any value can be assigned to it, thereby resetting the local
clock.\journalonly{\comsP{shall we discuss clock drift?}}

In our model of a sequential process $p$ running on a node, time can elapse only when $p$ is
transmitting a message to another node, or when $p$ currently has no way to proceed---for instance,
when waiting on input, or for its local clock to reach a specified value. All other actions of $p$, such as assigning values to variables,
evaluating guards, communicating with other processes running on the same node, or communicating
with clients of the modelled protocol hooked up at that node, are assumed to be an
order of magnitude faster, and in our model take no time at all. Thus they are executed in preference to time steps.

\subsection{The Syntax of \tawn}\label{ssec:syntax}

The syntax of \tawn is the same as the syntax of AWN \cite{ESOP12,TR13}, except for the presence of
the variable \now of the new type \keyw{TIME}.
This brings the advantage that any specification written in AWN can be interpreted and 
analysed in a timed setting. The rest of this Section
\ref{ssec:syntax} is almost copied verbatim from the original articles about AWN \cite{ESOP12,TR13}.

\subsubsection{A Language for Sequential Processes.}

The internal state of a process is determined, in part, by the values
of certain data variables that are maintained by that process.  To
this end, we assume a data structure with several types, variables
ranging over these types, operators and predicates. First order
predicate logic yields terms (or \emph{data expressions}) and formulas
to denote data values and statements about them.\label{pg:undefvalues}\footnote{As
    \label{fn:undefvalues}%
    operators we also allow \emph{partial} functions with the
    convention that any atomic formula containing an undefined subterm
    evaluates to {\tt false}.} Our data structure
always contains the types \keyw{TIME}, \tDATA, \tMSG, {\tIP} and $\pow(\tIP)$ of \emph{time values}, which
we take to be integers (together with the special value $\infty$),
\emph{application layer data}, \emph{messages}, \emph{IP addresses}---or any
other node identifiers---and \emph{sets of IP addresses}. We further assume that 
there is a variable \now of type \keyw{TIME} and a function $\newpktID:\tDATA \times \tIP \rightarrow \tMSG$
that generates a message with new application layer data for a particular destination. 
The purpose of this function is to inject data into the protocol; details will be given later.

In addition, we assume a type \keyw{SPROC} of \emph{sequential processes},
and a collection of \emph{process names}, each being an operator of type
$\keyw{TYPE}_1 \times \cdots \times \keyw{TYPE}_n \rightarrow \keyw{SPROC}$
for certain data types $\keyw{TYPE}_i$.
Each process name $X$ comes with a \emph{defining equation}\vspace{-1pt}
\[X(\keyw{var}_1\comma\ldots\comma\keyw{var}_n) \stackrel{{\it def}}{=} p\ ,\]
\vspace{-1pt}in which, for each $i=1,\ldots,n$, $\keyw{var}_i$ is a variable of type
$\keyw{TYPE}_i$ and $p$ a \emph{guarded}\footnote{\label{guarded}%
An expression $p$ is \emph{guarded} if each call of a process
name $X(\dexp{exp}_1,\ldots,\dexp{exp}_n)$  occurs with a subexpression
$\cond{\varphi}q$, $\assignment{\keyw{var}:=\dexp{exp}}q$,
$\alpha.q$ or $\unicast{\dexp{dest}}{\dexp{ms}}.q \prio r$ of $p$.}
\emph{sequential process expression}
defined by the grammar below. The expression $p$ may contain the variables
$\keyw{var}_i$ as well as $X$; however, all occurrences of data
variables in $\p$ have to be \emph{bound}.%
\journalonly{\footnote{An occurrence of a
  data variable in $p$ is \emph{bound} if it is one of the variables
  $\keyw{var}_i$, a variable {\msg} occurring in a subexpression
  $\receive{\msg}.\q$, a variable \keyw{var} occurring in a subexpression
  $\assignment{\keyw{var}:=\dexp{exp}}\q$, or an occurrence in a
  subexpression $\cond{\varphi}\q$ of a variable occurring free in $\varphi$.
  Here $\q$ is an arbitrary sequential process expression.}}
The choice of the underlying data structure and the process names with
their defining equations can be tailored to any particular application
of our language; our decisions made for modelling AODV are presented
in Section~\ref{sec:aodv}.  The process
names are used to denote the processes that feature in this
application, with their arguments $\keyw{var}_i$ binding the current
values of the data variables maintained by these processes.

The \emph{sequential process expressions} are given by the following grammar:
\begin{align*}
\SP\ ::=\ &X(\dexp{exp}_1,\ldots,\dexp{exp}_n) ~\mid~ \cond{\varphi}\SP ~\mid~ \assignment{\keyw{var}:=\dexp{exp}} \SP~\mid~ \SP+\SP ~\mid\\
     &\alpha.\SP ~\mid~ \unicast{\dexp{dest}}{\dexp{ms}}.\SP \prio \SP \\
\alpha\ ::=\ &
  \broadcastP{\dexp{ms}} ~\mid~ \groupcastP{\dexp{dests}}{\dexp{ms}}
  ~\mid~ \send{\dexp{ms}} ~\mid\\
  &\deliver{\dexp{data}} ~\mid~\receive{\msg}
\end{align*}
Here $X$ is a process name, $\dexp{exp}_i$ a data expression of the
same type as $\keyw{var}_i$, $\varphi$ a data formula,
$\keyw{var}\mathop{:=}\dexp{exp}$ an assignment of a data expression
\dexp{exp} to a variable \keyw{var} of the same type, \dexp{dest},
\dexp{dests}, \dexp{data} and \dexp{ms} data expressions of types
{\tIP}, $\pow(\tIP)$, {\tDATA} and {\tMSG}, respectively, and $\msg$ a
data variable of type \tMSG.

The internal state of a sequential process described by an expression
$\p$ in this language is determined by $\p$, together with a
\emph{valuation} $\xi$ associating data values $\xi(\keyw{var})$ to
the data variables \keyw{var} maintained by this process.
Valuations naturally extend to \emph{$\xi$-closed} data
expressions---those in which all variables are either bound or in the
domain of $\xi$.

Given a valuation of the data variables by concrete data values, the
sequential process $\cond{\varphi}\p$ acts as $\p$ if $\varphi$
evaluates to {\tt true}, and deadlocks if $\varphi$ evaluates to
{\tt false}. In case $\varphi$ contains free variables that are not
yet interpreted as data values, values are assigned to these variables
in any way that satisfies $\varphi$, if possible.
The sequential process $\assignment{\keyw{var}:=\dexp{exp}}\p$
acts as $\p$, but under an updated valuation of the data variable \keyw{var}.
The sequential process $\p+\q$ may act either as $\p$ or as
$\q$, depending on which of the two processes is able to act at all.  In a
context where both are able to act, it is not specified how the choice 
is made. \pagebreak The sequential process $\alpha.\p$ first performs the action
$\alpha$ and subsequently acts as $\p$.  The action
$\broadcastP{\dexp{ms}}$ broadcasts (the data value bound to the
expression) $\dexp{ms}$ to the other network nodes within transmission range,
whereas $\unicast{\dexp{dest}}{\dexp{ms}}.\p \prio \q$ is a sequential process
that tries to unicast the message $\dexp{ms}$ to the destination
\dexp{dest}; if successful it continues to act as $\p$ and otherwise
as $\q$. In other words, $\unicast{\dexp{dest}}{\dexp{ms}}.\p$ is
prioritised over $\q$; only if the action $\unicast{\dexp{dest}}{\dexp{ms}}$
is not possible, the alternative $q$ will happen. It models an abstraction
of an acknowledgment-of-receipt mechanism that is typical for unicast
communication but absent in broadcast communication, as implemented by
the link layer of relevant wireless standards such as IEEE 802.11~\cite{IEEE80211s}.
The process $\groupcastP{\dexp{dests}}{\dexp{ms}}.\p$ tries
to transmit \dexp{ms} to all destinations $\dexp{dests}$, and proceeds
as $\p$ regardless of whether any of the transmissions is successful.
Unlike {\bf unicast} and  {\bf broadcast}, the expression {\bf groupcast} 
does not have a unique counterpart in networking.
Depending on the protocol and the implementation it can 
be an iterative unicast, a broadcast, or a multicast;
thus  {\bf groupcast} abstracts from implementation details. The
action $\send{\dexp{ms}}$ synchronously transmits a message to another
process running on the same network node; this action can occur only
when this other sequential process is able to receive the message.
The sequential process $\receive{\msg}.\p$ receives any message $m$
(a data value of type \tMSG) either from another node, from another
sequential process running on the same node or from the client hooked
up to the local node.  It then proceeds as $\p$, but with the data
variable $\msg$ bound to the value $m$. 
The submission of data from a client
is modelled by the receipt of a message $\newpkt{\dval{d}}{\dval{dip}}$,
where the function $\newpktID$ generates a message containing the
data $\dval{d}$ and the intended destination $\dval{dip}$. 
Data is delivered to the client by \deliver{\dexp{data}}.%

\subsubsection{A Language for Parallel Processes.}

\emph{Parallel process expressions} are given by the grammar\vspace{-6pt}
$$\PP ~::=~ \xi,\SP ~\mid~ \PP \parl \PP\ ,$$
where $\SP$ is a sequential process expression and $\xi$ a valuation.
An expression $\xi,\p$ denotes a sequential process expression
equipped with a valuation of the variables it maintains.
The process $P\parl Q$ is a parallel composition of $P$ and $Q$,
running on the same network node. An action $\receive{\dval{m}}$ of~$P$
synchronises with an action $\send{\dval{m}}$ of $Q$ into an internal
action $\tau$, as formalised in Table~\ref{tab:sos parallel}. 
These receive actions of $P$ and send actions of $Q$
cannot happen separately. All other actions of $P$ and $Q$, except time steps, including
receive actions of $Q$ and send actions of $P$, occur interleaved in $P\parl Q$.
\journalonly{Thus, in an expression $(P\parl Q)\parl R$, for example,
the send and receive actions of $Q$ can communicate only with $P$ and
$R$, respectively, but the receive actions of $R$, as well as the send
actions of $P$, remain available for communication with the environment.}%
Therefore, a parallel process expression denotes a parallel composition
of sequential processes $\xi,P$ with information flowing from right to
left. The variables of different sequential processes running on the
same node are maintained separately, and thus cannot be shared.

Though $\parl$ only allows information flow in one direction, it reflects reality of 
WMNs. Usually two sequential processes run on the same node:
$
P \parl Q
$.
The main process $P$ deals with all protocol details  of the node, e.g., message handling 
and maintaining the data such as routing tables.
The process $Q$ manages the queueing of messages as they arrive; it is always able to
receive a message even if $P$ is busy. 
The use of message queueing in combination with $\parl$ is crucial in order to create input-enabled
nodes (cf.~\Sect{intro}).\pagebreak

\subsubsection{A Language for Networks.}\label{ssec:networks}

We model network nodes in the context of a wireless mesh network by
\emph{node expressions} of the form $\dval{ip}:\PP:R$. Here $\dval{ip}
\in \tIP$ is the \emph{address} of the node, $\PP$ is a parallel process
expression, and $R\subseteq \tIP$ is the \emph{range} of the node---the
set of nodes that are currently within transmission range of $\dval{ip}$.

A \emph{partial network} is then modelled by a \emph{parallel
  composition} $\|$ of node expressions, one for every node in the
network, and a \emph{complete network} is a partial network within an
\emph{encapsulation operator} $[\_]$ that limits the communication of
network nodes and the outside world to the receipt and the delivery
of data packets to and from the application layer attached to the
modelled protocol in the network nodes.
This yields the following grammar for network expressions:
\[N ::= [M] \qquad\qquad M::= ~~ \dval{ip}:\PP:R ~~\mid~~ M \| M\ .\vspace{-3pt}\]


\subsection{The Semantics of \tawn}\label{sec:tawnsos}

\newcommand{\tick}{{\rm tick}}
\newcommand{\timestep}{{\rm ts}}
\newcommand{\tr}{{\rm t}}
\newcommand{\w}{{\rm w}}
\newcommand{\wtr}{{\rm wr}}
\newcommand{\ws}{{\rm ws}}
\newcommand{\wrs}{{\rm wrs}}
\newcommand{\timeinc}{[\now\mathord+\mathord{\mkern-1mu}\mathord+]}
\newcommand{\defined}{\mathord\downarrow}
\newcommand{\range}{\dval{dsts}\xspace}

\newcommand{\LB}{\keyw{LB}\xspace}
\newcommand{\LG}{\keyw{LG}\xspace}
\newcommand{\LU}{\keyw{LU}\xspace}
\newcommand{\DB}{\keyw{\Delta B}\xspace}
\newcommand{\DG}{\keyw{\Delta G}\xspace}
\newcommand{\DU}{\keyw{\Delta U}\xspace}
\newcommand{\UB}{\chP{\keyw{UB}\xspace}}
\newcommand{\UG}{\chP{\keyw{UG}\xspace}}
\newcommand{\UU}{\chP{\keyw{UU}\xspace}}

As mentioned in the introduction, the transmission of a message 
takes time.  Since our main application assumes wireless links and node mobility, 
the packet delivery time varies. Hence we assume a minimum time that is required 
to send a message, as well as an optional extra transmission time. 
In \tawn the values of these parameters are given for each type of sending separately: 
\LB, \LG, and \LU, satisfying $\LB,\LG,\LU > 0$,
specify the minimum bound, in units of time, on the duration of a broadcast, groupcast
and unicast transmission; the optional additional transmission times are denoted by 
$\DB$, $\DG$ and $\DU$, satisfying $\DB,\DG,\DU \geq 0$.
Adding up these parameters (e.g.\ $\LB$ and $\DB$) yields maximum
transmission times. We allow any execution consistent with these parameters.
For all other actions our processes can take we postulate execution times of 0.

\begin{table}[tp]
\caption{Structural operational semantics for sequential process expressions
\label{tab:sos sequential}}

\vspace{-5ex}
$$\begin{array}{@{}l@{\ }r@{~}l@{\hspace{-3em}}r@{}}
\tablelabel{bc} &
\xi,\broadcastP{\dexp{ms}}.\p &\ar{\tau} \xi,\stargcast{\tIP}{\xi(\dexp{ms})}[\LB,\DB].\p \prio \p
  \hspace{-7em}
& \mbox{\small(if $\xi(\dexp{ms})\defined$)}
\\[3pt]
\tablelabel{gc} &
\xi,\groupcastP{\dexp{dests}}{\dexp{ms}}.\p &\ar{\tau} \xi,\stargcast{\xi(\dexp{dests})}{\xi(\dexp{ms})}[\LG,\DG].\p \prio \p
\hspace{-10em}
 \\&&& \hspace{-10em}\mbox{\small(if $\xi(\dexp{dests})\defined$ and $\xi(\dexp{ms})\defined$)}
\\[3pt]
\tablelabel{uc} &
  \xi,\unicast{\dexp{dest}}{\dexp{ms}}.\p \prio \q &\ar{\tau} \xi,\stargcast{\{\xi(\dexp{dest})\}}{\xi(\dexp{ms})}[\LU,\DU].\p \prio \q
\hspace{-10em}
\\
&
&& \mbox{\small(if $\xi(\dexp{dest})\defined$ and $\xi(\dexp{ms})\defined$)}
\\[3pt]
\multicolumn{2}{l}{\tablelabel{tr}\,
  \xi,\hspace{-.5pt}\stargcast{\range}{m}[n\mathord+1,o].\p \prio \q}
  &\ar{\Rw[\mathbin]{\w}}
  \xi\timeinc,\stargcast{(\range\mathop\cap R)}{m}[n,o].\p \prio \q
\hspace{-10em}
\\
&
&& \hspace{-20em}\mbox{\small($\forall R\subseteq \tIP$)}
\\[3pt]
\tablelabel{tr-o}&\\
\multicolumn{2}{@{}r@{\ }}{
  \xi,\stargcast{\range}{m}[n\mathord+1,o\mathord+1].\p \prio \q}
  &
  \ar{\Rw[\mathbin]{\w}}
  \xi\timeinc,\stargcast{(\range\mathop\cap R)}{m}[n\mathord+1,o].\p \prio \q 
\hspace{-7em}
\\
&
&& \hspace{-20em}\mbox{\small($\forall R\subseteq \tIP$)}
\\[3pt]
\tablelabel{sc} &
  \xi,\stargcast{\range}{m}[0,o].\p \prio \q & \ar{\colonact{\range}{\starcastP{m}}} \xi,\p
& \hspace{-3pt}\mbox{\small(if $\range\neq \emptyset$)}
\\[3pt]
\tablelabel[nsc]{$\neg$sc} &
    \xi,\stargcast{\range}{m}[0,o].\p \prio \q & \ar{\colonact{\range}{\starcastP{m}}} \xi,\q
  & \hspace{-3pt}\mbox{\small(if $\range=\emptyset$)}
\\[3pt]
\tablelabel{snd} &
  \xi,\send{\dexp{ms}}.\p &\ar{\send{\xi(\dexp{ms})}} \xi,\p
& \mbox{\small(if $\xi(\dexp{ms})\defined$)}
\\[3pt]
\tablelabel{ws}&
  \xi,\send{\dexp{ms}}.\p &\ar{\ws~} \xi\timeinc,\send{\dexp{ms}}.\p
& \mbox{\small(if $\xi(\dexp{ms})\defined$)}
\\[3pt]
\tablelabel{del}&
  \xi,\deliver{\dexp{data}}.\p &\ar{\deliver{\xi(\dexp{data})}} \xi,\p
& \mbox{\small(if $\xi(\dexp{data})\defined$)}
\\[3pt]
\tablelabel{rcv}&
  \xi,\receive{\keyw{msg}}.\p &\ar{\receive{m}} \xi[\keyw{msg}:=m],\p
  & \mbox{\small($\forall m\in \tMSG$)}
\\[3pt]
\tablelabel{wr}&
  \xi,\receive{\keyw{msg}}.\p &\ar{\wtr~} \xi\timeinc,\receive{\keyw{msg}}.\p
\\[3pt]
\tablelabel{ass}&
  \xi,\assignment{\keyw{var}:=\dexp{exp}}\p &\ar{\tau} \xi[\keyw{var}:=\xi(\dexp{exp})],\p
& \mbox{\small(if $\xi(\dexp{exp})\defined$)}
\\[3pt]
\tablelabel{w}&
  \xi,\p & \ar{\w}  \xi\timeinc,p  & \mbox{\small(if $\xi(p)\mathord\uparrow$)}
\\[6pt]
\tablelabel{rec}&
\multicolumn{2}{c}{\displaystyle
  \frac{\emptyset[\keyw{var}_i:=\xi(\dexp{exp}_i)]_{i=1}^n,\p \ar{a} \xii,\p'}
  {\xi,X(\dexp{exp}_1,\ldots,\dexp{exp}_n) \ar{a} \xii,\p'}
  ~\mbox{(\small$X(\keyw{var}_1,\ldots,\keyw{var}_n) \stackrel{{\it def}}{=} \p$)}}
   \hspace{-10em}
\\[-8pt]
&
  && \hspace{-10em}
  \mbox{\small($\forall a\mathbin\in \act-\W$, if $\xi(\dexp{exp}_i)\defined$)}
\\[6pt]
\tablelabel{rec-w}\hspace{-2.5pt}&
\multicolumn{2}{c}{\displaystyle
  \frac{\emptyset[\keyw{var}_i:=\xi(\dexp{exp}_i)]_{i=1}^n,\p \ar{w_1} \xii,\p'}
  {\xi,X(\dexp{exp}_1,...,\dexp{exp}_n) \ar{w_1} \xi\timeinc,X(\dexp{exp}_1,...,\dexp{exp}_n)}
  ~\mbox{(\small$X(\keyw{var}_1,...,\keyw{var}_n) \mathop{\stackrel{{\it def}}{=}} \p$)}\hspace{-10em}}
\\[-8pt]
&
  && \hspace{-8em}\mbox{$(\forall w_1\mathord\in \W\mbox{, if~}\xi(\dexp{exp}_i)\defined)$}
\\[6pt]
\multicolumn{3}{@{}l@{}}{
\tablelabel{grd}\hspace{3mm} 
{\displaystyle\
  \frac{ \xi \stackrel{\varphi}{\rightarrow}\xii}
       {\xi,\cond{\varphi}\p \ar{\tau} \xii,\p} \qquad\qquad
\tablelabel[ngrd]{$\neg$grd}\hspace{3mm} 
  \frac{ \xi \nar{\varphi}}
       {\xi,\cond{\varphi}\p \ar{\w} \xi\timeinc,\cond{\varphi}\p}
       \hspace{-20em}
  }}
\\[14pt]
\multicolumn{3}{@{}l@{}}{
\tablelabel{alt-l}\hspace{3mm} 
{\displaystyle
  \frac{\xi,\p \ar{a} \xii,\p'}{\xi,\p+\q \ar{a} \xii,\p'} \qquad\qquad
\tablelabel{alt-r}\hspace{3mm}  
  \frac{\xi,\q \ar{a} \xii,\q'}{\xi,\p+\q \ar{a} \xii,\q'}
         
  }}\hspace{-15em}
  & \mbox{\small($\forall a\in \act-\W$)}
\\[14pt]
\tablelabel{alt-w}\hspace{-2.5mm}&
\multicolumn{2}{c}{\displaystyle
  \frac{\xi,\p \ar{w_1} \xii,\p' \quad \xi,\q \ar{w_2} \xii,\q'}{\xi,\p+\q \ar{w_1\wedge w_2} \xii,\p'+\q'}
  }
  & \hspace{-5em}\mbox{\small$(\forall w_1,w_2\in \W)$}
\end{array}$$
\end{table}

\subsubsection{Sequential Processes.}

The structural operational semantics of \tawn, given in
Tables~\ref{tab:sos sequential}--\ref{tab:sos network}, is in the style of Plotkin \cite{Pl04}
and describes how one internal state can evolve into another by performing an \emph{action}.

A difference with AWN is that some of the transitions are time steps.
On the level of node and network expressions they are labelled ``$\tick$''
and the parallel composition of multiple nodes can perform such a transition iff each of those
nodes can---see the third rule in Table~\ref{tab:sos network}.
On the level of sequential and parallel process expressions, time-consuming
transitions are labelled with \emph{wait actions} from $\W = \{\w,\ws,\wtr,\wrs\}\subseteq\act$ and
\emph{transmission actions} from $\RW=\{\Rw{w_1} \mid w_1 \mathbin\in \W \wedge R \subseteq \tIP\}\subseteq\act$.
Wait actions $w_1\mathbin\in\W$ indicate that the system is waiting, possibly only as long as it
fails to synchronise on a \textbf{receive} action ($\wtr$), a
\begin{wrapfigure}[6]{r}{0.26\textwidth}\vspace{-4ex}
\centering
\small $\begin{array}{c|cccc}
\wedge& \w   & \wtr & \ws  & \wrs \\
\hline
 \w   & \w   & \wtr & \ws  & \wrs \\
 \wtr & \wtr & \wtr & \wrs & \wrs \\
 \ws  & \ws  & \wrs & \ws  & \wrs \\
 \wrs & \wrs & \wrs & \wrs & \wrs
\end{array}$
\end{wrapfigure}%
 \textbf{send} action ($\ws$) or both of those ($\wrs$); actions $\Rw{w_1}$ indicate that the system
is transmitting a message
while the current transmission range of the node is $R \subseteq \tIP$.
In the operational rule for choice (+) we combine any two
wait actions $w_1,w_2\in\W$ with the operator $\wedge$, which 
joins the conditions under which these wait actions can occur. 

In Table~\ref{tab:sos sequential}, which gives the semantics of sequential process expressions,
a state is given as a pair $\xi,p$ of a sequential process expression $p$ and a valuation $\xi$ of
the data variables maintained by $p$.
The set $\act$ of actions that can be executed by sequential and parallel process expressions, and
thus occurs as transition labels, consists of
$\colonact{R}{\starcastP{m}}$,
$\send{m}$,
$\deliver{\dval{d}}$,
$\receive{m}$,
durational actions $w_1$ and $\Rw{w_1}$,
and internal actions~$\tau$, for each choice of $R\subseteq \tIP$, $m \mathop{\in}\tMSG$, $\dval{d}\mathop{\in}\tDATA$ and $w_1\in\W$.
Here $\colonact{R}{\starcastP{m}}$ is the action of transmitting the message $m$, to be
received by the set of nodes $R$, which is the intersection of the set of intended
destinations with the nodes that are within transmission range throughout the transmission.
We do not distinguish whether this message has been broadcast, groupcast or unicast---the
differences show up merely in the value of $R$.

In Table~\ref{tab:sos sequential} $\xi[\keyw{var}:= v]$ denotes the
valuation that assigns the value $v$ to the variable \keyw{var}, and
agrees with $\xi$ on all other variables.
We use $\xi\timeinc$ as an abbreviation for $\xi[\now:=\xi(\now)\mathord+1]$,
the valuation $\xi$ in which the variable \now is incremented by 1. This describes the state
of data\linebreak[3] variables after 1 unit of time elapses, while no other changes in data occurred.
The empty valuation $\emptyset$ assigns values to no variables. Hence
$\emptyset[\keyw{var}_i:=v_i]_{i=1}^n$ is the valuation that
\emph{only} assigns the values $v_i$ to the variables $\keyw{var}_i$ for $i=1,\ldots,n$.
Moreover, $\xi(\dval{exp})\defined$, with \dval{exp} a data expression, is the
statement that $\xi(\dval{exp})$ is defined; this might fail because \dexp{exp} contains a variable
that is not in the domain of $\xi$ or because \dexp{exp} contains a partial function that is given
an argument for which it is not defined.

A state $\xi,r$ is \emph{unvalued}, denoted by $\xi(r)\mathord\uparrow$, if $r$ has the form
$\broadcastP{\dexp{ms}}.\p$,
$\groupcastP{\dexp{dests}}{\dexp{ms}}.\p$, \hfill
$\unicast{\dexp{dest}}{\dexp{ms}}.\p$, \hfill
$\send{\dexp{ms}}.\p$, \hfill
$\deliver{\dexp{data}}.\p$, \linebreak[3]
$\assignment{\keyw{var}:=\dexp{exp}}\p\,$ or
$\,X(\dexp{exp}_1,\ldots,\dexp{exp}_n)$ with
either $\xi(\dexp{ms})$ or $\xi(\dexp{dests})$ or $\xi(\dexp{dest})$ or $\xi(\dexp{data})$ or
$\xi(\dexp{exp})$ or some $\xi(\dexp{exp}_i)$ undefined.
From such a state no progress is possible.
However, Rule \tableref{w} in Table~\ref{tab:sos sequential} does allow time to progress.
We use $\xi(r)\defined$ to denote that a state is not unvalued.

Rule \tableref{rec} for process names in Table~\ref{tab:sos sequential} is
motivated and explained in \cite[\textsection4.1]{TR13}.
The variant \tableref{rec-w} of this rule for wait actions $w_1\in\W$ has been modified such that the recursion
is not yet unfolded while waiting. 
  \journalonly{
  \newcommand{\XP}[1][\ensuremath{n}]{\keyw{X}(#1)}%
  \newcommand{\YP}[1][\ensuremath{n}]{\keyw{Y}(#1)}%
  This is necessary to keep a (non-deterministic) choice 
  in one single data state available if no action is enabled at the moment. 
  For example, the process $\XP[n\mathop+1] + \YP[n\mathop+2]$, with 
  \raisebox{0pt}[10pt]{\XP$\stackrel{{\it def}}{=}$[$\now\mathop> 42$]\XP} and 
  \raisebox{0pt}[10pt]{\YP$\stackrel{{\it def}}{=}$[$\now\mathop< 42$]\YP},
  would have 
  either to choose for either \XP\ or \YP\ in case \now equals $42$ (if the second last line \comP{refer to names} of 
  Table~\ref{tab:sos sequential} would be modified),
  or would lead the same data state (last line).
  The former choice could lead to a deadlock. 
   If the choice is postponed, neither the deadlock occurs nor the problem with multiple data states arises.
    Moreover, the rule also}
This simulates the behaviour of AWN where a process is only 
    unwound if the first action of the process can be performed.

In the subsequent rules \tableref{grd} and \tableref[ngrd]{$\neg$grd} for variable-binding guards $[\varphi]$, the notation
\plat{$\xi\stackrel{\varphi}{\rightarrow}\xii$} says that
$\xii$ is an extension of $\xi$ that satisfies $\varphi$: a valuation that agrees with
$\xi$ on all variables on which $\xi$ is defined, and valuates the
other variables occurring free in $\varphi$, such that the formula
$\varphi$ holds under $\xii$. All variables not free in $\varphi$ and
not evaluated by $\xi$ are also not evaluated by $\xii$.
Its negation \plat{$\xi\nar{\varphi}$} says that no such extension exists, and thus that $\varphi$
is false in the current state, no matter how we interpret the variables whose values are still undefined.
If that is the case, the process $[\varphi]p$ will idle by performing the action $\w$ (of waiting)
without changing its state, except that the variable \now will be incremented.

\journalonly{%
\begin{example}\cite[\textsection4.1]{TR13}~
Let $\xi(\keyw{numa})=7$ and $\xi(\keyw{numb})$, $\xi(\keyw{numc})$
be undefined, with \keyw{numa}, \keyw{numb} and \keyw{numc} data variables
of type $\NN$. Then the sequential process given by the pair
$\xi,[\keyw{numa} = \keyw{numb}+\keyw{numc}]p$ admits several transitions of the form
$$\xi,[\keyw{numa} = \keyw{numb}+\keyw{numc}]p \ar{\tau} \xii,p$$
such as the one with $\xii(\keyw{numb})=2$ and $\xii(\keyw{numc})=5$.
On the other hand, since $\keyw{numb}\in\NN$, the process $\xi,[\keyw{numa} = \keyw{numb}+8]p$ only admits the transition
$$\xi,[\keyw{numa} = \keyw{numb}+8]p \ar{\w} \xi[\now\mathord+\mathord+],[\keyw{numa} = \keyw{numb}+8]p\;.$$
\end{example}
}

\begin{example}
The process $\assignment{\keyw{timeout}:=\now+2} [\now = \keyw{timeout}] p$ first
sets the variable \keyw{timeout} to 2 units after the current time. Then it encounters a guard that
evaluates to {\tt false}, and therefore takes a $\w$-transition, twice. After two time units,
the guard evaluates to {\tt true} and the process proceeds as $p$.
\pagebreak
\end{example}

The process $\receive{\msg}.p$ can receive any message $m$ from the environment in which this
process is running. As long as the environment does not provide a message, this process will wait.
This is indicated by the transition labelled $\wtr$ in Table~\ref{tab:sos sequential}.
The difference between a $\wtr$-and a $\w$-transition is that the former can be taken only when the
environment does not synchronise with the $\textbf{receive}$-transition. In our semantics any state with
an outgoing $\wtr$-transition also has an outgoing $\textbf{receive}$-transition (see \Thm{SP-classification}), which
conceptually has priority over the $\wtr$-transition. Likewise the transition labelled $\ws$ is only enabled in
states that also admit a $\textbf{send}$-transition, and is taken only in a context where the
$\textbf{send}$-transition cannot be taken.

Rules \tableref{alt-l} and \tableref{alt-r}, defining the behaviour of the choice operator for non-wait actions 
are standard. Rule \tableref{alt-w} for wait actions says that a process $p+q$ can wait only if both $p$ and $q$
can wait; if one of the two arguments can make real progress, the choice process $p+q$ always chooses
this progress over waiting.  This is a direct generalisation of the law $p+\nil=p$ of CCS \cite{Mi89}.
As a consequence, a condition on the possibility of $\p$ or $\q$ to wait is inherited by $\p+\q$.
This gives rise to the transition label $\wrs$, that makes waiting conditional on the environment
failing to synchronising with a $\textbf{receive}$ as well as a $\textbf{send}$-transition.
In understanding the target $\xii,\p'{+}\q'$ of this rule, it is helpful to realise that whenever
$\xi,\p \ar{w_1} \xii,\q$, then $\q = \p$ and $\xii = \xi[\now\mathord+\mathord+]$;
see \Prop{time determinism}.

In order to give semantics to the transmission constructs (broadcast, groupcast, unicast),
the language of sequential processes is extended with the auxiliary construct\vspace{-1ex}
$$\stargcast{\dval{\range}}{\dval{m}}[n,o].\SP \prio \SP\ ,$$ 
with $m\in\tMSG$, $n,o\in\NN$ and $\range\subseteq \tIP$.
This is a variant of the \textbf{broadcast}-, \textbf{groupcast}- and \textbf{unicast}-constructs,
describing intermediate states of the transmission of message \dval{m}.
The argument \range of \textbf{*cast} denotes those intended
destinations that were not out of transmission range during the part of the transmission that
already took place.

In a state $\stargcast{\range}{\dval{m}}[n,o].p \prio q$ with $n>0$ the transmission
still needs between $n$ and $n\mathord+o$ time units to complete.
If $n=0$ the actual \textbf{*cast}-transition will take place; resulting in state $p$ if
the message is delivered to at least one node in the network (\range is non-empty), and $q$ otherwise.

Rule \tableref{gc} says that once a process commits to a
\textbf{groupcast}-transmission, it is going to behave as
$\colonact{\range}{\textbf{*cast}}(m)[n,o]$
with time parameters $n:=\LG$ and $o:=\DG$. The transmitted message $m$ is calculated
by evaluating the argument \dexp{ms}, and the  transmission range \range of
this \textbf{*cast} is initialised by evaluating the argument \dexp{dests}, indicating the intended
destinations of the \textbf{groupcast}. 
Rules \tableref{bc} and \tableref{uc} for \textbf{broadcast} and \textbf{unicast} are the same, except that in the case of
\textbf{broadcast} the intended destinations are given by the set $\tIP$ of \emph{all} possible destinations,
whereas a \textbf{unicast} has only one intended destination.
Moreover, only \textbf{unicast} exploits the difference in the continuation process depending on
whether an intended destination is within transmission range.
Subsequently, Rules \tableref{tr} and \tableref{tr-o} come into force; they allow time-consuming transmission steps to take
place, each decrementing one of the time parameters $n$ or $o$. Each time step of a transmission corresponds
to a transition labelled $\Rw{\w}$, where $R$ records the current transmission range. Since
sequential processes store no information on transmission ranges---this information is added only
when moving from process expressions to node expressions---at this stage of the description all
possibilities for the transmission range need to be left open, and hence there is a transition
labelled $\Rw{\w}$ for each choice of $R$.\footnote{Similar to $\receive{\msg}.p$ having a
transition for each possible incoming message~$m$.} When transitions for process expressions are
inherited by node expressions, only one of the transitions labelled $\Rw{\w}$ is going to survive, namely the one
where $R$ equals the transmission range given by the node expression (cf.~Rule~\tableref{n-t} in \Tab{sos node}).
Upon doing a transition $\Rw\w$, the range \range of the \textbf{*cast} is restricted to $R$.
As soon as $n\mathbin=0$, regardless of the value of $o$, the transmission is completed by the execution of
the action $\colonact{\range}{\textbf{*cast}(m)}$ (Rules \tableref{sc} and \tableref[nsc]{$\neg$sc}). Here the actual message
$m$ is passed on for synchronisation with \textbf{receive}-transitions of all nodes $\dval{ip}\in\range$.

This treatment of message transmission is somewhat different from the one in AWN\@.
There, the rule $\xi,\groupcastP{\dexp{dests}}{\dexp{ms}}.\p \ar{\groupcastP{\xi(\dexp{dests})}{\xi(\dexp{ms})}}
\xi,\p$ describes the behaviour of the \textbf{groupcast} construct for sequential processes, and the rule
\vspace{-4ex} 

{\small$$  \frac{P \ar{\groupcastP{D}{m}} P'}
  {\rule[10pt]{0pt}{1pt}
   \dval{ip}\mathbin:P\mathbin:R \ar{\colonact{R\cap D}{\starcastP{m}}} \dval{ip}\mathbin:P'\mathbin:R}$$}%
lifts this behaviour from processes to nodes. In this last stage the \textbf{groupcast}-action is
unified with the \textbf{broadcast}- and \textbf{unicast}-action into a \textbf{*cast},
at which occasion the range of the \textbf{*cast} is calculated as the intersection of the intended
destinations $D$ of the \textbf{groupcast} and the ones in transmission range $R$. In \tawn, on the other hand,
the conversion of \textbf{groupcast} to \textbf{*cast} happens already at the level of sequential processes.

{
\renewcommand{\parl}{\mathop{\ensuremath{\!\langle\!\langle\!}}}
\begin{table}[t]
\caption{Structural operational semantics for parallel process expressions
\label{tab:sos parallel}}

\vspace*{-5ex}
$$\begin{array}{@{}c@{}}
\tablelabel{p-al}\hspace{1mm} 
\displaystyle
  \frac{P \ar{a} P'}{P\parl Q \ar{a} P'\parl Q}%
  \,\mbox{\small$\left(\begin{array}{@{}c@{}}\forall a\neq \receive{m},\\a \mathop{\not\in} \W, a \mathop{\not\in} \RW\\
  \end{array}\right)$}
  \quad\hspace{-1pt}
\tablelabel{p-ar}\hspace{1mm} 
  \frac{Q \ar{a} Q'}{P\parl Q \ar{a} P\parl Q'}%
\,\mbox{\small$\left(\begin{array}{@{}c@{}}\forall a\neq \send{m},\\a \mathop{\not\in} \W, a \mathop{\not\in} \RW\end{array}\right)$}
\\[15pt]
\tablelabel{p-a}\hspace{1mm} 
\displaystyle
  \frac{P \ar{\receive{m}} P'\quad Q \ar{\send{m}} Q'}
       {P\parl Q \ar{\tau} P'\parl Q'}
    ~\,\mbox{\small($\forall m\mathbin\in\tMSG$)}
\qquad
\tablelabel{p-w}\hspace{1mm} 
  \frac{P \ar{w_1} P'\quad Q \ar{w_2} Q'}
       {P\parl Q \ar{w_3} P'\parl Q'}
\\[15pt]
\tablelabel{p-tl}\hspace{1mm} 
\displaystyle
\frac{P\!\ar{R\mathbin{:}w_1} P'\quad\! Q\!\ar{w_2} Q'}
	       {P\parl Q \ar{R\mathbin{:}w_3} P'\hspace{-1pt}\parl Q'}
\quad
\tablelabel{p-tr}\hspace{1mm} 
\frac{P\! \ar{w_1} P'\quad\! Q\! \ar{R\mathbin{:}w_2} Q'}
       {P\parl Q \ar{R\mathbin{:}w_3} P'\parl Q'}
\quad
\tablelabel{p-t}\hspace{1mm} 
\frac{P\! \ar{R\mathbin{:}w_1} P'\quad\! Q\! \ar{R\mathbin{:}w_2} Q'}
       {P\parl Q \ar{R\mathbin{:}w_3} P'\parl Q'}
\\[12pt]
\hfill\!(\forall w_{1},w_{2},w_{3}\in\W, w_3=w_1\parl w_2)
\end{array}$$
\vspace*{-5.5ex}
\end{table}
}

\vspace*{-0ex}
\subsubsection{Parallel Processes.}
Rules \tableref{p-al}, \tableref{p-ar} and \tableref{p-a} of Table \ref{tab:sos parallel} are taken from AWN, and formalise
the description of the operator $\parl$ given in \SSect{syntax}.
Rule \tableref{p-w} stipulates under which conditions a process $P\parl Q$
can do a wait action, and 
\begin{wrapfigure}[5]{r}{0.25\textwidth}
\vspace{-5ex}
\label{parallel}
\hypertarget{parallel}{
\centering
\small $\begin{array}{c|cccc}
\parl
      & \w   & \wtr & \ws  & \wrs \\
\hline
 \w   & \w   & \wtr & \w   & \wtr \\
 \wtr & \w   & \wtr &  -   &  -   \\
 \ws  & \ws  & \wrs & \ws  & \wrs \\
 \wrs & \ws  & \wrs &  -   &  -
\end{array}$}
\end{wrapfigure}%
of which kind. 
Here $\parl$ is also
a partial binary function on the set $\W$, specified by
the table on the right.
The process $P\parl Q$ can do a wait action only if both $P$ and $Q$ can do so.
In case $P$ can do a $\wtr$ or a $\wrs$-action, $P$ can also do a \textbf{receive}
and in case $Q$ can do a $\ws$ or a $\wrs$, $Q$ can  \pagebreak also do a \textbf{send}.
When both these possibilities  apply, the \textbf{receive} of $P$
synchronises with the \textbf{send} of $Q$ into a $\tau$-step, which
has priority over waiting. 
In the other 12 cases no synchronisation between $P$ and $Q$ is
possible, and we do obtain a wait action.
Since a \textbf{receive}-action of $P$ that does not synchronise with $Q$
is dropped, so is the corresponding side condition of a wait action of 
$P$. 
Hence (within the remaining 12 cases) a $\wtr$ of $P$ is treated as a $\w$,
and a $\wrs$ as a $\ws$. Likewise a $\ws$ of $Q$ is treated as a $\w$,
and a $\wrs$ as a $\wtr$. This leaves 4 cases to be decided.
In all four, we have $\w_1 \parl \w_2 = \w_1 \wedge \w_2$.

Time steps $\Rw{w_1}$ are treated exactly like wait actions from $\W$
(cf.\ Rules \tableref{p-tl}, \tableref{p-tr} and \tableref{p-t}).
If for instance $P$ can do a $\Rw{\w}$, meaning that it spends a unit of time
on a transmission, while $Q$ can do a $\wtr$, meaning that it waits a unit
of time only when it does not receive anything from another source,
the result is that $P\parl Q$ can spend a unit of time transmitting something,
but only as long as $P\parl Q$ does not receive any message; if it does,
the receive action of $Q$ happens with priority over the wait action of $Q$,
and thus occurs before $P$ spends a unit of time transmitting.

\subsubsection{Node and Network Expressions.}
The operational semantics of node and network expressions of
Tables~\ref{tab:sos node} and~\ref{tab:sos network} uses transition labels
\tick,
$\colonact{R}{\starcastP{m}}$,
$\colonact{H\neg K}{\listen{m}}$,
$\colonact{\dval{ip}}{\deliver{\dval{d}}}$,
$\textbf{connect}(\dval{ip},\dval{ip}')$,
$\textbf{disconnect}(\dval{ip},\dval{ip}')$,
$\tau$
and $\colonact{\dval{ip}}{\textbf{newpkt}(\dval{d},\dval{dip})}$.
As before, $m\in\tMSG$, $d\in\tDATA$, $R\subseteq \tIP$, and $\dval{ip},\dval{ip}'\in\tIP$.
Moreover, $H,K\subseteq\tIP$ are sets of IP addresses.

The actions $\colonact{R}{\starcastP{m}}$ are inherited by nodes from the processes that
run on these nodes (cf.~Rule \tableref{n-sc}).
The action $\colonact{H\neg K}{\listen{m}}$ states that the message $m$
simultaneously arrives at all addresses $\dval{ip}\mathbin\in H$, and fails to
arrive at all addresses $\dval{ip}\mathbin\in K$.
The rules of Table~\ref{tab:sos network} let a $\colonact{R}{\starcastP{m}}$-action of one node
synchronise with an $\listen{m}$ of all other nodes, where this
$\listen{m}$ amalgamates the arrival of message $m$ at the nodes in
the transmission range $R$ of the $\starcastP{m}$, and the non-arrival at the
other nodes. Rules \tableref{n-rcv} and \tableref{n-dis}
state that arrival of a message at a node happens if and only if the
node receives it, whereas non-arrival can happen at any time.
This embodies our assumption that, at any time, any message that is
transmitted to a node within range of the sender is actually received by that node.
(Rule \tableref{n-dis}\journalonly{, having no
  premises,} may appear to say that any node \dval{ip} has the option to
  disregard any message at any time. However, the encapsulation
  operator (below) prunes away all such disregard transitions that do
  not synchronise with a cast action for which \dval{ip} is out of range.)

The action $\send{m}$ of a process does not give rise to any action of
the corresponding node---this action of a sequential process cannot
occur without communicating with a receive action of another sequential
process running on the same node.
Time-consuming actions $w_1$ and $\Rw{w_1}$, with $w_1\in\W$, of a process are renamed into $\tick$ on
the level of node expressions.\footnote{Rule \tableref{n-t} ensures that only those
  $R\mathbin{:}w_1$-transitions survive for which $R$ is the current transmission range of the node.}
 All we need to remember of these actions is that they take one unit
of time. Since on node expressions the actions $\send{m}$ have been dropped, 
\pagebreak the side condition
making the wait actions $\ws$ and $\wrs$ conditional on the absence of a \textbf{send}-action can be
dropped as well. The priority of \textbf{receive}-actions over the wait action $\wtr$ can now also be
dropped, for in the absence of \textbf{send}-actions, \textbf{receive}-actions are entirely
reactive.  A node can do a \textbf{receive}-action only when another node, or the application layer,
casts a message, and in
this case that other node is not available to synchronise with a $\tick$-transition.

\begin{table}[t]
\caption{Structural operational semantics for node expressions
\label{tab:sos node}}

\vspace{-5ex}
$$\begin{array}{@{}c@{\quad}c@{}}
\tablelabel{n-sc}\hspace{1mm} 
\displaystyle
  \frac{P \ar{\colonact{\range}{\starcastP{m}}} P'}
  {\rule[13pt]{0pt}{1pt}
   \dval{ip}\mathbin{:}P\mathord{:}R \ar{\colonact{\range}{\starcastP{m}}} \dval{ip}\mathbin{:}P'\mathord{:}R}
&
\tablelabel{n-rcv}\hspace{1mm} 
\displaystyle
  \frac{P \ar{\receive{m}} P'}
  {\rule[13pt]{0pt}{1pt}
   \dval{ip}\mathbin{:}P\mathord{:}R \ar{\colonact{\{\dval{ip}\}\neg\emptyset}{\listen{m}}} \dval{ip}\mathbin{:}P'\mathord{:}R}
\\[18pt]
\tablelabel{n-del}\hspace{1mm} 
\displaystyle
  \frac{P \ar{\deliver{\dval{d}}} P'}
  {\rule[13pt]{0pt}{1pt}
   \dval{ip}\mathbin{:}P\mathord{:}R \ar{\colonact{\dval{ip}}{\deliver{\dval{d}}}} \dval{ip}\mathbin{:}P'\mathord{:}R}
&
\tablelabel{n-dis}\hspace{1mm} 
 \dval{ip}\mathbin{:}P\mathord{:}R \ar{\colonact{\emptyset\neg\{\dval{ip}\}}{\listen{m}}} \dval{ip}\mathbin{:}P\mathord{:}R
\\[16pt]\multicolumn{2}{@{}c@{}}{
\tablelabel[n-t]{n-$\tau$}\hspace{1mm} 
\displaystyle
  \frac{P \ar{\tau} P'}
  {\dval{ip}\mathbin{:}P\mathord{:}R \ar{\tau} \dval{ip}\mathbin{:}P'\mathord{:}R}
\quad
\tablelabel{n-w}\hspace{1mm} 
\displaystyle
  \frac{P \ar{w_1} P'}
  {\rule[13pt]{0pt}{1pt}
   \dval{ip}\mathbin{:}P\mathord{:}R \ar{\tick} \dval{ip}\mathbin{:}P'\mathord{:}R}
\quad
\tablelabel{n-t}\hspace{1mm} 
\displaystyle
  \frac{P \ar{R\mathord{:}w_1} P'}
  {\rule[13pt]{0pt}{1pt}
   \dval{ip}\mathbin{:}P\mathord{:}R \ar{\tick} \dval{ip}\mathbin{:}P'\mathord{:}R}
}
\\[16pt]
\multicolumn{2}{@{}r@{}}{
\mbox{\footnotesize $\left(\!\forall w_1 \mathbin\in \W \right)$}
}
\\[5pt]
\tablelabel{con}\hspace{1mm} 
\displaystyle
  \dval{ip}\!\mathbin{:}\!P\!\mathbin{:}\!R \ar{\textbf{connect}(\dval{ip},\dval{ip}')} \dval{ip}\!\mathbin{:}\!P\!\mathbin{:}\!R\mathop\cup\{\dval{ip}'\}
  \hspace{-2pt}
&
\tablelabel{dis}\hspace{1mm} 
  \dval{ip}\!\mathbin{:}\!P\!\mathbin{:}\!R \ar{\textbf{disconnect}(\dval{ip},\dval{ip}')} \dval{ip}\!\mathbin{:}\!P\!\mathbin{:}\!R\mathop-\{\dval{ip}'\}
\end{array}$$
\end{table}

\begin{table}[t]
\caption{Structural operational semantics for network expressions
\label{tab:sos network}}

\vspace{-5ex}
{\small
\[\begin{array}{@{}c@{}}
\tablelabel{nw-tl/nw-tr}\hspace{2mm} 
\displaystyle
  \frac{M \ar{\colonact{R}{\starcastP{m}}} M' \quad N \ar{\colonact{H\neg K}{\listen{m}}} N'}
  {\rule[11pt]{0pt}{1pt}
   M \| N \ar{\colonact{R}{\starcastP{m}}} M' \| N'\qquad
   N \| M \ar{\colonact{R}{\starcastP{m}}} N' \| M'}
   \qquad\mbox{\footnotesize
  $\left(\begin{array}{@{}c@{}}H\subseteq R,\\K {\cap} R = \emptyset\end{array}\right)$}
\\[21pt]
\tablelabel{arr}\hspace{2mm} 
  \displaystyle
  \frac{M \ar{\colonact{H\neg K}{\listen{m}}} M' \quad N \ar{\colonact{H'\neg K'}{\listen{m}}} N'}
  {\rule[11pt]{0pt}{1pt}M \| N \ar{\colonact{(H\cup H')\neg(K\cup K')}{\listen{m}}} M' \| N'}
  \qquad\!
  \tablelabel{tck}\hspace{2mm} 
  \frac{M \ar{\tick} M' \quad N \ar{\tick} N'}
  {\rule[11pt]{0pt}{1pt}M \| N \ar{\tick} M' \| N'}
\\[16pt]
  \displaystyle
  \tablelabel{nw-al}\hspace{2mm} 
  \frac{M \ar{a} M'}{M \| N \ar{a} M' \| N}
  \qquad
  \tablelabel{nw-ar}\hspace{2mm} 
  \frac{N \ar{a} N'}{M \| N \ar{a} M \| N'}
  \qquad
  \tablelabel{e-a}\hspace{2mm} 
  \frac{M \ar{a} M'}{
   [M] \ar{a} [M']}
\\[14pt]
~\hfill\!\!
\mbox{\footnotesize $
(\forall a\mathbin\in\{
                        \colonact{\dval{ip}}{\deliver{\dval{d}}},\tau,     
                                     \textbf{connect}(\dval{ip},\dval{ip}'),
                                      \textbf{disconnect}(\dval{ip},\dval{ip}')
\})$}
\\[8pt]
  \displaystyle
  \tablelabel{e-tck}\hspace{1mm} 
  \frac{M \ar{\tick} M'}{[M] \ar{\tick} [M']}
  \quad\!
  \tablelabel{e-sc}\hspace{1mm} 
  \frac{M \mathbin{\ar{\colonact{R}{\starcastP{m}}}} M'}{[M] \ar{\tau} [M']}
  \quad\!
  \tablelabel{e-np}\hspace{1mm} 
  \frac{M \ar{\colonact{\{\dval{ip}\}\neg K}{\listen{\newpkt{\dval{d}}{\dval{dip}}}}} M'}
  {\rule[11pt]{0pt}{1pt}
   [M] \ar{\colonact{\dval{ip}}{\textbf{newpkt}(\dval{d},\dval{dip})}} [M']}
\end{array}\]}
\vspace*{-5ex}
\end{table}

Internal actions $\tau$ and the action $\colonact{\dval{ip}}{\deliver{\dval{d}}}$
are simply inherited by node expressions from the processes that run
on these nodes (Rules \tableref[n-t]{n-$\tau$} and \tableref{n-del}), and are interleaved in the parallel composition of
nodes that makes up a network. Finally, we allow actions
$\textbf{connect}(\dval{ip},\dval{ip}')$ and
$\textbf{disconnect}(\dval{ip},\dval{ip}')$ for
$\dval{ip},\dval{ip}'\mathop\in \tIP$ modelling a change in network
topology.  \journalonly{These actions can be thought of as occurring nondeterministically, or as
actions instigated by the environment of the modelled network protocol.}%
In this formalisation node $\dval{ip}'$ may be in the range of node
$\dval{ip}$, meaning that $\dval{ip}$ can send to $\dval{ip}'$,
even when the reverse does not hold. For some applications, in
particular the one to AODV in \Sect{aodv}, it is useful to assume that
$\dval{ip}'$ is in the range of $\dval{ip}$ if and only if $\dval{ip}$
is in the range of $\dval{ip}'$. 
This symmetry can be enforced by adding the following rules to Table~\ref{tab:sos node}:\pagebreak[3]

\noindent
{\small
\[\begin{array}{@{}c@{\qquad}c@{}}
\displaystyle
  \dval{ip}\!:\!P\!:\!R \ar{\textbf{connect}(\dval{ip}',\dval{ip})} \dval{ip}\!:\!P\!:\!R\cup\{\dval{ip}'\}
&
  \dval{ip}\!:\!P\!:\!R \ar{\textbf{disconnect}(\dval{ip}',\dval{ip})} \dval{ip}\!:\!P\!:\!R-\{\dval{ip}'\}
\\[11pt]\displaystyle
  \frac{\dval{ip} \not\in \{\dval{ip}'\!,\dval{ip}''\}}
  {\rule[11pt]{0pt}{1pt}
   \dval{ip}\!:\!P\!:\!R \ar{\textbf{connect}(\dval{ip}'\!,\dval{ip}'')} \dval{ip}\!:\!P\!:\!R}
&\displaystyle
  \frac{\dval{ip} \not\in \{\dval{ip}'\!,\dval{ip}''\}}
  {\rule[11pt]{0pt}{1pt}
   \dval{ip}\!:\!P\!:\!R \ar{\textbf{disconnect}(\dval{ip}'\!,\dval{ip}'')} \dval{ip}\!:\!P\!:\!R}
\end{array}\]}%
and replacing the rules in the third line of \Tab{sos network} for (dis)connect actions by\vspace{-1ex}
{\small
\[
  \frac{M \ar{a} M' \quad
        N \ar{a} N'}
  {M \| N \ar{a} M' \| N'}
  \qquad
  \frac{M \ar{a} M'}{
   [M] \ar{a} [M']}
  \qquad\mbox{\footnotesize
  $
  \left(\forall a\in\left\{\begin{array}{c}\textbf{connect}(\dval{ip},\dval{ip}'),\\
                                           \textbf{disconnect}(\dval{ip},\dval{ip}')
                    \end{array}\right\}\right).$}
\]}%
The main purpose of the encapsulation operator is to ensure that no
messages will be received that have never been sent. In a parallel
composition of network nodes, any action $\receive{\dval{m}}$ of one
of the nodes \dval{ip} manifests itself as an action $\colonact{H\neg
K}{\listen{\dval{m}}}$ of the parallel composition, with $\dval{ip}\in
H$. Such actions can happen (even) if within the parallel composition
they do not communicate with an action $\starcastP{\dval{m}}$ of
another component, because they might communicate with a
$\starcastP{\dval{m}}$ of a node that is yet to be added to the
parallel composition. However, once all nodes of the network are
accounted for, we need to inhibit unmatched arrive actions,
as otherwise our formalism would allow any node at any time to receive
any message. One exception however are those arrive actions
that stem from an action $\receive{\newpkt{\dval{d}}{\dval{dip}}}$ of a
sequential process running on a node, as those actions represent
communication with the environment. Here, we use the function
$\newpktID$, which we assumed to exist.\footnote{%
\newcommand{\npkt}{\textbf{newpkt}\xspace}%
To avoid the function $\newpktID$ we could have introduced a new
primitive \npkt, which is dual to \textbf{deliver}.}
It models the injection of new data $\dval{d}$ for destination $\dval{\dip}$.

The encapsulation operator passes through internal actions, as well as
delivery of data to destination nodes, this being an interaction
with the outside world (Rule \tableref{e-a}). $\starcastP{m}$-actions are declared internal
actions at this level (Rule \tableref{e-sc}); they cannot be steered by the outside world.
The connect and disconnect actions are passed through in
Table~\ref{tab:sos network} (Rule \tableref{e-a}), thereby placing them under control of the
environment; to make them nondeterministic, their rules should have a
$\tau$-label in the conclusion, or alternatively
$\textbf{connect}(\dval{ip},\dval{ip}')$ and $\textbf{disconnect}(\dval{ip},\dval{ip}')$
should be thought of as internal actions. Finally, actions
$\listen{m}$ are simply blocked by the encapsulation---they
cannot occur without synchronising with a $\starcastP{m}$---except for
$\colonact{\{\dval{ip}\}\neg K}{\listen{\newpkt{\dval{d}}{\dval{dip}}}}$
with $\dval{d}\in\tDATA$ and $\dval{dip}\in \tIP$ (Rule \tableref{e-np}). This action
represents new data \dval{d} that is submitted by a
client of the modelled protocol to node $\dval{ip}$, for delivery at
destination \dval{dip}.

\subsubsection{Optional Augmentations to Ensure Non-Blocking Broadcast.}
\label{ssec:non-blocking}

Our process algebra, as presented above, is intended for networks in which
each node is \emph{input enabled} \cite{LT89}, meaning that it is
always ready to receive any message, i.e., able to engage in the
transition $\receive{m}$ for any $m\in \tMSG$---in the default
version of \tawn, network expressions are required to
have this property. In our model of AODV
(Section~\ref{sec:aodv}) we will ensure this by equipping each node
with a message queue that is always able to accept messages for later
handling---even when the main sequential process is currently busy.
This makes our model input enabled and hence \emph{non-blocking}, meaning that no sender can
be delayed in transmitting a message simply because one of the
potential recipients is not ready to receive it.

In \cite{ESOP12,TR13} we additionally presented\journalonly{\outP{ed}\comsP{?}}
two versions of AWN without the
requirement that all nodes need to be input enabled: one in which we
kept the same operational semantics and simply accept blocking,
and one were we added operational rules to avoid blocking,
thereby giving up on the requirement that any broadcast
message is received by all nodes within transmission range.

The first solution does not work for \tawn, as it would give rise to
\emph{time deadlocks}, reachable states where time is unable to
progress further. \journalonly{This happens for instance when one node can only do
a {\starcastP{m}} action, and another node, within transmission range,
needs to do a time step before it can receive it.}

The second solution is therefore our only alternative to
requiring input enabledness for \tawn. As in \cite{ESOP12,TR13},
it is implemented by the addition of the rule
\hypertarget{non-blocking}{$$\frac{P \nar{\receive{m}}}  {\rule[13pt]{0pt}{1pt}
  \dval{ip}:P:R \ar{\colonact{\{\dval{ip}\}\neg\emptyset}{\listen{m}}} \dval{ip}:P:R}\;.
  \label{non-blocking}$$}
It states that a message may arrive at a node \dval{ip} regardless whether
the node is ready to receive it or not; if it is not ready, the message is simply
ignored, and the process running on the node remains in the same state.

In \cite[\textsection4.5]{TR13} also a variant of this idea is presented that 
avoids negative premises, yet leads to the same transition system.
The same can be done to \tawn in the same way, we
skip the details and refer to \cite[\textsection4.5]{TR13}.

\subsection{Results on the Process Algebra}

In this section we list a couple of useful properties of our timed process algebra. 
In particular, we show that wait actions do not change the data state, except for the value of \now.
Moreover, we show the absence of \emph{time deadlocks}:
a complete network $N$ described by T-AWN always admits a transition,\vspace{1pt} independently of the outside environment.
More precisely, either $N \ar{\tick}$, or $N \ar{\colonact{\dval{ip}}{\deliver{\dval{d}}}}$ or $N \ar{\tau}$.
We also show that
our process algebra admits a translation into one without data structure. 
The operational rules of the translated process algebra are in the de Simone format \cite{dS85}, which 
immediately implies that strong bisimilarity is a congruence, and
yields the associativity of
our parallel operators. Last, we show that \tawn and \awn are related by a simulation relation.
Due to lack of space, most of the proofs are omitted; they
are deferred to \SubApp{proofs}.

\begin{proposition} On the level of sequential processes, wait actions
change only the value of the variable \now, i.e.,
$\xi, \p \ar{w_1} \zeta, \q  \Rightarrow (\p=\q \wedge \zeta=\xi\timeinc)$.
\label{prop:time determinism}%
\label{w-caracterisation} %
\end{proposition}

\begin{proofsketch}
One inspects all rules of Table~\ref{tab:sos sequential} that can
generate $w$-steps, and then reasons inductively on the derivation of
these steps.
\end{proofsketch}

\prfatend
Only the six rules below generate a $w_1$-step (under certain conditions).
\begin{enumerate}\vspace{-3pt}
\item $\xi,\send{\dexp{ms}}.\p \ar{\ws~} \xi\timeinc,\send{\dexp{ms}}.\p$
\item $\xi,\receive{\keyw{msg}}.\p \ar{\wtr~} \xi\timeinc,\receive{\keyw{msg}}.\p$
\item $\xi,\p \ar{\w}  \xi\timeinc,\p $ 
\item $\displaystyle \frac{\emptyset[\keyw{var}_i:=\xi(\dexp{exp}_i)]_{i=1}^n,\p \ar{w_1} \xii,\p'}
  {\xi,X(\dexp{exp}_1,...,\dexp{exp}_n) \ar{w_1} \xi\timeinc,X(\dexp{exp}_1,...,\dexp{exp}_n)}
  ~\mbox{(\small$X(\keyw{var}_1,...,\keyw{var}_n) \stackrel{{\it def}}{=} p$)}$ 
\item $\displaystyle \frac{\xi \nar{\varphi}}{\xi,\cond{\varphi}\p \ar{\w} \xi\timeinc,\cond{\varphi}\p}$\vspace{6pt}
\item $\displaystyle \frac{\xi,\p \ar{w_1} \xii,\p' \quad \xi,\q \ar{w_2} \xii,\q'}{\xi,\p+\q \ar{w_1\wedge w_2} \xii,\p'+\q'}$
\end{enumerate}

\noindent We reason inductively on the derivation of the $w_1$-step. If
one of the Rules 1--5 is applied
then the result follows by the form of the rule.
For Rule 6, by the induction hypothesis, 
$\p = \p'$, $\q = \q '$ and hence $\p+\q = \p'+\q'$. Moreover, $\zeta = \xi \timeinc$.
\eprf
Similarly, it can be observed that for transmission actions (actions from the set $\R\mathop{:}\W$)
the data state does not change either; 
the process, however, changes. That means \plat{$\xi, \p \ar{rw} \zeta, \q  \Rightarrow \zeta=\xi\timeinc$}
for all $rw \in \RW$.
Furthermore, this result can easily be lifted to all other layers of our process algebra (with minor adaptations: for example on node expressions one has
to consider $\tick$ actions).

To shorten the forthcoming definitions and properties we 
use the following abbreviations:
\noindent
\newcommand{\arrec}{\ar{\makebox[22pt]{\scriptsize\bf rcv.}}}
\newcommand{\arsend}{\ar{\makebox[22pt]{\scriptsize\bf send}}}
\newcommand{\arwait}{\ar{\makebox[22pt]{\scriptsize\bf wait}}}
\newcommand{\arother}{\ar{\makebox[22pt]{\scriptsize\bf other}}}
\newcommand{\arinb}{\ar{\makebox[19.7pt]{\scriptsize\bf inb}}}
\newcommand{\narrec}{\nar{\makebox[19.7pt]{\scriptsize\bf rcv.}}}
\newcommand{\narsend}{\nar{\makebox[19.7pt]{\scriptsize\bf send}}}
\newcommand{\narwait}{\nar{\makebox[19.7pt]{\scriptsize\bf wait}}}
\newcommand{\narother}{\nar{\makebox[19.7pt]{\scriptsize\bf other}}}
\newcommand{\narinb}{\nar{\makebox[19.7pt]{\scriptsize\bf inb}}}
\begin{enumerate}
\item $P {\arrec}$ iff $P{\ar{\receive{m}}}$ for some $m\in\tMSG$,
\item $P {\arsend}$ iff $P{\ar{\send{m}}}$ for some $m\in\tMSG$,
\item $P {\arwait}$ iff $P{\ar{w_1}}$ for some $w_1\in\W$,
\item $P {\arother}$ iff $P{\ar{a}}$ for some $a\in\act$ not of the forms above,
\end{enumerate}
where $P$ is a parallel process expression---possibly
incorporating the construct
$\stargcast{\range}{\dval{m}}[n,o].\p$, but never in a $+$-context.
Note that the last line covers also transmission actions $rw\in\RW$.
The following result%
\journalonly{\rpPa{The following result}{\Thm{SP-classification}}} shows that the wait actions of a sequential
process (with data evaluation) $P$ are completely determined
by the other actions $P$ offers.

\toApp
{
\renewcommand{\thelemma}{A.\arabic{applemmacount}}
\stepcounter{applemmacount}
\addtocounter{lemma}{-1}
\begin{lemma}\label{lem:initials}
Let $X(\keyw{var}_1\comma\ldots\comma\keyw{var}_n) \stackrel{{\it def}}{=} p$.
Then\vspace{-1ex}
\begin{enumerate}
\item $\xi,X(\dexp{exp}_1,\ldots,\dexp{exp}_n) \arrec$ \quad iff \quad
  $\emptyset[\keyw{var}_i:=\xi(\dexp{exp}_i)]_{i=1}^n,\p \arrec$\,,
\item $\xi,X(\dexp{exp}_1,\ldots,\dexp{exp}_n) \arsend$ \quad iff \quad
  $\emptyset[\keyw{var}_i:=\xi(\dexp{exp}_i)]_{i=1}^n,\p \arsend$\,,
\item $\xi,X(\dexp{exp}_1,\ldots,\dexp{exp}_n) \arother$ \quad iff \quad
  $\emptyset[\keyw{var}_i:=\xi(\dexp{exp}_i)]_{i=1}^n,\p \arother$\,,
\item $\xi,X(\dexp{exp}_1,\ldots,\dexp{exp}_n) \ar{w_1}$\hspace{11.7pt}
\quad iff \quad
  $\emptyset[\keyw{var}_i:=\xi(\dexp{exp}_i)]_{i=1}^n,\p \ar{w_1}%
  $\,,\hfill$\forall w_1\mathop\in\W$\,.
\end{enumerate}
\end{lemma}
}
\prf
The first three claims follow immediately from Rule \tableref{rec},
and the last claim from \tableref{rec-w}.
\eprf
\etoApp

\begin{theorem}
\label{thm:SP-classification}
Let $P$ be a state of a sequential process.\vspace{-1ex}
\begin{enumerate}
\item $P{\stackrel{\w}{\longrightarrow}}$   	\quad iff \quad $P {\narrec} \wedge P {\narsend} \wedge P {\narother}$\,.
\item $P{\stackrel{\wtr}{\longrightarrow}}$	\quad iff \quad $P {\arrec}   \wedge P {\narsend} \wedge P {\narother}$\,.
\item $P{\stackrel{\ws}{\longrightarrow}}$	\quad iff \quad $P {\narrec} \wedge P {\arsend}   \wedge P {\narother}$\,.
\item $P{\stackrel{\wrs}{\longrightarrow}}$ \quad iff \quad $P {\arrec}   \wedge P {\arsend}   \wedge P {\narother}$\,.
\end{enumerate}
\end{theorem}

\begin{proofsketch}
The proof is by structural induction.
It requires, however, a distinction between guarded terms (as defined in Footnote~\ref{guarded})
 and unguarded ones.
\end{proofsketch}
We could equivalently have omitted all transition rules involving wait actions from
\Tab{sos sequential}, and defined the wait transitions for sequential processes as described by 
\Thm{SP-classification} and \Prop{time determinism}.
That our transition rules give the same result constitutes a sanity check of our operational semantics.

\prfnoboxatend
Let $P=\xi,s$.
Let us first show the result for guarded terms~$s$ (as defined in Footnote~\ref{guarded}).
We reason inductively on the structure of $s$.
\vspace{-1ex}
\begin{itemize}
  \item $s = \unicast{\dexp{dest}}{\dexp{ms}}\p \prio \q $ or $s = \alpha.\p $ with
    $\alpha \in \{\groupcastP{\dexp{dests}}{\dexp{ms}},\linebreak[3]
    \broadcastP{\dexp{ms}}, \deliver{\dexp{data}}, \assignment{\keyw{var}\mathbin{:=}\dexp{exp}}\}$.
    In case $\xi(s)\defined$ we have $P {\arother}$ and $P {\narwait}$,
    using the rules of Table \ref{tab:sos sequential}.
    In case $\xi(s){\uparrow}$ we have $P{\ar{\w}}$ and
    $P {\narrec} \wedge P {\narsend} \wedge P {\narother}$.
  \item $s = \receive{\msg}.\p $.
    In this case $P{\ar{\wtr}}$ and
    $P {\arrec} \wedge P {\narsend} \wedge P {\narother}$.
  \item $s = \send{\dexp{ms}}.\p $.
    In case $\xi(s)\defined$ we have $P{\ar{\ws}}$ and
    $P {\narrec} \wedge P {\arsend} \wedge P {\narother}$.
    In case $\xi(s){\uparrow}$ we have $P{\ar{\w}}$ and
    $P {\narrec} \wedge P {\narsend} \wedge P {\narother}$.
  \item $s = \stargcast{\range}{\dval{m}}[n,o].\p$.
    In this case $P {\arother}$ and $P {\narwait}$.
  \item $s = [\varphi]\p$.
    In case \plat{$\xi\stackrel{\varphi}{\rightarrow}\xii$} for some $\xii$ we have
    $P{\ar{\tau}}$, hence $P {\arother}$, and $P {\narwait}$.
    In case \plat{$\xi\nar{\varphi}$} we have $P{\ar{\w}}$ and
    $P {\narrec} \wedge P {\narsend} \wedge P {\narother}$.
  \item $s = s_1 + s_2$. Since $s$ is a guarded term, also $s_1$ and $s_2$ must be guarded terms.\vspace{2pt}
    \begin{itemize}
    \item Assume $\xi,s_i{\arother}$ for $i=1$ or $2$.
          By induction, $\xi,s_i{\narwait}$, and hence $\xi,s{\narwait}$.
          Moreover, by Rules~\tableref{alt-l} and~\tableref{alt-r} of Table \ref{tab:sos sequential}, $\xi,s {\arother}$.\\
          For the remaining cases assume that $\xi,s_i{\narother}$ for $i=1$ and $2$.
    \item Depending on whether $\xi,s_i {\arrec}$ and $\xi,s_i {\arsend}$ for $i=1,2$
          there are $2^4{=}16$ cases left. As they all proceed in the same way, we show only one.\\
          Assume $\xi,s_1 {\narrec} \ans \xi,s_1 {\arsend} \ans \xi,s_2 {\arrec}
          \ans \xi,s_2 {\narsend}$. By induction $\xi,s_1 {\ar{\ws}}$ and $\xi,s_2 {\ar{\wtr}}$.
          By Rule~\tableref{alt-w} of Table \ref{tab:sos sequential} $\xi,s {\ar{\wrs}}$, and by
          Rules~\tableref{alt-l} and~\tableref{alt-r}
          $\xi,s {\arrec}$ and $\xi,s {\arsend}$.\vspace{2pt}
    \end{itemize}
  \item $s = X(\dexp{exp}_1,\dots,\dexp{exp}_n)$.
    This case cannot occur, as $s$ is not a guarded term.
\end{itemize}
Let us now show the result for all terms, again using structural induction on $s$.\linebreak[2]
All cases proceed exactly as above (but skipping the guardedness check in the case for $+$),
except for the case $s=X(\dexp{exp}_1,\dots,\dexp{exp}_n)$.\vspace{-2ex}
\begin{itemize}
  \item $s = X(\dexp{exp}_1,\dots,\dexp{exp}_n)$.
  In this case $X(\keyw{var}_1\comma\ldots\comma\keyw{var}_n) \stackrel{{\it def}}{=} p$ for a guarded term $p$.
  Now the claim is an immediate corollary of \Lem{initials} and the result for guarded terms $p$ obtained above.
\qed
\end{itemize}
\eprf

\toApp
It is tempting to integrate the two parts of the above proof into one treatment of guarded and
unguarded terms alike. A problem with that approach would be that in the very last case $p$ is a
bigger term than $s$, so that structural induction on $s$ does not work. It is not a priori clear
which inductive argument would take its place. In fact, there is no straightforward solution to
this, because if there were, the result would hold without the restriction of \tawn to guarded
recursion, considering that that this restriction is not needed for \Lem{initials} and is not used
anywhere else in the above proof than in the topmost case distinction. \journalonly{Yet,}
\etoApp

\Thm{SP-classification}
does not hold in the presence of unguarded recursion.\vspace{1pt}
A counterexample is given by the expression $X()$
with \plat{$X()\stackrel{{\it def}}{=}X()$}, for which we would have
$X() {\narrec} \wedge X() {\narsend} \wedge X() {\narother} \wedge X() {\narwait}$.

\begin{lemma}\label{lem:transmission actions}
Let $P$ be a state of a sequential or parallel process.
If \plat{$P \ar{\Rw w_1}$} for some $R\subseteq\tIP$ and $w_1\mathbin\in\W$ then
\plat{$P \ar{R'\!\mathop:w_1}$} for any $R'\subseteq \tIP$.
\end{lemma}
\prfatend
For sequential processes, this follows directly from 
Rules \tableref{tr} and \tableref{tr-o} of \Tab{sos sequential}.
For parallel processes, it is a trivial structural induction.
\eprf

\begin{observation}\label{obs:transmission actions}
Let $P$ be a state of a sequential process.
If \plat{$P \ar{\Rw w_1}$} for some $w_1\mathbin\in\W$ then $w_1$ must be $\w$ and all
outgoing transitions of $P$ are labelled $R'\mathop:\w$.
\end{observation}
\journalonly{\comsP{skip/split the following sentence}}
For $N$ a (partial) network expression, or a parallel process expression, write
$N {\arinb}$ iff $N{\ar{a}}$ with $a$ of the form\vspace{1pt}
$\colonact{R}{\starcastP{m}}$, $\colonact{\dval{ip}}{\deliver{\dval{d}}}$ (or
$\deliver{\dval{d}}$) or $\tau$---an \emph{instantaneous non-blocking action}.\vspace{1pt}
\journalonly{\chW{For a parallel process expression $P$, write
$P {\arinb}$ iff $P{\ar{a}}$ with $a$ of the form\vspace{1pt}
$\colonact{R}{\starcastP{m}}$, $\colonact{\dval{ip}}{\deliver{\dval{d}}}$ or $\tau$---an \emph{instantaneous non-blocking action}.
}}
Hence, for a parallel process expression $P$, \plat{$P {\arother}$} iff \plat{$P {\arinb}$} or \plat{$P \ar{\Rw w_1}$} for $w_1\mathbin\in\W$.
Furthermore, write $P {\ar{\bf time}}$ iff $P{\ar{w_1}}$ or $P{\ar{\Rw w_1}}$ for some $w_1\mathbin\in\W$.
We now lift Theorem~\ref{thm:SP-classification} to the level of parallel processes.\pagebreak[3]

\begin{theorem}
\label{thm:PP-classification}
Let $P$ be a state of a parallel process.\vspace{-1ex}
\begin{enumerate}
\item $P{\stackrel{\w}{\longrightarrow}} \vee P{\ar{\Rw\w~\,}}$\hspace{-0.7pt} \quad iff \quad $P {\narrec} \wedge P {\narsend} \wedge P {\narinb}$\,.
\item $P{\stackrel{\wtr}{\longrightarrow}} \vee P{\ar{\Rw\wtr\,}}$ \quad iff \quad $P {\arrec} \wedge P {\narsend} \wedge P {\narinb}$\,.
\item $P{\stackrel{\ws}{\longrightarrow}} \vee P{\ar{\Rw\ws\,}}$ \quad iff \quad $P {\narrec} \wedge P {\arsend} \wedge P {\narinb}$\,.
\item $P{\stackrel{\wrs}{\longrightarrow}} \vee P{\ar{\Rw\wrs}}$\! \quad iff \quad $P {\arrec} \wedge P {\arsend} \wedge P {\narinb}$\,.
\end{enumerate}
\end{theorem}

\prfatend
We apply structural induction on $P$. First suppose $P$ has the form $\xi,s$.
In case $P{\ar{R:w_1}}$ with $w_1\mathbin\in\W$, the claim follows from \Obs{transmission actions}.\\
In case $P{\nar{R:w_1}}$ for all $w_1\mathbin\in\W$, the claim follows from \Thm{SP-classification}.

Now consider an expression $P\parl Q$. In case $P{\arinb}$ or $Q{\arinb}$ then also
$P\parl Q {\arinb}$ by Rules \tableref{p-al} and \tableref{p-ar} of \Tab{sos parallel}. By induction,
$P{\nar{\bf time}}$ or $Q{\nar{\bf time}}$, so $P\parl Q {\nar{\bf time}}$.
For the remaining cases assume that $P{\nar{\bf inb}}$ and $Q{\nar{\bf inb}}$.

In case $P{\arrec}$ and $Q{\arsend}$ we have $P\parl Q {\ar{\tau}}$ by the third
rule of \Tab{sos parallel}. Moreover, $P\parl Q {\nar{\bf time}}$.
For the remaining cases assume that the combination $P{\arrec}$ and $Q{\arsend}$ does
not apply, so that $P\parl Q {\nar{\bf inb}}$.

In case $P{\narsend}$ and $Q{\narrec}$ we have
$P\parl Q {\narsend}$ and $P\parl Q {\nar{\bf rec}}$.
By induction, $P{\ar{\w}} \vee P{\ar{\Rw\w}} \vee P{\ar{\wtr}} \vee P{\ar{\Rw\wtr}}$
and $Q{\ar{\w}} \vee Q{\ar{\Rw\w}} \vee Q{\ar{\ws}} \vee Q{\ar{\Rw\ws}}$, so that
$P\parl Q {\ar{\w}} \vee P\parl Q {\ar{\Rw\w}}$.

In case $P{\narsend}$ and $Q{\arrec}$ we have
$P\parl Q {\narsend}$ and $P\parl Q {\ar{\bf rec}}$.
By induction, $P{\ar{\w}} \vee P{\ar{\Rw\w}} \vee P{\ar{\wtr}} \vee P{\ar{\Rw\wtr}}$
and $Q{\ar{\wtr}} \vee Q{\ar{\Rw\wtr}} \vee Q{\ar{\wrs}} \vee Q{\ar{\Rw\wrs}}$, so that
$P\parl Q {\ar{\wtr}} \vee P\parl Q {\ar{\Rw\wtr}}$.

In case $P{\arsend}$ and $Q{\narrec}$ we have
$P\parl Q {\arsend}$ and $P\parl Q {\nar{\bf rec}}$.
By induction, $P{\ar{\ws}} \vee P{\ar{\Rw\ws}} \vee P{\ar{\wrs}} \vee P{\ar{\Rw\wrs}}$
and $Q{\ar{\w}} \vee Q{\ar{\Rw\w}} \vee Q{\ar{\ws}} \vee Q{\ar{\Rw\ws}}$, so that
$P\parl Q {\ar{\ws}} \vee P\parl Q {\ar{\Rw\ws}}$.

In case $P{\arsend}$ and $Q{\arrec}$ we have
$P\parl Q {\arsend}$ and $P\parl Q {\ar{\bf rec}}$.
By induction, $P{\ar{\ws}} \vee P{\ar{\Rw\ws}} \vee P{\ar{\wrs}} \vee P{\ar{\Rw\wrs}}$
and $Q{\ar{\wtr}} \vee Q{\ar{\Rw\wtr}} \vee Q{\ar{\wrs}} \vee Q{\ar{\Rw\wrs}}$, so that
$P\parl Q {\ar{\wrs}} \vee P\parl Q {\ar{\Rw\wrs}}$.
\eprf

\begin{corollary}
\label{cor:PP-classification}
Let $P$ be a state of a parallel process.
Then 
$P{\ar{\bf time}}$ 
iff 
$P {\nar{\bf inb}}$.
\qed
\end{corollary}

\journalonly{\chW{\noindent
For a (partial) network expression $N$ write
\vspace{-1ex}
\begin{itemize}
\item[] $N {\arinb}$ iff $N{\ar{a}}$ with $a$ of the form
$\colonact{R}{\starcastP{m}}$, $\colonact{\dval{ip}}{\deliver{\dval{d}}}$ or $\tau$.
\end{itemize}
}}

\toApp
{
\renewcommand{\thelemma}{A.\arabic{applemmacount}}
\stepcounter{applemmacount}
\addtocounter{lemma}{-1}
\begin{lemma}\label{lem:arrive-enabled node}
$\dval{ip}:P:R{\ar{\colonact{\{\dval{ip}\}\neg \emptyset}{\listen{\dval{m}}}}}$ for any $m\in\tMSG$,
and any $\dval{ip}$, $P$ and~$R$.
\end{lemma}
}
\prf
This is our only proof in which the selected version of \tawn matters---see
Pages~\pageref{ssec:non-blocking}--\hyperlink{non-blocking}{\pageref*{non-blocking}}.

In the default version, we require for any node expression \mbox{$\dval{ip}:P:R$} that
$P \ar{\receive{m}}$ for any $m$---this is the definition of \emph{input enabledness}.
The claim then follows from Rule \tableref{n-rcv} of \Tab{sos node}.

In the alternative version, the claim follows from that rule, in combination with the rule with a negative
premise on \hyperlink{non-blocking}{Page~\pageref*{non-blocking}}.
\eprf
\etoApp

\begin{lemma}\label{lem:arrive-enabled}
Let $N$ be a partial network expression\vspace{1pt} with $L$ the set of addresses of the nodes of $N$.
Then $N{\ar{\colonact{H\neg K}{\listen{\dval{m}}}}}$\,, for any partition
$L=H\dcup K$ of $L$ into sets $H$ and $K$,
and any $m\in\tMSG$.
\end{lemma}
\prfatend
We apply structural induction on $N$.
If $N$ is a node expression $\dval{ip}:P:R$, we have to show that
$\dval{ip}:P:R{\ar{\colonact{\{\dval{ip}\} \neg \emptyset}{\listen{\dval{m}}}}}$
and also that\\ 
$\dval{ip}:P:R{\ar{\colonact{\emptyset \neg \{\dval{ip}\}}{\listen{\dval{m}}}}}$.
The former follows by \Lem{arrive-enabled node}, and the latter by 
Rule \tableref{n-dis}.

In case $N \mathbin= M_1\|M_2$ the result is obtained using Rule \tableref{arr} of \Tab{sos network}.
\eprf

\noindent
Using this lemma, we can finally show one of our main results: an (encapsulated)
network expression can perform a time-consuming action iff an instantaneous non-blocking action is not possible. 
\begin{theorem}
\label{thm:N-classification}
Let $N$ be a partial or complete network expression.\\
Then $N{\ar{\tick}}$ \quad iff \quad $N {\nar{\bf inb}}$.
\end{theorem}
\prf
We apply structural induction on $N$. First suppose $N$ is a node expression $\dval{ip}\mathop:P\mathop:R$.
Then $N{\ar{\tick}}$ iff $P{\ar{w_1}}\vee P{\ar{\Rw w_1}}$ for some $w_1\mathbin\in\W$.
By \Lem{transmission actions} this is the case iff 
$P{\ar{w_1}}\vee P{\ar{R'\mathop:w_1}}$ for some $R'\subseteq\tIP$ and $w_1\mathbin\in\W$, i.e., iff $P {\ar{\bf time}}$.
Moreover $N{\arinb}$ iff $P{\arinb}$.
Hence the claim follows from \Cor{PP-classification}.

Now suppose $N$ is a partial network expression $M_1\|M_2$.
In case $M_i{\nar{\bf inb}}$ for $i=1,2$ then $N{\nar{\bf inb}}$.
By induction $M_i{\ar{\tick}}$ for $i=1,2$, and hence $N{\ar{\tick}}$.
Otherwise, $M_i{\arinb}$ for $i=1$ or $2$. Now $N {\arinb}$.
In case $M_i{\ar{\tau}}$ or $M_i{\ar{\colonact{\dval{ip}}{\deliver{\dval{d}}}}}$
this follows from the third line of \Tab{sos network}; if 
$M_i{\ar{\colonact{R}{\starcastP{m}}}}$
it follows from the first line, in combination with \Lem{arrive-enabled}.
By induction $M_i{\nar{\tick}}$, and thus $N{\nar{\tick}}$.

Finally suppose that $N$ is a complete network expression $[M]$.
By the rules of \Tab{sos network} $N{\ar{\tick}}$ iff $M{\ar{\tick}}$, and
$N{\arinb}$ iff $M{\ar{\bf inb}}$, so the claim follows from the case for partial network expressions.
\eprf

\begin{corollary}
\label{no_deadlock}
A complete network $N$ described by T-AWN always admits a transition, independently of the outside environment, 
i.e., $\forall N, \exists a$ such that $N \ar{a}$ and $a \not \in \{\textbf{connect}(\dval{ip},\dval{ip}'),\textbf{disconnect}(\dval{ip},\dval{ip}'),\newpkt{\dval{d}}{\dval{dip}}\}$.\\
More precisely, either $N \ar{\tick}$ or $N \ar{\colonact{\dval{ip}}{\deliver{\dval{d}}}}$ or $N \ar{\tau}$.
\qed
\end{corollary}

\journalonly{\subsection{Eliminating Data Structures}}
Our process algebra admits a translation into one without data
structures (although we cannot \emph{describe} the target algebra without
using data structures).  The idea is to replace any variable by all
possible values it can take. 
The target algebra differs from the
original only on the level of sequential processes; the subsequent
layers are unchanged.
A formal definition can be found in \SubApp{eliminating}.
\toApp
\subsection{Eliminating Data Structures}\label{subapp:eliminating}
Our process algebra admits a translation into one without data
structures (although we cannot \emph{describe} the target algebra without
using data structures). 
The target algebra differs from the
original only on the level of sequential processes; the subsequent
layers are unchanged. The syntax of the target language of sequential processes is
given by the following grammar:
\begin{align*}
P\ ::=\ & \nil ~\mid~ P+P ~\mid~ \alpha.P ~\mid~ \stargcast{\range}{\dval{m}}. P \prio P ~\mid~
     X ~\mid\\
     & \sum_{i\in I} P_i ~\mid~
        {\rm \Delta}^i P ~\mid~
         \Deltas{i=k}{} P_i ~\mid~
        \DeltaX[o]{i=k}{*}{n} P_{i}
\\
\alpha\ ::=\ &\tau ~\mid~ \send{\dval{m}} ~\mid~\receive{\dval{m}} ~\mid~ \deliver{\dval{d}}
\end{align*}
Its structural operational semantics displayed in 
\Tab{sos sequential without data}.

Here $\nil$ denotes the inactive process (that can only wait), $+$ is a binary choice (as before) and
$\sum_{i\in I}$ denotes a choice with one argument $P_i$ for each index $i$ from a possibly infinite
set $I$---the chosen process cannot start with a wait action, however.
The process $\tau.P$ performs an internal action $\tau$ and
proceeds as $P$. The actions $\send{\dval{m}}$ and
$\receive{\dval{m}}$ are as before, but now there is one such action%
\begin{table}[tp]\normalsize
\caption{Structural operational semantics for sequential process expressions after elimination of data structures
\label{tab:sos sequential without data}}
\renewcommand{\p}{P}
\renewcommand{\q}{Q}
\vspace{-2.5ex}
$$\begin{array}{@{}l@{\ }r@{}}
\nil \ar{\w} \nil \qquad\qquad\qquad
  \displaystyle\frac{\p \ar{w_1} \p' \quad \q \ar{w_2} \q'}{\p+\q \ar{w_1\wedge w_2} \p'+\q'}
&\mbox{\small $(\forall w_1,w_2\in \W)$}\\[12pt]
\displaystyle
\displaystyle\frac{ \p_j \ar{a} \p'}{\sum_{i\in I}\p_i \ar{a} \p'}
\qquad
\frac{\p \ar{a} \p'}{\p+\q \ar{a} \p'}
\qquad
\displaystyle
\frac{\q \ar{a} \q'}{\p+\q \ar{a} \q'}
&
\hspace{-20pt}\mbox{\small $\left(\forall j\mathbin\in I,~ a\mathbin\in\act\mathord-\W\right)$}
\\[17pt]
\tau.\p \ar{\tau} \p
\qquad
\send{\dexp{m}}.\p \ar{\send{\dexp{m}}} \p
\\[6pt]
\receive{\dexp{m}}.\p \ar{\receive{\dexp{m}}} \p
\qquad
\deliver{\dval{d}}.\p \ar{\deliver{\dval{d}}} \p  
\\[6pt]
\stargcast{\range}{\dexp{m}}.\p \prio \q \ar{R\mathbin{:}\w} \stargcast{(\range\mathop \cap R)}{\dexp{m}}.\p\prio \q
&\mbox{\small $(\forall R\subseteq \tIP)$}
\\[6pt]
\multicolumn{2}{l@{}}{
\stargcast{\emptyset}{\dexp{m}}.\p \prio \q \ar{\stargcast{\emptyset}{\dexp{m}}} \q
\quad\ 
\stargcast{R}{\dexp{m}}.\p \prio \q \ar{\stargcast{R}{\dexp{m}}} \p
\,\;\mbox{\small $(\forall R \mathop{\not=} \emptyset)$}}
\\[6pt]
\displaystyle\frac{\p \ar{a} \p'}{X \ar{a} \p'}
  ~\mbox{(\small$X \mathbin{\stackrel{{\it def}}{=}} \p$)}
\qquad
{\rm \Delta}^{i+1} \p \ar{\w} {\rm \Delta}^{i} \p 
\qquad
\displaystyle \frac{\p \ar{a} \p'}{{\rm \Delta}^{0} \p \ar{a} \p'}
& \mbox{\small $(\forall i\geq0,~ a \in \act)$}
\\[12pt]
\displaystyle\frac{\p_k \ar{w_1} \p'}
{\Deltas{i=k}{} \p_i \ar{w_1} \Deltas{i=k+1}{} \p_i}
\qquad
\displaystyle\frac{\p_k \ar{a} \p'}{\Deltas{i=k}{} \p_i \ar{a} \p'}
&\hspace{-40pt}\mbox{(\small$\forall w_1\mathbin\in\W$, $a\mathbin\in \act-\W$, $k\geq 0$)}
\\[24pt]
\displaystyle\frac{\p_k \ar{\send{\dexp{m}}} \p'}{\Deltas{i=k}{} \p_i \ar{\ws} \Deltas{i=k+1}{} \p_i}
\qquad
\displaystyle\frac{\p_k \ar{\receive{\dexp{m}}} \p'}
{\Deltas{i=k}{} \p_i \ar{\wtr} \Deltas{i=k+1}{} \p_i}
\\[24pt]
\displaystyle\frac{\p_{i} \ar{R\mathbin{:}\w} \p_{i}' ~(\forall i \mathbin\in [k..o])}{\DeltaX[o]{i=k}{*}{n} \p_{i} \ar{R\mathbin{:}\w}\DeltaX[o]{i=k}{*}{n-1} \p'_{i}}
\qquad
\displaystyle\frac{\p_{i} \ar{R\mathbin{:}\w} \p_{i}' ~(\forall i \mathbin\in [k..o{+}1])}{\DeltaX[o+1]{i=k}{*}{n} \p_{i} \ar{R\mathbin{:}\w}
\DeltaX[o+1]{i=k+1}{*}{n\!\!\!} \p'_{i}}
&\hspace{-10pt}\left(\mbox{\small$\begin{array}{@{}ll@{}}\forall R\subseteq \tIP,\\ \forall n{>}0 \mbox{~and~} o{\geq}k{\geq}0\end{array}$}\right)
\\[24pt]
\displaystyle\frac{\p_k \ar{\stargcast{R}{\dexp{m}}} \p'}{\DeltaX[o]{i=k}{*}{0} \p_k \ar{\stargcast{R}{\dexp{m}}}\p'}
&\hspace{-20pt}\mbox{\small($\forall R\mathbin\subseteq \tIP \mbox{~and~} o{\geq}k{\geq}0$)}
\end{array}
\vspace{3ex}
$$
for each message $m\mathbin\in\tMSG$ (not an expression that evaluates to a
message). Likewise, there is one action $\deliver{\dval{d}}$ for each \mbox{$\dval{d} \mathbin\in \tDATA$}.
The process $\stargcast{\range}{\dval{m}}. P \mathbin\prio Q$ can cast the message $m\mathbin\in\tMSG$ to the
destinations $\dval{dsts}\mathbin\subseteq\tIP$ and then proceeds as $P$ or~$Q$, depending on whether
$\dval{dsts}=\emptyset$ or not. Alternatively, $\stargcast{\range}{\dval{m}}. P \mathbin\prio Q$
can perform an action $R{:}\w$ and restrict its set of destinations $\range$ to $R$.
The language features process names $X$ with defining equations
\plat{$X\stackrel{{\it def}}{=} P$}, as usual \cite{Mi89}.

\hspace{1em}
The unary operator ${\rm \Delta}^i\!$,
parametrised with a natural number $i$, performs exactly
$i$ wait actions $\w$ before proceeding as its argument $P$.
The operator $\Deltas{i=k}{}$, with a countably infinite sequence of arguments $P_i$,
performs a number of wait
\end{table}

\noindent
actions $w_1\in\W$ (either $\w$, $\wtr$, $\ws$, or $\wrs$)---possibly 0 or $\infty$;
if a finite amount of wait actions is taken it proceeds as one of its
arguments. In any state during its initial sequence of wait actions it
has a choice between proceeding as its first argument $P_k$, provided $P_k$ starts with a
non-wait action, or doing another wait action $w_1$ and dropping $P_k$ from the list of arguments.
The latter is possible if and only if (1) $P_k$ can do a \textbf{send}-action, in which case $w_1:=\ws$,
(2) $P_k$ can do a \textbf{receive}-action, in which case $w_1:=\wtr$, or (3) $P_k$ itself can do $w_1\in \W$.
%
The operator $\DeltaX[o]{i=k}{*}{n}$ has parameters $k, n \geq 0$ and $o\geq k$,
and $o\mathord-k\mathord+1$ arguments. As long as $n>0$ all its arguments can synchronously
perform an $\Rw{\w}$-action, thereby either decrementing $n$ or incrementing $k$,\vspace{-2pt} in the latter case
loosing its first argument. When $n=0$ the process $\DeltaX[o]{i=k}{*}{n}P_{i}$
behaves as its first argument $P_k$, provided it starts with a \textbf{*cast}-action; otherwise the process deadlocks.

The idea behind the translation is to replace any variable by all possible values it can take.
Formally, processes $\xi,p$ are replaced
by $\T_\xi(p)$, where $\T_\xi$ is defined inductively by\vspace{.5mm}

$\T_\xi(p)=
\begin{cases} 
\nil
& \!\!\!\!\mbox{\small ${\rm if}~ \xi[\now\mathop{:=}\now+i](p){\uparrow}\quad \forall i$} 
\\ {\rm \Delta}^{i_{0}}\T_{\xi[\now:=\now+i_{0}]}(p) 
& \!\!\!\!\mbox{\small $\begin{array}{@{}l@{}}{\rm with}~ i_{0} = \min\limits_{i\in\NN}(\xi[\now:=\now+i](p)\defined),
  \end{array}$}
\end{cases}$
\hfill if~$\xi(p)\mathord\uparrow$;\\[0.5mm]
Otherwise ($\xi(p){\defined}$):\vspace{.5mm}

\begin{tabular}{@{}l@{}}
$\T_\xi(\broadcastP{\dexp{ms}}.p)= \tau.\T_\xi(\stargcast{\tIP}{\xi(\dexp{ms})}[\keyw{LB},\keyw{\Delta B}].p \prio \p)$,\\[0.5mm]
$\T_\xi(\groupcastP{\dexp{dests}}{\dexp{ms}}.p)=  \tau.\T_\xi(\stargcast{\xi(\dexp{dests})}{\xi(\dexp{ms})}[\keyw{LG},\keyw{\Delta G}].p \prio \p)$,\\[0.5mm]
$\T_\xi(\unicast{\dexp{dest}}{\dexp{ms}}.p \prio q)=  \tau.\T_\xi(\stargcast{\{\xi(\dexp{dest})\}}{\xi(\dexp{ms})}[\keyw{LU},\keyw{\Delta U}].p \prio \q)$,\\[0.5mm]
$\T_\xi(\stargcast{\range}{\dval{m}}[n,o].p \prio \q)=$\\
\hfill  $\DeltaX[o]{i=0}{*}{n}\stargcast{\range}{\dval{m}}.\T_{\xi[\now:=\now+i+n]}(p) \prio \T_{\xi[\now:=\now+i+n]}(q)$,\\
$\T_\xi(\send{\dexp{ms}}.p)= \Deltas{i=0}{}P_i$, with\\[0.5mm]
\qquad\qquad\qquad\quad $P_{i}= 
\begin{cases}
\multicolumn{2}{l}{\send{\xi[\now:=\now+i](\dexp{ms})}.\T_{\xi[\now:=\now+i]}(p)}\\
&\mbox{\small {\rm if}~$\xi[\now:=\now{+}i](\dexp{ms})\defined$}\\
\nil \qquad & \mbox{\small otherwise,}
\end{cases}$\\[0.5mm] 

$\T_\xi(\receive{\msg}.p)= \Deltas{i=0}{}\sum_{m\in\tMSG}\receive{m}.\T_{\xi[\msg:=m;~\now:=\now+i]}(p)$,\\[0.5mm]
$\T_\xi(\deliver{\dexp{data}}.p)= \deliver{\xi(\dexp{data})}.\T_\xi(p)$,\\[0.5mm]
$\T_\xi(\assignment{\keyw{var}:=\dexp{exp}}p)= \tau.\T_{\xi[\keyw{var}:=\xi(\dexp{exp})]}(p)$,\\[0.5mm]
$\T_\xi([\varphi]p)=
\begin{cases} 
\nil
& \mbox{\small {\rm if}~ $\xi[\now:=\now+i] \,\,\,\,\not\!\!\!\!\stackrel{\varphi}{\longrightarrow}\quad \forall i $}
\\ {\rm \Delta}^{i_{0}}\sum_{ \{\zeta\mid\xi[\now:=\now+i_{0}] \stackrel{\varphi}{\rightarrow}\xii\}} \tau.\T_{\xii}(p) 
& \mbox{\small $\begin{array}{@{}l@{}} {\rm with}\\[-5pt]
  i_{0} = \min\limits_{i\in\NN}( \xi[\now\mathbin{:=}\now{+}i]\stackrel{\varphi}{\rightarrow}),
  \end{array}$}
\end{cases}$\\[0.5mm] 
$\T_\xi(p+q)=\T_\xi(p) + \T_\xi(q)$,\\[0.5mm]
$\T_\xi(X(\dexp{exp}_1,\ldots,\dexp{exp}_n))= \Deltas{i=0}{}X_{\xi[\now:=\now+i](\dexp{exp}_1),\ldots,\xi[\now:=\now+i](\dexp{exp}_n)}$.\\[0.5mm]
\end{tabular}

\noindent
The last equation requires the introduction of a process name
$X_{\vec{v}}$ for every name
$X: \keyw{TYPE}_1 \times \cdots \times \keyw{TYPE}_n \rightarrow \keyw{SPROC}$
(with defining equation \plat{$X(\overrightarrow{\keyw{var}}) \stackrel{{\it def}}{=} p$})
in the source language
and every vector $\vec{v} \in \keyw{TYPE}_1 \times \cdots \times \keyw{TYPE}_n$ of data values of
the appropriate types the for arguments of $X$.
Its defining equation in the data-free target language is\vspace{-6pt}
\[X_{\vec{v}} \stackrel{{\it def}}{=} \T_{\emptyset[\overrightarrow{\keyw{var}}:=\vec{v}]}(p)\ .\]
\etoApp
The resulting process algebra has a structural operational semantics in the (infinitary)
\emph{de Simone} format, generating the same transition system---up to strong
bisimilarity, $\bis$\,---as the original, which provides some results `for free'. 
\toApp
The resulting process algebra has a structural operational semantics in the (infinitary)
\emph{de Simone} format \cite{dS85}, generating the same transition system---up to strong
bisimilarity, $\bis$\,---as the original.

{
\renewcommand{\thetheorem}{A.\arabic{apptheoremcount}}
\stepcounter{apptheoremcount}
\addtocounter{theorem}{-1}
\begin{theorem}\label{thm:elimination}
There exists a relation $\B$, a \emph{bisimulation}, between states $\xi,p$ of sequential processes in
the source algebra, and sequential processes $P$ in the target algebra, such that\vspace{-1ex}
\begin{itemize}
\item $(\xi,p) \B \T_\xi(p)$ for all sequential process expressions $p$ and valuations $\xi$,
\item if $(\xi,p) \B P$ and $\xi,p \ar{a} \xi',p'$ then $\exists P'$ such that
  $P \ar{a} P'$ and $(\xi',p') \B P'$,
\item if $(\xi,p) \B P$ and $P \ar{a} P'$ then $\exists \xi'\!,p'$ such that
  $\xi,p \ar{a} \xi'\!,p'$ and $(\xi',p') \B P'\!$.\vspace{-1ex}
\end{itemize}
\end{theorem}
}
\begin{proof}
We call $\DeltaX[o]{i=k}{*}{n}P_{i}$ a \emph{variant} of $\DeltaX[o-k]{i=0}{*}{n}P_{i+k}$. \vspace{-6pt}
Likewise, $\Deltas{i=k}{}P_i$ is a variant of $\Deltas{i=0}{} P_{i+k}$.

The relation $\B$ relates any state $\xi,p$ to $\T_\xi(p)$ and to all
variants of $\T_\xi(p)$. It therefore automatically satisfies the first requirement of \Thm{elimination}.
That it satisfies the second requirement follows by a straightforward induction on the derivation of
$\xi,p \ar{a} \xi',p'$ from the rules of \Tab{sos sequential}.
That it satisfies the last requirement follows by a straightforward induction on the derivation of
$P \ar{a} P'$ from the rules of \Tab{sos sequential without data}.
\qed
\end{proof}

\etoApp
For example, it follows that $\bis\,$, and many other
semantic equivalences, are congruences on our language.

\begin{theorem}
Strong bisimilarity is a congruence for all operators of \tawn.
\end{theorem}
This is a deep result that usually takes many pages to establish (e.g.,~\cite{SRS10}).
Here we get it directly from the existing theory on structural
operational semantics, as a result of carefully designing our
language within the disciplined framework described by de Simone~\cite{dS85}.

\begin{theorem}
$\parl$ is associative, and $\|$ is associative and commutative, up to $\bis\,$.
\end{theorem}

\prf
The operational rules for these operators fit a format presented in \cite{CMR08},
guaranteeing associativity up to~$\bis$. The details are similar to
the case for AWN, as elaborated in \cite{ESOP12,TR13}; the only extra complication is
the associativity of the operator $\parl$ on $\W$, as defined on
\hyperlink{parallel}{Page~\pageref*{parallel}}, which we checked automatically by means of the theorem prover Prover9~\cite{Prover9}.
Commutativity of $\|$ follows by symmetry of the rules.
\eprf

\journalonly{\subsection{Simulation of \tawn by AWN}}

\toApp
\subsection{Simulation Results}\label{subapp:simulation}

A \emph{doubly labelled transition system} (L$^2$TS) (over sets $\act$ and $\rm\Sigma$) is a triple
\mbox{$(\bbbp,\rightarrow,\ell)$}, where $\bbbp$ is as a set of \emph{processes} or \emph{states},
$\mathord\rightarrow \subseteq \bbbp \times \act \times \bbbp$ is a \emph{transition relation} and
\emph{$\ell: \bbbp \rightarrow \rm\Sigma$} a \emph{state labelling}.

A \emph{simulation} is a binary relation between the states of two L$^2$TSs satisfying the
\emph{transfer property}: any transition from a state in the source L$^2$TS can be mimicked by a
``similar'' transition from a related state in the target L$^2$TS, such that the end states of both
transitions are again related. Usually a similar transition is taken to be one with the same
label, but here we generalise this by parametrising a simulation with an explicit similarity
relation $\U^\act$ between the transition labels of the two L$^2$TSs. We additionally require
related states to have a similar state label, a notion parametrised by an explicit similarity relation
$\U^{\rm\Sigma}$ between the state labels of the two L$^2$TSs.
A \emph{weak} simulation \cite{vG93} allows, in satisfying the transfer property, internal actions $\tau$
to precede and follow the mimicking transition---moreover, it allows internal transitions $\tau$ to be
mimicked by doing no transition.

For $P,Q\in\bbbp$, write $P \ar a Q$ for $(P,a,Q)\in\mathord{\rightarrow}$.
Suppose that $\act$ contains the \emph{internal action}~$\tau$.
Then $P \Longrightarrow Q$ for $P,Q\in\bbbp$ denotes an arbitrary (possibly empty) sequence of
$\tau$-transitions, i.e., there are states
$P_{0},\dots,P_n$ with $P=P_{0}\ar\tau P_{1} \cdots P_{n-1} \ar\tau P_n=Q$.
Moreover, $P \dar{a} Q$, with $a \in \act$, denotes 
\plat{$P \Longrightarrow \ar{a} \Longrightarrow Q$}, 
and $P \dar{\hat a} Q$ denotes $P \Longrightarrow Q$ if $a=\tau$
 and 
$P \dar{a} Q$ otherwise.


{
\renewcommand{\thedefinition}{A.\arabic{appdefcount}}
\stepcounter{appdefcount}
\addtocounter{definition}{-1}
\begin{definition}\rm\label{def:simulation}
Let $(\bbbp_i,\rightarrow_i,\ell_i)$ for $i=1,2$ be two L$^2$TSs,
labelled over sets $\act_i$ and $\rm\Sigma_i$, respectively.
Furthermore, let $\U^\act \subseteq \act_1 \times \act_2$ and
$\rm\U^\Sigma \subseteq \Sigma_1 \times \Sigma_2$.\linebreak[3]
A \emph{weak simulation} w.r.t.\ $\U^\act$ and $\rm\U^\Sigma$ is a binary relation
$\mathord{\Sim}\subseteq \bbbp_1\times\bbbp_2$ between the states of the two L$^2$TSs,
such that\vspace{-1ex}
\begin{itemize}
\item if $P \Sim Q$ and $P\!\ar{a} P'$ then $\exists Q',b$ such that
  $Q\!\dar{\hat b} Q'$, $a \U^\act\! b$ and $P' \Sim Q'$,%
  \footnote{In case $b=\tau$, no action needs to be taken, that means, $Q=Q'$ is allowed if 
  $P' \Sim Q$.}
\item if $P \Sim Q$ then $\ell_1(P) \U^{\rm\Sigma} \ell_2(Q)$.
\end{itemize}
When $P \Sim Q$, we speak of a weak simulation of $P$ \emph{by} $Q$.
\end{definition}
}

In \otawn the sequential processes $\xi,p$, equipped with the transition
relation generated by the rules of \Tab{sos sequential}, form a double labelled transition system by
taking $\ell(\xi,\p)\mathbin{:=}\xi$, the \emph{data state} of $\xi,p$.
In generalising this idea to parallel processes and network expressions we postulate two requirements on applications of \otawn:
\begin{enumerate}
\vspace{-1ex}
\item In a parallel process expression $\xi_1,p_1 \parl \dots \parl \xi_n,p_n$ the variables maintained by the
  $p_i$ (i.e., the domains of the partial functions $\xi_i$) are
  pairwise disjoint.
\item Each node expression $\dval{ip}:P:R$ occurring in a (partial) network expression has a different
  address $\dval{ip}$.
\vspace{-1ex}
\end{enumerate}
The first requirement is not a restriction at all, since it can easily be achieved by renaming.%
\footnote{If the variable \now is renamed, the SOS rules of \Sect{tawnsos} have to be adapted accordingly.}
The second requirement is satisfied for all applications we encountered so far. Dropping one of the 
requirements would
merely increase the bookkeeping effort of defining the notion of a global data state.
Requirement 1 allows us to define the data state $\ell(P)$ of a parallel process expression 
$P = \xi_1,p_1 \parl \dots \parl \xi_n,p_n$ as $\bigcup_{i=1}^n \xi_n$, whereas Requirement 2
enables a definition of the global data state $\ell(N)$ of a (partial) network expression as a
partial function $\sigma$ that associates with each address $\dval{ip}$ of a node expression
$\dval{ip}:P:R$ occurring in $N$ the data state of $P$.

The above definitions yield L$^2$TSs for the processes and network expressions of \otawn.
To compare the behaviour of AWN and \tawn, we construct weak simulations between their L$^2$TSs.
These will show that each AWN network expression $N$, seen as a \tawn network expression, is weakly
simulated by the AWN-expression $N$, and likewise for AWN process expressions.

To this end, we first define similarity relations between the state and transition labels that occur
in the semantics of \tawn and AWN\@. The only difference in their data states is that \tawn processes
maintain the variable $\now$, which is absent in AWN\@. Consequently,
the similarity relation between the state labels of processes is given by $\xi \U^{\rm\Sigma} \exclnow$
for any \tawn-valuation $\xi$.
Here $\exclnow$ is the AWN valuation obtained by omitting the value of $\now$ from $\xi$.
For networks, this generalises to $\sigma \U^{\rm\Sigma} \exclnow[\sigma]$, where
$\exclnow[\sigma](\dval{ip}):=\exclnow[\sigma(\dval{ip})]$ for all $\dval{ip}\in\mbox{dom}(\sigma)$.

The translation labels of AWN processes include the actions $\broadcastP{\dval{m}}$,
$\groupcastP{\dval{D}}{\dval{m}}$, $\unicast{\dval{dip}}{\dval{m}}$ and
$\neg \unicast{\dval{dip}}{\dval{m}}$ for $m\in\tMSG$, $D\subseteq\tIP$ and $\dval{dip}\in\tIP$;
in \tawn these are replaced by 
$\colonact{\range}{\starcastP{m}}$. Furthermore, \tawn processes have transition labels 
$w_1$ and $R{:}w_1$ for $w_1\in\W$ and $R\subseteq\tIP$, which are absent in AWN\@.
Consequently, the similarity relation between the transition labels of processes is given by
\begin{itemize}
\vspace{-1ex}
\item the identity relation on the \otawn transition labels
$\send{\dval{m}}$, $\deliver{\dval{d}}$, $\receive{\dval{m}}$ and internal actions $\tau$,
for all $\dval{m}\in\tMSG$, $\dval{d}\in\tDATA$,
\item $\colonact{\range}{\starcastP{\dval{m}}} \U^\act b$,
where $b$ is either $\broadcastP{\dval{m}}$, $\groupcastP{\dval{D}}{\dval{m}}$ with
$\dval{dsts}\subseteq D$, $\unicast{\dval{dip}}{\dval{m}}$ with $\dval{dsts}=\{\dval{dip}\}$
or $\neg \unicast{\dval{dip}}{\dval{m}}$ with $\dval{dsts}=\emptyset$,
\item $w_1 \U^\act \tau$ and $R{:}w_1 \U^\act \tau$ for $w_1\in\W$ and $R\subseteq\tIP$.
\vspace{-1ex}
\end{itemize}
On the level of network expressions, the only difference is the \tawn transition label $\tick$,
which is absent in AWN\@. The similarity relation between the transition labels of network
expressions is given by 
\begin{itemize}
\vspace{-1ex}
\item the identity relation on AWN transition labels,%
	\footnote{The labels are $\colonact{R}{\starcastP{\dval{m}}}$,
				$\colonact{H\neg K}{\listen{\dval{m}}}$,
				$\colonact{\dval{ip}}{\deliver{\dval{d}}}$,
				$\textbf{connect}(\dval{ip},\dval{ip}')$,
				$\textbf{disconnect}(\dval{ip},\dval{ip}')$,
				$\colonact{\dval{ip}}{\textbf{newpkt}(\dval{d},\dval{dip})}$
				and $\tau$.}
\item $\colonact{\emptyset}{\starcastP{m}} \U^\act \tau$,
\item $\tick \U^\act \tau$.
\vspace{-1ex}
\end{itemize}

In order for our envisioned simulation to exist, we make one more abstraction:
we read all \textbf{(dis)connect}-actions as $\tau$s. It is with this modification of the L$^2$TSs
of AWN and \tawn in mind that we speak of weak simulations below.

{
\renewcommand{\thetheorem}{A.\arabic{apptheoremcount}}
\stepcounter{apptheoremcount}
\addtocounter{theorem}{-1}
\begin{theorem}
\label{thm:SP_simulation}
Given a common underlying data structure
modulo the variables $\now$\footnote{The variables $\now$ only occur in the data structure of \tawn.} there exists a weak simulation $\Sim$
w.r.t.\ $\U^\act$ and $\rm\U^\Sigma$ of the sequential processes of \tawn by the ones of AWN\@,
such that each AWN process simulates its interpretation as a \tawn process.
\end{theorem}
}

\begin{proof}
Define $\Sim$ as $\Sim_0 \cup \Sim_2$, where $\Sim_0$ consists of all pairs $((\xi,\p),(\exclnow,\p))$ for arbitrary
sequential AWN expressions $\p$---which also are sequential \tawn expressions---and \tawn valuations $\xi$.

\pagebreak[3]
Let $\Sim_1$ be the relation containing the following pairs:
\begin{itemize}
\item $((\xi,\stargcast{\dval{dsts}}{\xi(\dexp{ms})}[n,o].\p \prio \p),(\exclnow,\broadcastP{\dexp{ms}}.\p))$
  \mbox{}\hfill if $\xi(\dexp{ms})\defined$,
\item $((\xi,\stargcast{\dval{dsts}}{\xi(\dexp{ms})}[n,o].\p \prio \p),(\exclnow,\groupcastP{\dexp{dests}}{\dexp{ms}}.\p))$\\
  \mbox{}\hfill if $\xi(\dexp{ms})\defined$, $\xi(\dexp{dests})\defined$ and $\dval{dsts} \subseteq \xi(\dexp{dests})$,
\item $((\xi,\stargcast{\dval{dsts}}{\xi(\dexp{ms})}[n,o].\p \prio \q),(\exclnow,\unicast{\dexp{dest}}{\dexp{ms}}.\p \prio \q))$\\
  \mbox{}\hfill if $\xi(\dexp{ms})\defined$, $\xi(\dexp{dest})\defined$ and $\dval{dsts} \subseteq \{\xi(\dexp{dest})\}$,
\end{itemize} 
for \tawn valuations $\xi$, $\dval{dsts} \subseteq \tIP$, $n,o\in\NN$, sequential \awn processes $p$, $q$
(in the source algebra of $\Sim_1$ again interpreted as \tawn expressions),
and data expressions $\dexp{ms}$, $\dexp{dest}$ and $\dexp{dests}$ of type $\tMSG$, $\tIP$ and
$\pow(\tIP)$, respectively.
Then $\Sim_2$ is the smallest relation containing $\Sim_1$, such that
$((\xi,p),(\zeta,q))\in{\Sim_2} \Rightarrow ((\xi,p),(\zeta,r+q)), ((\xi,p),(\zeta,q+r)) \in{\Sim_2}$,
for all \awn-processes $r$.

To show that $\Sim$ is a weak simulation, we need to demonstrate that it satisfies the two requirements
of Definition~\ref{def:simulation}. The second requirement is satisfied by construction.
Moreover, we obtain a stronger version of the first requirement:
\begin{itemize}
\vspace{-1ex}
\item if $P \Sim_{0} Q$ and $P \ar{w_1} P'$ with $w_1\in\W$ then $P' \Sim_1 Q$,
\item if $P \Sim_{0} Q$ and $P \mathbin{\ar{\tau}} P'$ then either $P' \mathbin{\Sim_2} Q$ or $\exists Q'$ with
  $Q \mathbin{\ar{\tau}} Q'$ and $P' \mathbin{\Sim_1} Q'\!$,
\item if $P \Sim_{0} Q$ and $P \mathbin{\ar{a}} P'$ with $a\mathbin{\not\in}\W{\cup}\{\tau\}$ 
then $\exists Q'$ with
  $Q \mathbin{\ar{a}} Q'$ and $P' \mathbin{\Sim_1} Q'\!$,
\item if $P \Sim_2 Q$ and $P \ar{R{:}\w} P'$ then $P' \Sim_2 Q$,
\item if $P \Sim_2 Q$ and $P \ar{\colonact{\range}{\starcastP{m}}} P'$ then $\exists Q'$ such that
  $Q \ar{b} Q'$ and $P' \Sim_1 Q'$,
  where $b$ is either $\broadcastP{\dval{m}}$, or $\groupcastP{\dval{D}}{\dval{m}}$ with
  $\dval{dsts}\subseteq D$, or $\unicast{\dval{dip}}{\dval{m}}$ with $\dval{dsts}=\{\dval{dip}\}$
  or $\neg \unicast{\dval{dip}}{\dval{m}}$ with $\dval{dsts}=\emptyset$,
\vspace{-1ex}
\end{itemize}
considering that other combinations of $P \Sim Q$ and $P \ar{a} P'$ cannot occur.
The first of these properties follows from \Prop{time determinism}.
The fourth follows from Rules~\tableref{tr} and~\tableref{tr-o} of \Tab{sos sequential} and
a trivial induction on the definition of $\Sim_2$.\linebreak[4]
Demonstrating the others proceeds by a straightforward induction on the derivation of \tawn transitions
from the rules of \Tab{sos sequential}.
\qed
\end{proof}
We have shown that each sequential AWN process $P$, seen as a \tawn process, can be simulated by the AWN process $P$.
We will now lift this result to parallel processes, node expressions, partial network expressions
and finally complete networks. This constitutes the proof of \Thm{network_simulation}.
\etoApp
\begin{theorem}
Each AWN process $P$, seen as a \tawn process, can be simulated by the AWN process $P$.
Likewise, each AWN network $N$, seen as a \tawn network, can be simulated by the AWN network $N$.
\label{thm:network_simulation}
\end{theorem}
\prfatend
%
Given weak simulations $\Sim_1$ and $\Sim_2$ of parallel \tawn processes by parallel AWN processes,
define the simulation ${\Sim_1} \parl {\Sim_2}$ of parallel \tawn processes by parallel AWN processes by
\[ P_1 \parl P_2 \mathrel{({\Sim_1} \parl {\Sim_2})} Q_1 \parl Q_2 ~~:\Leftrightarrow~~  P_1 \Sim_1 Q_1 \wedge P_2 \Sim_2 Q_2\;. \] 
By construction, this relation satisfies the second requirement of Definition~\ref{def:simulation}.
A straightforward induction on the derivation of \tawn transitions
from the rules of \Tab{sos parallel} shows that ${\Sim_1} \parl {\Sim_2}$ also satisfies the first
requirement, and thus is a weak simulation indeed. Now a trivial induction on the number of
sequential processes occurring in a parallel process, with \Thm{SP_simulation} as base
case and the above observation as induction step, lifts \Thm{SP_simulation} to parallel processes.

Given a weak simulation $\Sim$ of parallel \tawn processes by parallel AWN processes,
define the simulation $\Sim'$ of \tawn node expressions by AWN node expressions by
\[\dval{ip\/}{:}P{:}R ~\Sim~ \dval{ip}'{:}\hspace{.5pt}Q\hspace{.5pt}{:}R'  ~~:\Leftrightarrow~~
  \dval{ip}=\dval{ip}' \wedge P \Sim Q \wedge R=R'\;.\]
By construction, this relation satisfies the second requirement of Definition~\ref{def:simulation}.
By induction on the derivation of \tawn transitions from the rules of \Tab{sos node}
we show that ${\Sim_1} \parl {\Sim_2}$ also satisfies the first requirement, and thus is a weak
simulation. 
The only non-trivial cases are when
{$\dval{ip\/}{:}P{:}R \ar{\;\colonact{\range}{\starcastP{m}}\;} \dval{ip\/}{:}P'{:}R$} and
\plat{$\dval{ip\/}{:}P{:}R \ar{\;\;\tick\;\;} \dval{ip\/}{:}P'{:}R$}.
In the former case
 $P \ar{\;\;\colonact{\range}{\starcastP{m}}\;\;} P'$, so by induction
$Q \dar{b} Q'$ for some $Q'$ with $P' \Sim Q'$,
  where $b$ is either $\broadcastP{\dval{m}}$, $\groupcastP{\dval{D}}{\dval{m}}$ with
  $\dval{dsts}\subseteq D$, $\unicast{\dval{dip}}{\dval{m}}$ with
  $\dval{dsts}=\{\dval{dip}\}$ or\\
  $\neg \unicast{\dval{dip}}{\dval{m}}$ with
  $\dval{dsts}=\emptyset$.
By the rules of \cite[Table 3]{ESOP12},
\[\dval{ip\/}{:}\hspace{.5pt}Q\hspace{.5pt}{:}R \dar{}
  \dval{ip\/}{:}\hspace{.5pt}Q\hspace{.5pt}{:}\dval{dsts} \dar{\;\colonact{\range}{\starcastP{m}}\;}
  \dval{ip\/}{:}\hspace{.5pt}Q'{:}\dval{dsts} \dar{}
  \dval{ip\/}{:}\hspace{.5pt}Q'{:}R\;,\]
except in the case that $b=\neg \unicast{\dval{dip}}{\dval{m}}$, when we obtain
\[\dval{ip\/}{:}\hspace{.5pt}Q\hspace{.5pt}{:}R \dar{}
  \dval{ip\/}{:}\hspace{.5pt}Q\hspace{.5pt}{:}\hspace{.5pt}\emptyset \dar{\;\tau\;}
  \dval{ip\/}{:}\hspace{.5pt}Q'{:}\hspace{.5pt}\emptyset \dar{}
  \dval{ip\/}{:}\hspace{.5pt}Q'{:}R\;.\]
Here the derivations
$\dval{ip\/}{:}\hspace{.5pt}Q\hspace{.5pt}{:}R \dar{}  \dval{ip\/}{:}\hspace{.5pt}Q\hspace{.5pt}{:}\dval{dsts}$
and $\dval{ip\/}{:}\hspace{.5pt}Q'{:}\dval{dsts} \dar{} \dval{ip\/}{:}\hspace{.5pt}Q'{:}R$
consists of \textbf{(dis)connect}-actions---this is the reason these are seen as $\tau$'s here.
In case \plat{$\dval{ip\/}{:}P{:}R \ar{\tick} \dval{ip\/}{:}P'{:}R$}, we have 
$P \ar{w_1} P'$ or $P \ar{R\mathord{:}w_1} P'$, so by induction there is a $Q'$ with $Q \dar{} Q'$
and $P' \Sim Q'$. By the rules of \cite[Table 3]{ESOP12} we have
$\dval{ip\/}{:}\hspace{.5pt}Q\hspace{.5pt}{:}R \dar{}\dval{ip\/}{:}\hspace{.5pt}Q'{:}R$.

Given weak simulations $\Sim_1$ and $\Sim_2$ of \tawn partial network expressions by AWN partial
network expressions,
define the simulation ${\Sim_1} \| {\Sim_2}$ of \tawn partial network expressions by AWN partial
network expressions by
\[ N_1 \| N_2 \mathrel{({\Sim_1} \| {\Sim_2})} M_1 \| M_2 ~~:\Leftrightarrow~~  N_1 \Sim_1 M_1 \wedge N_2 \Sim_2 M_2\;. \] 
By construction, this relation satisfies the second requirement of Definition~\ref{def:simulation}.
A straightforward induction on the derivation of \tawn transitions
from the rules of \Tab{sos network} shows that ${\Sim_1} \| {\Sim_2}$ also satisfies the first
requirement, and thus is a weak simulation indeed. Now a trivial induction on the number of
node expressions occurring in a partial network expression lifts \Thm{SP_simulation} to partial
network expressions.


Given a weak simulation $\Sim$ of \tawn partial network expressions by AWN partial
network expressions,
define the simulation $\Sim'$ of complete \tawn networks by complete AWN networks by
\[[N] \Sim' [M]  ~~:\Leftrightarrow~~ N \Sim M\;.\]
By construction, this relation satisfies the second requirement of Definition~\ref{def:simulation}.
A straightforward induction on the derivation of \tawn transitions
from the rules of \Tab{sos network} shows that $\Sim'$ also satisfies the first
requirement, and thus is a weak simulation indeed.
%
\eprf

\noindent
Here a \emph{simulation} refers to a \emph{weak simulation} as
  defined in \cite{vG93}, but treating \textbf{(dis)connect}-actions
  as $\tau$, and with the extra requirement that the data states
  maintained by related expressions are identical---except of course
  for the variables $\now$, that are missing in AWN\@.
Details can be found in \SubApp{simulation}.

Thanks to \Thm{network_simulation}, we can prove that all invariants on the data structure
of a process expressed in AWN are still preserved when the process is
interpreted as a \tawn expression.
As an application of this, an untimed version of AODV\@, formalised as an AWN process, has been
  proven loop free in \cite{TR13,GHPT16}; the same system, seen as a \tawn expression---and thus with
  specific execution times associated to uni-, group-, and broadcast actions---is still loop free when given
  the operational semantics of \tawn.

\toApp
An immediate corollary of this result is that for each \tawn network expression $N'$, reachable from
an (initial) AWN network expression $N$, seen as \tawn network expression, there exists an AWN
network expression $N''$, reachable from $N$, such that $\exclnow[\ell(N')] = \ell(N'')$,
i.e., having the same global data state as $N'$, except of course for the variables $\now$.

This in turn implies that any invariant for the AWN network $N$---a property that holds for all
global data states of network expressions reachable from $N$---is also an invariant for $N$ seen as
a \tawn network.
\etoApp


\section{Case Study: The AODV Routing Protocol}
\label{sec:aodv}
Routing protocols are crucial to the dissemination of data packets between nodes in 
WMNs and MANETs.
Highly dynamic topologies are a key feature of WMNs and MANETs, due to
mobility of nodes and/or the variability of wireless links. This makes
the design and implementation of robust and efficient routing
protocols for these networks a challenging task.
In this section we present a formal specification of the Ad hoc On-Demand
Distance Vector (AODV) routing protocol.
AODV~\cite{rfc3561} is a widely-used routing protocol designed for
MANETs, and is one of the four protocols currently standardised by the
IETF MANET working group\footnote{\url{http://datatracker.ietf.org/wg/manet/charter/}}.
It also forms the basis of new WMN routing protocols, including
HWMP in the IEEE 802.11s wireless mesh network standard~\cite{IEEE80211s}.

Our formalisation is based on an untimed formalisation of AODV~\cite{TR13,GHPT16}, written in \awn, 
and models the exact details of the core
functionality of AODV as standardised in IETF RFC 3561~\cite{rfc3561}; e.g., route discovery, route maintenance and
error handling. We demonstrate how {\tawn} can be used to reason about
critical protocol properties. As major outcome we demonstrate that 
AODV is \emph{not} loop free, which is in contrast to common belief.
Loop freedom is a critical property for any routing protocol, 
but it is particularly relevant and challenging for WMNs and MANETs. 
We close the section by discussing a fix to the protocol and prove that 
the resulting protocol is indeed loop free.

\subsection{Brief Overview}

{
\newcommand{\gennode}[1]{\ensuremath{\MakeUppercase{#1}}\xspace}
\newcommand{\na}{\gennode{a}}
\newcommand{\nd}{\gennode{d}}
\newcommand{\ns}{\gennode{s}}
AODV is a reactive protocol, which means that routes are established only
on demand.  If a node~\ns wants to send a data packet to  a node \nd, but
currently does not know a route, it temporarily buffers the
packet and initiates a route discovery process
by broadcasting a route request (RREQ) message in the network. An intermediate
node \na that receives the RREQ message creates a routing table entry for a route
towards node~\ns referred to as a \emph{reverse route}, and re-broadcasts the RREQ. This is repeated until the RREQ reaches the destination node~\nd, or alternatively a node that knows a route to \nd. In both cases, the node replies by unicasting a corresponding route reply (RREP) message back to the source~\ns,
via a previously established reverse route. When forwarding RREP messages, nodes
create a \rte for node \nd, called the \emph{forward route}. When the RREP reaches
the originating node \ns, a route from \ns to \nd  is established and data packets
can start to flow. Both forward and reverse routes are maintained in a routing table at every node---details are given below.
In the event of link and route breaks, AODV uses route error (RERR) messages to
notify the affected nodes: if a link break is detected by a node, it first invalidates all routes stored in the node's own routing table that actually use the broken link. Then it sends a RERR message containing the unreachable destinations to all (direct) neighbours using this route. 

In AODV, a routing table consists of a list of entries---at most one for each destination---each containing the following information:
(i) the destination IP address;
(ii) the \emph{destination sequence number}; 
(iii) the sequence-number-status flag---tagging whether the recorded sequence number
can be trusted;
(iv) a flag tagging the route as being valid or invalid---this flag is set to invalid
when a link break is detected or the route's lifetime is reached;
(v) the hop count, a metric to indicate the distance to the destination;
(vi) the next hop, an IP address that
identifies the next (intermediate) node on the route to the destination;
(vii) a list of precursors, a set of IP addresses of those
  $1$-hop neighbours that use this particular route; and
(viii) the lifetime (expiration or deletion time) of the route.
The destination sequence number constitutes a measure approximating the
relative freshness of the information held---a higher number denotes newer information.  
The  routing table  is updated whenever a node receives an AODV control message (RREQ, RREP or RERR)
or detects a link break. 

During the lifetime of the network, each node not only maintains its routing table, it also stores
its {\em own} {\em sequence number}.  This number is used as a local ``timer'' and is incremented 
whenever a new route request is initiated. It is the source of the destination
sequence numbers in routing tables of other nodes.

\hspace{-5pt}Full details of the protocol are outlined in the request for comments (RFC)~\cite{rfc3561}.
}

\subsection{Route Request Handling Handled Formally}\label{sec:rreqformal}
Our formal model consists of seven processes: 
$\AODV$ reads a message from the message queue (modelled in process \QMSG, see below) and,
		depending on the type of the message, calls other processes.
         Each time a message has been handled
		the process has the choice between handling another message,
		initiating the transmission of queued data packets or generating
		a new route request.
$\NEWPKT$ and $ \PKT$ describe all actions
                performed by a node when a data packet is
                received.  The former process handles a newly injected
		packet. The latter describes all actions performed
		when a node receives data from another node via the protocol. 
$\RREQ$ models all events that might occur
                after a route request message has been received. 
Similarly, $\RREP$ describes the reaction
		of the protocol to an incoming route reply.
$\RERR$ models the part of AODV that handles error messages.
The last process $\QMSG$ queues incoming messages. Whenever a message is received,
	  it is first stored in a message queue. When the corresponding node is able to handle a message,
	  it pops the oldest message from the queue and handles it.
An AODV network is an encapsulated parallel composition of node expressions, each with a different
node address (identifier), and all initialised with the parallel composition 
$\AODV(\dots)\parl\QMSG(\dots)$.

Here we only present parts of the $\RREQ$ process, depicted in \Pro{rreq_shortened};
the full formal specification of the entire protocol can be found in \SubApp{spec}. There, we also discuss 
all differences between the untimed version of AODV, as formalised in~\cite{TR13,GHPT16}, and 
the newly developed timed version. 
These differences mostly consist of setting expiration times for routing table
  entries and other data maintained by AODV, and handling the expiration of this data.

{
\newcommand{\dsnf}{sequence-number-status flag\xspace}

\renewcommand{\ip}{\dval{ip}}
\renewcommand{\dip}{\dval{dip}}
\renewcommand{\oip}{\dval{oip}}
\renewcommand{\sip}{\dval{sip}}
\renewcommand{\rip}{\dval{rip}}
\renewcommand{\rt}{\dval{rt}}
\renewcommand{\r}{\dval{r}}
\renewcommand{\osn}{\dval{osn}}
\renewcommand{\dsn}{\dval{dsn}}
\renewcommand{\rsn}{\dval{rsn}}
\renewcommand{\flag}{\dval{flag}}
\renewcommand{\hops}{\dval{hops}}
\renewcommand{\nhip}{\dval{nhip}}
\renewcommand{\pre}{\dval{pre}}
\renewcommand{\dests}{\dval{dests}}
\renewcommand{\rreqid}{\dval{rreqid}}
\renewcommand{\rreqs}{\dval{rreqs}}
A route discovery in AODV is initiated by a source node broadcasting a RREQ message; this message
is subsequently re-broadcast by other nodes.
Process~\ref{pro:rreq_shortened} shows major parts of our process algebra specification for handling a RREQ message received
by a node \ip.
The incoming message carries eight parameters, including
{\hops}, indicating how far the RREQ had travelled so far,
{\rreqid}, an identifier for this request,
{\dip}, the destination IP address,
and {\sip}, the sender of the incoming message;
the parameters {\ip}, \dval{sn} and {\rt}, storing the node's address, sequence number
and routing table, as well as {\rreqs} and \dval{store}, are maintained by the process RREQ itself.

\hspace{-2.3pt}Before handling the incoming message, the process first updates
{\rreqs} (Line~\ref{rreq_short:line1}), a list of (unique) pairs containing the originator IP address {\oip} and a route request identifier {\rreqid} received within 
 the last \pathdiscoverytime: 
 the 
\pagebreak 
       \vspace*{-5ex}
 \noindent   
 \mbox{update} removes identifiers that are too old. 
 Based on this list, the node then checks whether it has recently received a RREQ with the same
   {\oip} and {\rreqid}.

If this is the case, the RREQ message is ignored, and the protocol continues to execute the main AODV process (Lines~\ref{rreq_short:line2}--\ref{rreq_short:line3}). 
If the
RREQ is new (Line~\ref{rreq_short:line4}), the process updates the routing table by adding a ``reverse route'' entry
to {\oip}, the originator of the RREQ, via node {\sip}, with distance \hops+1 (Line~\ref{rreq_short:line6}).
 If there
already is a route to {\oip} in the node's routing table {\rt}, it is only updated with the new
route if the new route is ``better'', i.e., fresher and/or shorter
and/or replacing an invalid route. 
The lifetime of this reverse route is updated as well (Line~\ref{rreq_short:line7}): it is set to the
maximum of the currently stored lifetime and the minimal lifetime, which is determined by 
$\now+2\cdot\traversal-2\cdot(\hops+1)\cdot\nodetraversal$ \cite[Page 17]{rfc3561}.
The process also adds the message to the list of known RREQs (Line~\ref{rreq_short:line8}).

Lines~\ref{rreq_short:line10}--22 (only shown in \SSubApp{ModelAODV}) deal with the case where the node receiving the RREQ is the
intended destination, i.e., {\dip}={\ip} (Line~\ref{rreq_short:line10}).
}

\setcounter{algorithm}{3}
  \algsetup{linenodelimiter=.,linenosize=\tiny}
  \begin{algorithm}[t]
    {\scriptsize
      \caption{Parts of the RREQ handling}
      \label{pro:rreq_shortened}
      \begin{algorithmic}[1]\scriptsize
{
\renewcommand{\highlightUPD}{\UPD}
\renewcommand{\highlight}[1]{#1}
\DEFPROCESS{\RREQ}{\hops\comma\rreqid\comma\dip\comma\dsn\comma\dsk\comma\oip\comma\osn\comma\sip\,\comma\,\ip\comma\sn\comma\rt\comma\rreqs\comma\queues}
\highlightUPD{\exprreqs{\rreqs}{\now}}\label{rreq_short:line1}
\PAR
	\IF[the RREQ has been received previously]{$(\oip\comma\rreqid\comma\highlight*)\in\rreqs$}																							\label{rreq_short:line2}
		\aodvL{\ip}{\sn}{\rt}{\rreqs}{\queues} \COM{silently ignore RREQ, i.e., do nothing}																			\label{rreq_short:line3}
	\ELSIF[the RREQ is new to this node]{$(\oip\comma\rreqid\comma\highlight*)\not\in\rreqs$}																								\label{rreq_short:line4}
		\UPD{\rt:=\upd{\rt}{(\oip,\osn,\kno,\val,\hops+1,\sip,\emptyset,\highlight{\now+\valtime})}}	
		\label{rreq_short:line6}
		\highlightUPD{\rt:=\settimert{\rt}{\oip}{\now+2\cdot\traversal-2\cdot(\hops+1)\cdot\nodetraversal}}\label{rreq_short:line7}%
		\UPD{\rreqs:=\rreqs\cup\{(\oip,\rreqid,\highlight{\now+\pathdiscoverytime})\}}		\COMMENT{update \rreqs}											\label{rreq_short:line8}
		\PAR																																															\label{rreq_short:line9}
		\IF[this node is the destination node]{$\dip=\ip$}																															\label{rreq_short:line10}
\algsetup{linenosize=\color{white}}
\STATE$[\dots]$\STATE

\makeatletter
\setcounter{ALC@line}{22}
\makeatother

		\ELSIF[this node is not the destination node]{$\dip\not=\ip$}																											\label{rreq_short:line20}
			\PAR																																															\label{rreq_short:line21}
			\COMLINE{valid route to \dip\ that is fresh enough}
			\IF{$\!\dip\mathbin\in\akD{\rt} \wedge \dsn \mathbin\leq  \sqn{\rt}{\!\dip} \wedge\sqnf{\rt}{\!\dip}\mathbin=\kno\!$}		\label{rreq_short:line22}
					\COMLINE{update \rt\ by adding precursors}																														\label{rreq_short:line23}
					\UPD{\rt := \addprecrt{\rt}{\dip}{\{\sip\}}}																																\label{rreq_short:line24}
					\UPD{\rt := \addprecrt{\rt}{\oip}{\{\nhop{\rt}{\dip}\}}}																												\label{rreq_short:line25}
				\COMLINE{unicast a RREP towards the \oip\ of the RREQ}
				\STARTPRIO
				    \textbf{unicast}(\nosp{\nhop{\rt}{\oip}}\comma
				    \\\hspace{23.2pt}\nosp{\rrep{\dhops{\rt}{\dip}}{\dip}{\sqn{\rt}{\dip}}{\oip}{\highlight{\selt{\rt}{\dip}-\now}}{\ip}}\,.

					\label{rreq_short:line26}%
					\aodvL{\ip}{\sn}{\rt}{\rreqs}{\queues}																																	\label{rreq_short:line26a}
				\PRIO
					\COMspec{If the transmission is unsuccessful, a RERR message is generated}\label{rreq_short:line26b}
					
\algsetup{linenosize=\color{white}\tiny}
\STATE$[\dots]$\COM{update local data structure}
\algsetup{linenosize=\color{black}\tiny}
\makeatletter
\setcounter{ALC@line}{39}
\makeatother
					\groupcast{\pre}{\rerr{\dests}{\ip}}\ . 																																	\label{rreq_short:line31}
					\aodv{\ip}{\sn}{\rt}{\rreqs}{\queues}
				\ENDPRIO
			\ELSIF[$\!$no fresh route$\!$]{$\dip\mathbin{\not\in}\akD{\rt} \vee \sqn{\rt}{\!\dip} <  \dsn \vee\sqnf{\rt}{\!\dip}\mathbin=\unkno$}					\label{rreq_short:line32}
				\COMLINE{no further update of \rt}
				\broadcast{\rreq{$\hops+1$}{\rreqid}{\dip}{\max(\sqn{\rt}{\dip}\comma\dsn)}{\dsk}{\oip}{\osn}{\ip}}\ .										\label{rreq_short:line34}
				\aodvL{\ip}{\sn}{\rt}{\rreqs}{\queues}				\label{rreq_short:line35}
			\ENDIFii
			\ENDPAR																																													\label{rreq_short:line36}
		\ENDIFii
		\ENDPAR																																														\label{rreq_short:line37}
	\ENDIFii
	\ENDPAR
}

	\end{algorithmic}
    }
  \end{algorithm}%

\setcounter{algorithm}{0}

{
\renewcommand{\ip}{\dval{ip}}%
\renewcommand{\dip}{\dval{dip}}%
\renewcommand{\oip}{\dval{oip}}%
\renewcommand{\sip}{\dval{sip}}%
\renewcommand{\rip}{\dval{rip}}%
\renewcommand{\rt}{\dval{rt}}%
\renewcommand{\r}{\dval{r}}%
\renewcommand{\osn}{\dval{osn}}%
\renewcommand{\dsn}{\dval{dsn}}%
\renewcommand{\rsn}{\dval{rsn}}%
\renewcommand{\flag}{\dval{flag}}%
\renewcommand{\hops}{\dval{hops}}%
\renewcommand{\nhip}{\dval{nhip}}%
\renewcommand{\pre}{\dval{pre}}%
\renewcommand{\dests}{\dval{dests}}%
\renewcommand{\rreqid}{\dval{rreqid}}%
\renewcommand{\rreqs}{\dval{rreqs}}%

Lines~\ref{rreq_short:line20}--\ref{rreq_short:line36} deal with the case where the node receiving the RREQ is
not the destination, i.e., ${\dip}\mathop{\neq}{\ip}$ (Line~\ref{rreq_short:line20}).  
The node can respond to the RREQ with a
corresponding RREP on behalf of the destination node \dip, if its route to {\dip} is ``fresh enough''
(Line~\ref{rreq_short:line22}).  This means that (a) the node has a valid route to \dip, (b) the destination sequence
\pagebreak
number in the node's current \rte $(\sqn{\rt}{\dip})$ is greater than or equal to the
requested sequence number to {\dip} in the RREQ message, and (c) the node's destination sequence number is 
trustworthy ($\sqnf{\rt}{\!\dip}\mathbin=\kno$).  If these three conditions are met
(Line~\ref{rreq_short:line22}), the node generates a RREP message, and unicasts it back to the originator
node {\oip} via the reverse route. 
Before
unicasting the RREP message, the intermediate node updates the forward {\rte} to
$\dip$ by placing the last hop node ($\sip$)
into the precursor list for
that entry (Line~\ref{rreq_short:line24}).
Likewise, it updates the reverse {\rte} to {\oip} by placing the first hop $\nhop{\rt}{\dip}$
towards $\dip$ in the precursor list for that entry
(Line~\ref{rreq_short:line25}).
To generate the RREP message, the process copies the sequence number for the
destination $\dip$ from the routing table $\rt$ into the destination sequence number field of the RREP message and it places its distance in hops from the destination ($\dhops{\rt}{\dip}$) in the corresponding field of the new reply
(Line~\ref{rreq_short:line26}).
 The RREP message is unicast to the next hop along the reverse route back to the originator of the
corresponding RREQ message. If this unicast is successful, the process goes back to the AODV routine (Line~\ref{rreq_short:line26a}).
If the unicast of the RREP fails, we proceed with Lines~\ref{rreq_short:line26b}--\ref{rreq_short:line31}, in which a route error
(RERR) message is generated and sent. This conditional unicast is implemented in our model with the \otawn construct 
$\unicast{\dexp{dest}}{\dexp{ms}}.P \prio{Q}$.
In the latter case, the node sends a RERR message to all nodes that rely on the broken link for one of their routes.
For this, the process first determines which destination nodes are affected by the broken link, i.e., the nodes
that have this unreachable node listed as a next hop in the routing table (not shown in the shortened specification).
Then, it invalidates any affected \rtes, and determines the list of
\emph{precursors}, which are the
neighbouring nodes that have a
route to one of the affected destination nodes via the broken link.
Finally, a RERR message is sent via groupcast to all these precursors
(Line~\ref{rreq_short:line31}).

If the node is not the destination and there is either no route to the
destination $\dip$ inside the routing table or the route is not fresh enough,
the route request received has to be forwarded. This happens in Line~\ref{rreq_short:line34}.
The information inside the forwarded request is mostly copied from the request received.
Only the hop count is increased by $1$ and the destination sequence number is set
to the maximum of the destination sequence number in the RREQ packet
and the current sequence number for $\dip$ in the routing table.
In case $\dip$ is an unknown destination, $\sqn{\rt}{\dip}$ returns the
unknown sequence number $0$.
}

To ensure that our time-free model from \cite{TR13,GHPT16} accurately captures the intended behaviour of AODV \cite{rfc3561}, 
we spent a long time reading and interpreting the RFC, inspecting open-source
implementations, and consulting network engineers.
We now prove that our timed version of AODV behaves similar to our original formal specification, 
and hence (still) captures the intended behaviour.
 
\begin{theorem}\label{thm:timedVSuntimed}
The timed version of AODV (as presented in this paper) is a proper extension of 
the untimed version (as presented in~\cite{TR13}). By this we mean that if 
all timing constants, such as \valtime\hspace{-1pt}, are set to $\infty$\hspace{-1pt},
and the maximal number of pending route request retries 
\retries\ is set to $1$, then the (\tawn) transition systems of both versions of AODV are 
weakly bisimilar.
\end{theorem}
\begin{proofsketch}
First, one shows that the newly introduced functions, such as \fnexprreqs\ and \fnsettimert\  do not change the data state in case the time parameters equal $\infty$; 
and hence lead to transitions of the form \plat{$\xi,p\ar{\tau}\xi,p'$}. This kind of transitions are the ones that make the bisimulation weak, since they do not occur 
in the formal specification of~\cite{TR13}.
\journalonly{Then one proves that, in case $\retries=1$, the function $\fnunsetrrf$ and $\fnincretries$, as well as $\fnsetrrf$ and $\fnsetretries$
behave equivalently, when identifying $1$ with $\pen$ and $0$ with $\nonpen$.}
Subsequently, one proves that all other transitions are basically identical.
\end{proofsketch}

\subsection{Loop Freedom}\label{ssec:lf}
Loop freedom is a critical property for any routing protocol, but it is particularly relevant and
challenging for WMNs and MANETs.
``{A routing-table loop is a path specified in the nodes' routing tables at a particular point in time that visits the same node more than once before reaching the intended destination}"~\cite{Garcia-Luna-Aceves89}.
Packets caught in a routing loop can quickly saturate the links and have a detrimental impact on network performance.

For AODV and many other protocols sequence numbers are used to guarantee loop freedom.
Such  protocols usually claim to be loop free due to the use of monotonically increasing sequence numbers.
For example, AODV ``{uses destination sequence numbers to ensure loop freedom at all
 times (even in the face of anomalous delivery of routing control messages), ...}"~\cite{rfc3561}. 
 It has been shown that sequence numbers do not a  priori guarantee loop freedom~\cite{AODVloop}; 
 for some plausible interpretations%
 \footnote{By a plausible interpretation of a protocol standard written in English prose we mean an interpretation that fills the missing bits, and 
 resolves ambiguities and contradictions occurring in the standard in a sensible and meaningful way.}
 of different versions of AODV, however, loop freedom has been
 proven~\cite{AODV99,BOG02,ZYZW09,SRS10,MSWIM12,TR13,GHPT16,NT15}%
\footnote{%
 The proofs in \cite{AODV99} and \cite{BOG02} are incorrect;
 the model of \cite{SRS10} does not capture the full behaviour of the routing protocol; 
 and \cite{ZYZW09} is based on a subset of AODV that
 does not cover the ``intermediate route reply'' feature,
 a source of loops.
In \cite{NT15} a draft of a new version of AODV is modelled, without intermediate route reply.
 For a more detailed discussion see~\cite{GHPT16}.}.
With the exception of~\cite{BOG02}, all these papers consider only untimed versions of AODV. 
As mentioned in \Sect{intro} untimed analyses revealed many shortcomings of AODV;
hence they are necessary. At the same time, a timed analysis is required as well.
\cite{BOG02} shows that the premature deletion of invalid routes, and a too quick restart of a node after a reboot,
can yield routing loops.
Since then, AODV has changed to such a degree that the examples of \cite{BOG02} do not apply any longer. 

In \cite{Garcia04}, ``it is shown that the use of a $\deltime$ in  the current AODV specification can result in loops''.
However, the loop constructed therein at any time passes through at least one invalid routing table entry.
As such, it is not a routing loop in the sense of \cite{TR13,GHPT16}---we only consider loops
consisting of valid routing table entries, since invalid ones do not forward data
packets. In a loop as in \cite{Garcia04} data packets cannot be sent in circles forever.

It turns out that AODV as standardised in the RFC (and carefully formalised  in \Sect{rreqformal} and \SubApp{spec})
is \emph{not} loop free.
A potential cause of routing loops, sketched in \Fig{loop1}, is a situation where a
node $B$ has a {\val}id 
\begin{wrapfigure}[7]{l}{0.55\textwidth}
	\scriptsize
		\begin{tabular}{ll@{\qquad}ll}
			\multicolumn{2}{l}{\ \ \ \includegraphics{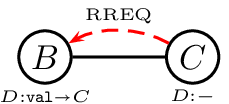}}   &   
			\multicolumn{2}{c}{\includegraphics{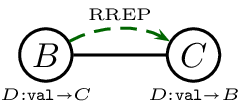}}\\	[5pt]
			a)&$C$'s entry for $D$ has  &b)& $B$ sends old data;\\
                 &expired; $B$'s has not    &&     $C$ establishes loop\\[-2ex]
		\end{tabular}
		\caption{Premature Route Expiration\label{fig:loop1}}
\end{wrapfigure}%
routing table entry for a destination $D$ (in \Fig{loop1} denoted $D{:}\val{\raisebox{.8pt}{$\scriptstyle\rightarrow$}}C$),
but the next hop $C$ no longer has a routing table entry for $D$ ($D{:}-$), valid or invalid.
In such a case, $C$ might search for a new route to $D$ and create a new
routing table entry pointing to $B$ as next hop, or to a node $A$ upstream
from $B$. We refer to this scenario as a case of \emph{premature route expiration}.

A related scenario, which we  also call premature route expiration, is when a node $C$ sends a RREP message with destination~$D$
or a RREQ messages with originator $D$ to a node $B$, but looses its route to $D$ before that
message arrives. This scenario can easily give rise to the scenario above. 

Premature route expiration can be avoided by setting {\deltime} to $\infty$, 
which is essentially the case in the untimed version of AODV (cf.\ \Thm{timedVSuntimed}).
In that case, no routing table entry expires or is erased. 
\journalonly{routes are invalided after the receipt of route error messages only.}
Hence,
the situation where $C$ no longer has a routing table entry for $D$ is prevented.

In \cite{TR13} we studied 5184 possible interpretations of the AODV RFC \cite{rfc3561},
a proliferation due to ambiguities, contradictions and cases of
underspecification that could be resolved in multiple ways.
In 5006 of these readings of the standard, including some rather
plausible ones, we found routing loops, even when excluding all loops
that are due to timing issues \cite{AODVloop,TR13}. In \cite{MSWIM12,TR13,GHPT16}
we have chosen a default reading of the RFC that avoids these loops,
formalised it in AWN, and formally proved loop freedom, still
assuming (implicitly) $\deltime = \infty$.

After taking this hurdle, the present paper continues the investigation by 
allowing arbitrary values for time parameters and for \retries; 
hence dropping the simplifying assumption that $\deltime = \infty$.

One of our key results is that for the formalisation of AODV presented
here, premature route expiration is the \emph{only} potential source
of routing loops. Under the assumption that premature route expiration
does not occur, it turns out that, with minor modifications, the 
loop freedom proof of \cite{TR13,GHPT16} applies to our timed model of AODV
as well. A proof of this result is presented in \SubApp{invariants}.
There, we revisit all the invariants from \cite{TR13} that contribute
to the loop-freedom proof, 
and determine which of them are still valid
in the timed setting, and how others need to be modified.

It is trivial to find an example where premature route expiration does
occur in AODV, and a routing loop ensues. This can happen when a
message spends an inordinate amount of time in the queue of incoming
messages of a node. However, this situation tends not to occur in
realistic scenarios. To capture this, we now make the assumption that 
the transmission time of a message plus
the period it spends in the queue of incoming messages of the
receiving node is bounded by \nodetraversal. We also assume that the
period a route request travels through the network is bounded by \traversal.

These assumptions eliminate the ``trivial'' counterexample mentioned above.
As we show in \SubApp{D}, we now \emph{almost} can prove an invariant that
essentially says that premature route expiration does not occur.
Following the methodology from \cite{MSWIM12,TR13,GHPT16}, we establish our
invariants by showing that they hold in all possible initial states of
AODV, and are preserved under the transitions of our operation
semantics%
, which correspond to the line numbers in our process algebraic specification.

We said ``almost'', because, as indicated in
\SubApp{D}, our main invariant is not preserved by
five lines of our AODV specification.
Additionally, we need to make the assumption that when a RREQ message is forwarded, the forwarding node has
a valid routing table entry to the originator of the route request.
This does not hold for our formalisation of AODV: in \Pro{rreq_shortened} no check 
is performed on $\oip$, only the routing table to the destination node $\dip$ has to satisfy certain conditions
(Lines \ref{rreq_short:line20} and \ref{rreq_short:line32}).

It turns out that for each of these failures we can
construct an example of premature route expiration, and, by that, a
counterexample to loop freedom.

However, if we skip all five offending lines (or adapt them in
appropriate ways) and make a small change to process RREQ that makes
the above assumption valid,%
\footnote{The change basically introduces the test ``$\oip\in\akD{\rt}$'' in
  Line~\ref{rreq_short:line32} or~\ref{rreq_short:line9} of \Pro{rreq_shortened}.\label{repair}}
we obtain a proof of loop freedom for the resulting version of AODV\@. This follows immediately from
the invariants established in \SubApp{D}.

\section{Conclusion }
In this paper we have proposed \tawn, a timed process algebra for wireless networks.
We are aware that there are many other timed process algebras, such as 
timed CCS~\cite{MT90}, 
timed CSP~\cite{RR86,OS06}, 
timed ACP~\cite{BB96},
ATP \cite{NS94} and TPL \cite{HR95}.
However, none of these algebras provides the unique set of features needed 
for modelling and analysing protocols for wireless networks (e.g. a conditional unicast).%
\footnote{This is similar to the untimed situation. A detailed comparison between \awn and other process calculi for 
wireless networks is given in \cite[Section 11.1]{TR13}; this discussion can directly be transferred to the timed case.}
These features are provided by \otawn, though. 
Our treatment of time is based on design decisions that appear
  rather different from the ones in \cite{MT90,RR86,OS06,BB96,NS94}.
  Our approach appears to be closest to \cite{HR95}, but
  avoiding the negative premises that play a crucial role in the
  operational semantics of~\cite{HR95}.

We have illustrated the usefulness of \tawn by analysing 
the  Ad hoc On-Demand Distance
Vector routing protocol, and have shown that, contrary to claims in the
literature and to common belief, it fails to be loop free. We have also discussed
boundary conditions for a fix ensuring
that the resulting protocol is loop free.

\paragraph*{Acknowledgement.}
NICTA is funded by the Australian Government through the Department of Communications and the Australian Research Council through the ICT Centre of Excellence Program. 

\newpage
\bibliographystyle{eptcsini}
\bibliography{aodv}

\newpage
\appendix
\section*{Appendices}
\setcounter{secnumdepth}{3}
\section{Results on the Process Algebra\label{app:pa}}
\subsection{Deferred Proofs and Auxiliary Lemmas\label{subapp:proofs}}
    \printproofs
 
\section{Case Study: The AODV Routing Protocol}
\subsection{Formal Specification of AODV}\label{subapp:spec}
This appendix provides a complete and accurate formal specification of the AODV routing
protocol, as defined in IETF RFC 3561~\cite{rfc3561}. The presented formalisation is based on 
an untimed model formalised in AWN~\cite{TR13,GHPT16}, and includes all core
components of the protocol, but abstracts from optional protocol features.%
\footnote{A list and discussion of all omitted features occurs in \cite[Section 3]{TR13}.}
The only difference with our previous formalisation of AODV in \cite{TR13,GHPT16} is the inclusion of
timing issues.

To keep this appendix `short', we focus on the difference between the two models; 
an interested reader can either study the formal model on her own, or have a look 
at \cite{TR13,GHPT16}, where we explain each and every line of the specification. 

\subsubsection{Data Structure\label{ssubapp:data}}~

\noindent In \cite[Section 5]{TR13} we define the basic data structure needed for the detailed formal
specification of AODV, when abstracting from timing issues. Here we merely list the changes in this
data structure needed to deal with time.

\paragraph{\rm\bf Constants and Basic Types.}
The new type $\tTIME$ and the variable $\now$ have been introduced in \Sect{process_algebra}.
Additionally, we use the following constants of type $\tTIME$.
They all follow the RFC; their default values are found in~\cite[Section 10]{rfc3561}.
\begin{description}
\label{pg:time_constants}
\item[\deltime:] the lifetime of an invalid routing table entry;
\item[\valtime:] the time after which a valid entry is invalidated;
\item[\myroute:] amount of time entered as last parameter of a route reply issued by the destination
  node of a route request---to be used as the time during which the route to the destination created
  by that route reply remains valid;
\item[\nodetraversal:]  a conservative estimate of the average one hop traversal time for packets---it
  should include queueing delays and transfer times;
\item[\traversal:] a conservative estimate on the time it takes for a message to travel from one end
  of the network to the other and back---calculated as $2\cdot\nodetraversal\cdot\keyw{NET\_DIAMETER}$;
\item[\pathdiscoverytime:] time during which identifiers of handled route requests are kept.
\end{description}

The type $\tPendingRREQ$, which, in the original specification, was a Boolean flag indicating whether a new route request
  needs to be initiated, is now of type \NN; it tells the number of pending route request. 
  The constants $\pen$ and $\nonpen$ do not exist any more, but there is a new constant of type $\tPendingRREQ$,
  discussed in~\cite[Sections 6.3 and 10]{rfc3561}:
\begin{description}
\item[\retries:] maximal number of retries for a route discovery process.
\end{description}

Routing table entries are of type \tROUTE\ and  have an additional parameter, their expiration time, which in the RFC is called Lifetime.
An entry is now an eight-tuple of type 
\[\tIP \times \tSQN \times\tSQNK \times \tFLAG\times \NN \times \tIP \times \pow(\tIP)\times\tTIME\]

A tuple $(\dval{dip}\comma\dval{dsn}\comma\dval{dsk}\comma\dval{flag}\comma\dval{hops}\comma\dval{nhip}\comma\dval{pre}\comma\dval{ltime})$
describes a route to $\dval{dip}$ of length $\dval{hops}$ and validity
$\dval{flag}$; the very next node on this route is $\dval{nhip}$; the
last time the entry was updated the destination sequence number was
$\dval{dsn}$; \dval{dsk} denotes whether the sequence number
is ``outdated'' or can be used to reason about the freshness of the route.
$\dval{pre}$ is a set of all neighbours who are
``interested'' in the route to $\dval{dip}$.  
Finally, $\dval{ltime}$ states the expiration time of the route; 
if this time is reached, a valid route will be set to invalid, and an invalid route will 
be deleted from the routing table.
We use projections $\pi_{1},\dots\pi_{8}$ to select the
corresponding component from the $8$-tuple: for example,
$\pi_6:\tROUTE\to\tIP$ determines the next hop, and $\pi_8:\tROUTE\to\tIP$ distills the expiration time.
A routing table is an element of type \tRT\ and defined as a set of entries, with the restriction
that each has a different destination $\dip$, i.e., the first
component of each entry in a routing table is unique.

The data type $\tQUEUES$, which models a set of data queues for
injected data packets, is equipped with timers as well.
Each queue now contains a timer that
 indicates when a new/the next route request should be sent. 
The special value $0$ means ``to be sent immediately''.
\[\begin{array}{r@{\hspace{0.5em}}c@{\hspace{0.5em}}l}	
\tQUEUES &:=& \left\{\dval{store}\,\left|
\begin{array}{@{~}l@{}}\dval{store}\in\pow(\tIP\times\tPendingRREQ\times\tTIME\times[\tDATA])\ans\mbox{}\\
\big((\dval{dip},p,t,q),(\dval{dip},p',t',q')\in\dval{store} \Rightarrow\\\hfill p=p' \wedge t=t'\wedge q=q'\big)
\end{array}\right.\right\}
\end{array}\]
Here ${[\tDATA]}$ denotes a queue of elements from $\tDATA$.

The last type that needs to be changed is the set of 
pairs $(\oip, \rreqid)\in\tIP\times\tRREQID$, which uniquely identify route
requests. (For our specification we set $\tRREQID=\NN$.)
As such pairs are stored by nodes for a limited amount of time,
we add a third component to indicate when the pair can be dropped.
In our model, each node maintains a variable \keyw{rreqs} of type
\[
		\pow(\tIP\times\tRREQID\times\tTIME)
\]
of sets of such triples to store the set of route requests seen by the node so far.

\renewcommand{\fnsetrrf}{\keyw{setRRF}}
\renewcommand{\fnunsetrrf}{\keyw{unsetRRF}}

\paragraph{\rm\bf Functions.}
A brief overview of all functions used in our specifications can be found in 
Table~\ref{tab:types}. Here, we only list changes w.r.t. untimed formal specification of AODV. 

First, we need to change/update a couple of functions. This is mainly due to the new and changed data types.
For example, the function $\fnqD:\tQUEUES\to\pow(\tIP)$, which extracts the destinations for which there are unsent packets
is changed from $\{\dip\mid(\dip,*,*)\in\queues\}$ to $\{\dip\mid(\dip,*,*,*)\in\queues\}$.
In a similar (straightforward) manner the functions 
$\fnadd$,
$\fndrop$,
$\fnselqueue[\footnotesize]$,
$\fnakD$, $\fnikD$, $\fnkD$
$\fnaddprec$,
$\fnaddprecrt$, and
$\fnnrreqid$ 
are adapted.

In~\cite{TR13,GHPT16} the functions 
$\fnsqn$,
$\fnsqnf$,
$\fnstatus$,
$\fndhops$,
$\fnnhop$,
$\fnprecs$, and
$\fnprecs$
distill particular information for a specified route in the routing table (if it exists). 
Since routing table entries are extended with an additional field, we define a new function 
$\fnltime$ that selects the newly introduced expiration time:
\[\begin{array}{r@{\hspace{0.5em}}c@{\hspace{0.5em}}l}
              \fnltime :\ \tRT\times\tIP&\rightharpoonup& \tTIME\\
	 \ltime{\dval{rt}}{\dval{dip}}&:=&
	    \left\{
	    \begin{array}{l@{\qquad}l}
	        \pi_{8}(\dval{r}) &\mbox{if }\dval{r}\in\dval{rt}\wedge \pi_{1}(\dval{r}) = \dval{dip}\\
	        \mbox{undefined}&\mbox{otherwise\ .}
	    \end{array}\right.
	    \end{array}\]

Next to these changes, we now discuss changes in a few functions that either are non-trivial or of particular interest for 
the timed version of AODV.

{
\hypertarget{invalidate}{\emph{Invalidating routes} is a main feature of AODV; if a route is not valid
any longer, its validity flag has to be set to invalid. By doing
this, the stored information about the route, such as the sequence number and the hop count,
remains accessible.  
The function for invalidating routing table entries takes as arguments a routing
table, a set of destinations $\dval{dests}\in\pow(\tIP\times\tSQN)$,
and the expiration time for the newly invalidated routes.
Elements of {\dval{dests}} are $(\dval{rip},\dval{rsn})$-pairs that not only identify an
unreachable destination $\dval{rip}$, but also a sequence number that
describes the freshness of the faulty route.  We restrict ourselves to sets that have at most one entry for each
destination---formally we define {\dval{dests}} as a \emph{partial
function} from $\tIP$ to $\tSQN$, i.e., a subset of $\tIP\times\tSQN$
satisfying 
$(\dval{rip},\dval{rsn}),(\dval{rip},\dval{rsn}')\in \dval{dests} \Rightarrow \dval{rsn}=\dval{rsn}'\ .
$}
\[\begin{array}{@{}r@{\ \ }c@{\ \ }l@{}}
\multicolumn{3}{l}{\fninv : \tRT\times(\tIP\rightharpoonup\tSQN)\times \tTIME \to\tRT}\\
\inv{\dval{rt}}{\dval{dests}}{\dval{t}}&:=& \{\route\,|\,\route\in\rt\ans (\pi_{1}(\route),*)\not\in\dval{dests}\}\\
&\cup&\{(\pi_{1}(\route),\dval{rsn},\pi_{3}(\route),\inval,\pi_{5}(\route),\pi_{6}(\route),\pi_{7}(\route),\dval{t})\mid\\
&&\route\in\rt\ans(\pi_{1}(r),\dval{rsn})\in\dval{dests}\}
\end{array}\]

Similar to invalidate, \emph{updating a routing table} must take the expiration time of a route into account.
The update function now works on $8$-tuples as routing table entries,
and the new expiration time of a route is taken as the maximum of the one from the routing table (if any) and the one
from the incoming route, but only if the route is actually updated with new important information. 
This is in line with the RFC, which updates a route's expiration time
to the maximum of the {\tt ExistingLifetime} and the {\tt MinimalLifetime}. In AODV the 
minimal expiration time is often set to $\now + \valtime$.
As in \cite{TR13,GHPT16} we define an update function $\upd{\dval{rt}}{\dval{r}}$ of a routing table 
$\dval{rt}$ with an entry $\dval{r}$ only when
 $\dval{r}$ is valid, i.e., $\pi_{4}(\dval{r})=\val$,
$\pi_{2}(\dval{r})=0\Leftrightarrow\pi_{3}(\dval{r})=\unkno$,
and $\pi_{3}(\dval{r})=\unkno\Rightarrow\pi_{5}(\dval{r})=1$.
\hypertarget{update}{Formally the function $\fnupd:\tRT\times\tROUTE \rightharpoonup \tRT$ is given~by}
{
\newcommand{\nrt}{\dval{nrt}}
\newcommand{\nr}{\dval{nr}}
\newcommand{\ns}{\dval{ns}}
\newcommand{\s}{\dval{s}}
\renewcommand{\r}{\dval{r}}
\[
\upd{\dval{rt}}{\r}:= \left\{
\begin{array}{@{\,}ll@{}l@{}}
\dval{rt}\cup\{\r\} & \mbox{if }  \pi_{1}(\r)\not\in\kD{\dval{rt}}\\[1mm]
\nrt\cup\{\nr\}&\mbox{if }  \pi_{1}(\r)\in\kD{\dval{rt}} &\ans  \sqn{\dval{rt}}{\pi_{1}(\r)}<\pi_{2}(\r)\\[1mm]
\nrt\cup\{\nr\}&\mbox{if }  \pi_{1}(\r)\in\kD{\dval{rt}} &\ans \sqn{\dval{rt}}{\pi_{1}(\r)}=\pi_{2}(\r) \\
 																					&&\ans \dhops{\dval{rt}}{\pi_{1}(\r)}>\pi_{5}(\r)\\[1mm]
\nrt\cup\{\nr\}&\mbox{if }  \pi_{1}(\r)\in\kD{\dval{rt}} &\ans \sqn{\dval{rt}}{\pi_{1}(\r)}=\pi_{2}(\r) \\
 																					&&\ans\status{\dval{rt}}{\pi_{1}(\r)}=\inval\\[1mm]
\nrt\cup\{\nr'\}&\mbox{if } \pi_{1}(\r)\in\kD{\dval{rt}} &\ans  \pi_3(\r)=\unkno\\[1mm]
\nrt\cup\{\ns\}&\mbox{otherwise\ ,}
\end{array}
\right.\]
where $\s:=\selr[\footnotesize]{\dval{rt}}{\pi_{1}(\r)}$ is the current entry in the
routing table for the destination of $\route$ (if it exists), and
$\nrt := \dval{rt} -\{\s\}$ is the routing table without that entry.
The entry 
\[
	\nr:=\left(\pi_{1}(\r), \pi_{2}(\r), \pi_{3}(\r), \pi_{4}(\r), \pi_{5}(\r), \pi_{6}(\r), \pi_{7}(\r)\cup\pi_{7}(\s),\max(\pi_{8}(\r),\pi_{8}(\s))\right)
	\] 
is identical to~$\r$ except
that the precursors from $\s$ are added and the lifetime is set to the maximum of the routes {\r} and {\s}.
Similarly, $\ns:=\addprec{\s}{\pi_{7}(\r)} = (\pi_{1}(\s), \pi_{2}(\s), \pi_{3}(\s), \pi_{4}(\s),$ $\pi_{5}(\s), \pi_{6}(\s), \pi_{7}(\s)\cup\pi_{7}(\r),\pi_{8}(\s))$
is generated from $\s$ by adding the precursors from $\r$; the
	lifetime, however, is \emph{not} updated.\footnote{We could have
	updated the expiration time to $\max(\pi_{8}(\r),\pi_{8}(\s))$; our results
	on loop freedom are not affected by this choice.}
Lastly, 
$\nr':=\left(\pi_{1}(\r), \pi_{2}(\s), \pi_{3}(\r), \pi_{4}(\r), \pi_{5}(\r), \pi_{6}(\r), \pi_{7}(\r)\cup\pi_{7}(\s),\max(\pi_{8}(\r),\pi_{8}(\s))\right)$
is identical to $\nr$ except that the sequence number is replaced by the one from 
the route $s$. 
} 

One of the AODV control messages need{s} to be modified as well: the route reply. 
It is the only message type that carries, according to the RFC, a time parameter.
It specifies the time for which nodes receiving the RREP message consider the route to be valid. 
The function that generates a RREP message has the form 
$\rrepID:\NN \times \tIP \times \tSQN \times \tIP \times  \tTIME\times\tIP\rightarrow \tMSG$.

Since $\tPendingRREQ$ is not a Boolean flag anymore,
but of type $\NN$, the functions
  $\fnunsetrrf$ and $\fnsetrrf$ (for updating the request-required flag) are replaced by the
  functions $\fnincretries$ and $\fnsetretries$, respectively.

The function $\fnincretries$ increments the number of pending requests:
\[\begin{array}{l}
\fnincretries:\tQUEUES\times\tIP\to\tQUEUES\\
\incretries{\dval{store}}{\dval{dip}}:=\left\{
\begin{array}{ll@{}}
  \dval{store}-\{(\dval{dip},n,t,q)\}\cup\{(\dval{dip},n+1,t,q)\}\hspace{-9em}
      \\&\mbox{if } (\dval{dip},n,t,q)\in\dval{store}\\
  \dval{store}&\mbox{otherwise}\ .
\end{array}\right.
\end{array}\]

The function $\fnsetretries$ resets the number
of pending requests (to $0$):
\[\begin{array}{l}
\fnsetretries:\tQUEUES\times(\tIP\rightharpoonup\tSQN)\to\tQUEUES\\
\setretries{\dval{store}}{\dval{dests}}:=
\begin{array}[t]{l}
   \{st\,|\,st\in \dval{store} \ans (\pi_{1}(st),*)\not\in\dval{dests}\}\\
   \mathrel\cup\{(\pi_{1}(st),0,0,\pi_{4}(st))\mid\\
   \hspace{15pt}st\in\dval{store}\ans(\pi_{1}(st),*)\in\dval{dests}\}\;.\\
\end{array}\end{array}\]
It also resets the waiting time before a new route request may be scheduled.

We define two new (partial) functions that extract the number of
route requests already initiated for a particular destination,
and the time one has to wait before a new route request may be scheduled, respectively:
\[\begin{array}{l}
\fnfD[\footnotesize]:\tQUEUES\times\tIP\rightharpoonup\tPendingRREQ\\
\fD[\footnotesize]{\dval{store}}{\dval{dip}}:=\left\{
\begin{array}{l@{\quad}l}
  p&\mbox{if } (\dval{dip},p,*,*)\in\dval{store}\\
  \mbox{undefined}&\mbox{otherwise}\ ,
\end{array}\right.
\end{array}\]
\[\begin{array}{l}
\fnseltime[\footnotesize]:\tQUEUES\times\tIP\rightharpoonup\tTIME\\
\selt[\footnotesize]{\dval{store}}{\dval{dip}}:=\left\{
\begin{array}{l@{\quad}l}
  t&\mbox{if } (\dval{dip},*,t,*)\in\dval{store}\\
  \mbox{undefined}&\mbox{otherwise}\ .
\end{array}\right.
\end{array}\]

Finally, to cope with the newly introduced expiration times (lifetimes), we 
define new functions for modifying the routing tables and other data structures. 

\hypertarget{exprt}{A (valid) route that has expired, has to be marked as invalid; 
and an expired invalid route has to be removed from the routing table.\pagebreak[1]
The function $\fnexprt$ models this behaviour:}
\[\begin{array}{ll}
	      \multicolumn{2}{l}{\fnexprt :\tRT\times\tTIME\times\tTIME\to \tRT}\\
	    \exprt{\dval{rt}}{\dval{t}}{\dval{t}'}:=& \{\route\mid \route\in\rt\wedge \pi_{8}(\route) > \dval{t}
            \wedge \keyw{1hoplife}(\pi_6(r),\dval{t}) \}\ \cup \\
	     & \{(\dval{dip},\inc{\dval{dsn}},\dval{dsk},\inval,\dval{hops},\dval{nhip},\dval{pre},\dval{lifetime}+\dval{t}')\mid\\
	     &\qquad(\dval{dip},\dval{dsn},\dval{dsk},\val,\dval{hops},\dval{nhip},\dval{pre},\dval{lifetime})\in\rt\\
	     &\qquad \mbox{}\wedge (\dval{lifetime} \leq \dval{t} \vee \neg\keyw{1hoplife}(\dval{nhip},\dval{t}))\\
             &\qquad \mbox{}\wedge \dval{lifetime}+t' > \dval{t} \}\ .
\end{array} \]%
\journalonly{change {\lifetime} to t; discuss differences; still loop free? seems to be an ambiguity in the RFC}%
\hypertarget{1hoplife}{Here $\keyw{1hoplife}(\dval{nhip},\dval{t})$ is a shorthand for
\[(\dval{nhip},*,*,\val,1,*,*,\dval{ltime})\in\rt \Rightarrow \dval{ltime} > \dval{t}\;;\]
it says that if there is a valid routing table entry for node $\dval{nhip}$ with hop count $1$
(in the routing table), then it is not yet expired.}
The first set keeps all routing table entries that have not expired at time \dval{t}.
Here we take into account two ways a routing table entry $r$ can expire: when the (current) time \dval{t}
equals or exceeds its expiration time $\pi_8(r)$, or when the 1-hop routing table entry to its next
hop expires \cite[Section 6.1]{rfc3561}.
The second set selects all expired valid routes and marks them as
invalid, thereby incrementing the destination sequence number; 
it also sets a new expiration time to indicate when the entry should be removed.
In the (rare) case that even the new expiration time counts as expired, the entry is dropped.
Expired invalid routes are not added to the created set, and are hence erased.

In applications we take $t=\now$ and $t'=\deltime$.
In case an entry is invalidated, the new expiration time is set to be {\deltime} after the previous
expiration time. So valid entries with $\dval{lifetime}+\deltime\leq\now$ skip the phase of being
invalid and are erased right away.

\hypertarget{exprreqs}{Similar to $\fnexprt$ we define a function that modifies the set of route request identifiers
  by expunging the expired ones.}
\[\begin{array}{l}
	      \fnexprreqs :\pow(\tIP\times\tRREQID\times\tTIME)\times\tTIME\to \pow(\tIP\times\tRREQID\times\tTIME)\\
	    \exprreqs{\dval{rreqs}}{\dval{t}}:= \{\dval{rq}\mid \dval{rq}\in\dval{rreqs}\wedge \pi_{3}(\dval{rq}) > \dval{t}\}\footnotemark\ .
\end{array} \]
\footnotetext{Projections on route requests  identifiers are defined as usual. Here this means that $\pi_3:\tIP\times\tRREQID\times\tTIME\to\tTIME$ determines the expiration time of the triple.}

 In the same vain, we introduce a function that drops all packets enqueued for destinations
  that have $\retries$ pending route requests, and for which the waiting period has expired.
 This means that no further route request will be sent, and hence the packets will not be delivered.
\[\begin{array}{l}
	      \fnexpq :\tQUEUES\times\tTIME\to \tQUEUES\\
         \expq{\dval{store}}{t'} := \{(\dval{dip},p,t,*) \in\dval{store} \mid p < \retries \vee t>t'\}
\end{array} \]

 \hypertarget{settimert}{Last, but not least, we introduce two
 functions to update the expiration times in routing tables and in stores, respectively.}
\[\begin{array}{l}
	      \fnsettimert :\tRT\times\tIP\times\tTIME\to \tRT\\
	    \settimert{\dval{rt}}{\dval{dip}}{\dval{t}}:= 
	    \left\{\begin{array}{@{\,}ll@{}}
	    		\dval{rt} - \{r\} \cup \{nr\} & \mbox{if }  \dval{dip}\in\kD{\dval{rt}}\\[1mm]
	    		\dval{rt}  &\mbox{otherwise\ ,}
	    \end{array}\right.
\end{array} \]
where $r:=\selr[\footnotesize]{\dval{rt}}{\dval{dip}}=(\dval{dip}\comma\dval{dsn}\comma\dval{dsk}\comma\dval{flag}\comma\dval{hops}\comma\dval{nhip}\comma\dval{pre}\comma\dval{ltime})$ is the current entry in the routing table for {\dval{dip}}
and $\dval{nr}:=(\dval{dip}\comma\dval{dsn}\comma\dval{dsk}\comma\dval{flag}\comma\dval{hops}\comma\dval{nhip}\comma\dval{pre}\comma$ $\max(\dval{ltime},\dval{t}))$ 
is identical to~$r$ except for the expiration time, which is updated.
\[\begin{array}{l}
	      \fnsettimequeues :\tQUEUES\times\tIP\times\tTIME\to \tQUEUES	      \\
	     \settimequeues{\dval{store}}{\dval{dip}}{\dval{t}}:=
	     \left\{\begin{array}{@{\,}l@{\qquad }l@{}}
	    \multicolumn{2}{l}{\dval{store}-\{(\dval{dip},p,*,q)\}\cup\{(\dval{dip},p,t,q)\}}\\
	     		& \mbox{if } (\dval{dip},p,*,q)\in\dval{store}\\[1mm]
	    \dval{store}&\mbox{otherwise}
	    \end{array}\right.
\end{array} \]

}

\subsection*{Summary}
Table~\ref{tab:types} shows AODV's data structure; 
detailed explanations can be found in~\cite{TR13}.

\begin{table}
\vspace*{-15mm}
\caption{Data structure of AODV\label{tab:types}}

\vspace*{-2ex}
\begin{center}
\setlength{\tabcolsep}{2.5pt}
{\footnotesize\scriptsize
\begin{tabular}{@{}|l@{\,}|l@{}|l@{\,}|@{}}
\hline
\textbf{Type} & \textbf{Variables} & \textbf{Description}\\
\hline
 \tIP			&\ip,\,\dip,\,\oip,\,\rip,\,\sip,\,\nhip	&node identifiers\\
 \tSQN		&\dsn,\,\osn,\,\rsn,\,\sn	&sequence numbers\\
 \tSQNK		&\keyw{dsk}        			&sequence-number-status flag\\
 \tFLAG		&\flag					&route validity\\
 \NN		&\hops					&hop counts\\
 \tROUTE  	&					&routing table entries\\
 \tRT		&\rt					&routing tables\\		
 \tRREQID	&\rreqid				&request identifiers\\
 \tPendingRREQ&		                                &pending-request counter\\
 \tQUEUES       &\queues				&store of queued data packets\\
 \tMSG		&\msg					&messages\\
 ${[\tTYPE]}$						&			&queues with elements of type \tTYPE\\
 \quad[\tMSG]						&\msgs		&message queues\\
$\tIP\rightharpoonup\tSQN$      &\dests &sets of destinations with sequence numbers\\
$\pow(\tIP)$			&$\pre$	&sets of identifiers (precursors, destinations,  \dots)\\
$\pow(\tIP\mathord\times\tRREQID\mathord\times \tTIME)$	&\rreqs		&sets of request identifiers  with originator IP\\
\hline\hline
\multicolumn{2}{|l@{\,}|}{\textbf{Constant/Predicate}}& \textbf{Description}\\
\hline
\multicolumn{2}{|l@{\,}|}{$\kno,\unkno\mathop{:}\tSQNK$}&
	constants to distinguish  known and unknown sqns\\
\multicolumn{2}{|l@{\,}|}{$\val,\inval\mathop{:}\tFLAG$}&
	constants to distinguish  valid and invalid routes\\
\multicolumn{2}{|l@{\,}|}{$\retries\mathop{:}\tPendingRREQ$}&
	maximal number of RREQ attempts\\
\multicolumn{2}{|l@{\,}|}{$\deltime, \valtime,$}& time constants\\
\multicolumn{2}{|l@{\,}|}{$\myroute, \nodetraversal,$}& \\
\multicolumn{2}{|l@{\,}|}{$\traversal, \pathdiscoverytime\mathop{:}\tTIME$}&\\
\hline\hline
\multicolumn{2}{|l@{\,}|}{\textbf{Operator}} & \textbf{Description}\\
\hline
\multicolumn{2}{|l@{\,}|}{$\fnhead\mathop{:}[\tTYPE]\rightharpoonup\tTYPE$}&
	returns the ``oldest'' element in the queue\\
\multicolumn{2}{|l@{\,}|}{$\fntail\mathop{:}[\tTYPE]\rightharpoonup[\tTYPE]$}&
	removes the ``oldest'' element in the queue\\
\multicolumn{2}{|l@{\,}|}{$\fnappend\mathop{:}\tTYPE\mathord\times[\tTYPE]\rightarrow[\tTYPE]$}&
	inserts a new element into the queue\\
\multicolumn{2}{|l@{\,}|}{$\fndrop\mathop{:}\tIP\mathord\times\tQUEUES \rightharpoonup\tQUEUES$}&
	deletes a packet from the queued data packets\\
\multicolumn{2}{|l@{\,}|}{$\fnadd\mathop{:}\tDATA\mathord\times\tIP\mathord\times\tQUEUES \to\tQUEUES$}&
	adds a packet to the queued data packets\\
\multicolumn{2}{|l@{\,}|}{$\fnincretries\mathop{:}\tQUEUES\mathord\times\tIP\to\tQUEUES$}&
	increments the number of pending requests\\
\multicolumn{2}{|l@{\,}|}{$\fnsetretries\mathop{:}\tQUEUES\mathord\times(\tIP\rightharpoonup\tSQN)\to\tQUEUES$}&
	resets the number of pending requests\\
\multicolumn{2}{|l@{\,}|}{$\fnselqueue\mathop{:}\tQUEUES\mathord\times\tIP \rightarrow [\tDATA]$}&
	selects the data queue for a particular destination\\
\multicolumn{2}{|l@{\,}|}{$\fnfD\mathop{:}\tQUEUES\mathord\times\tIP\rightharpoonup\tPendingRREQ$}&
	returns the number of route requests initiated\\
\multicolumn{2}{|l@{\,}|}{$\fnseltime\mathop{:}\tQUEUES\mathord\times\tIP\rightharpoonup\tTIME$}&
	tells when the next route request should be sent\\
\multicolumn{2}{|l@{\,}|}{$\fnselroute\mathop{:}\tRT\mathord\times\tIP \rightharpoonup \tROUTE$}&
	selects the route for a particular destination\\
\multicolumn{2}{|l@{\,}|}{$(\_\,,\_\,,\_\,,\_\,,\_\,,\_\,,\_\,,\_\,)\mathop{:}$}&
	generates a routing table entry\\
\multicolumn{2}{|l@{\,}|}{$\quad
	\tIP\mathord\times \tSQN \mathord\times\tSQNK \mathord\times\tFLAG \mathord\times \NN
        \mathord\times \tIP \mathord\times \pow(\tIP)\mathord\times\tTIME \mathop\rightarrow \tROUTE$}&
	\\
\multicolumn{2}{|l@{\,}|}{$\fninc\mathop{:}\tSQN \rightarrow \tSQN$}&
	increments the sequence number\\
\multicolumn{2}{|l@{\,}|}{$\max\mathop{:}\tSQN\mathord\times\tSQN \to\tSQN$}&
	returns the larger sequence number\\	
\multicolumn{2}{|l@{\,}|}{$\fnsqn\mathop{:}\tRT \mathord\times \tIP \to \tSQN$}&
	returns the sequence number of a particular route\\	
\multicolumn{2}{|l@{\,}|}{$\fnsqnf\mathop{:}\tRT \mathord\times \tIP \to \tSQNK$}&
	determines whether the sequence number is known\\
\multicolumn{2}{|l@{\,}|}{$\fnstatus\mathop{:}\tRT\mathord\times\tIP\rightharpoonup\tFLAG$}&
	returns the validity of a particular route\\	
\multicolumn{2}{|l@{\,}|}{$\fndhops\mathop{:}\tRT \mathord\times \tIP \rightharpoonup \NN$}&
	returns the hop count of a particular route\\
\multicolumn{2}{|l@{\,}|}{$\fnnhop\mathop{:}\tRT \mathord\times \tIP \rightharpoonup \tIP$}&
	returns the next hop of a particular route\\
\multicolumn{2}{|l@{\,}|}{$\fnprecs\mathop{:}\tRT \mathord\times \tIP \rightharpoonup \pow(\tIP)$}&
	returns the set of precursors of a particular route\\
\multicolumn{2}{|l@{\,}|}{$\fnltime\mathop{:}\tRT \mathord\times \tIP \rightharpoonup \tTIME$}&
	returns the expiration time of a particular route\\
\multicolumn{2}{|l@{\,}|}{$\fnakD, \fnikD, \fnkD\mathop{:}\tRT \rightarrow\pow(\tIP)$}&
	returns the set of valid/invalid/known destinations\\	
\multicolumn{2}{|l@{\,}|}{$\fnqD\mathop{:}\tQUEUES \rightarrow \pow(\tIP)$}&
	returns the set of destinations with unsent packets\\	
\multicolumn{2}{|l@{\,}|}{$\fnaddprec \mathop{:} \tROUTE\mathord\times \pow(\tIP) \to \tROUTE$}&
	adds a set of precursors to a routing table entry\\	
\multicolumn{2}{|l@{\,}|}{$\fnaddprecrt \mathop{:} \tRT\mathord\times\tIP\mathord\times \pow(\tIP) \rightharpoonup \tRT$}&
	adds a set of precursors to an entry inside a table\\	
\multicolumn{2}{|l@{\,}|}{$\fnupd\mathop{:}\tRT \mathord\times \tROUTE \rightharpoonup \tRT$}&
	updates a routing table with a route\\
\multicolumn{2}{|l@{\,}|}{$\fninv\mathop{:}\tRT \mathord\times (\tIP\rightharpoonup\tSQN) \rightarrow \tRT$}&
	invalidates a set of routes within a routing table\\
\multicolumn{2}{|l@{\,}|}{$\fnnrreqid\mathop{:} \pow(\tIP\mathord\times\tRREQID\mathord\times\tTIME) \mathord\times \tIP \rightarrow \tRREQID$}&
	generates a new route request identifier\\
\multicolumn{2}{|l@{\,}|}{$\newpktID\mathop{:}\tDATA \mathord\times \tIP \rightarrow \tMSG$}&
	generates a message with new appl.~layer data\\
\multicolumn{2}{|l@{\,}|}{$\pktID\mathop{:}\tDATA \mathord\times \tIP \mathord\times \tIP \rightarrow \tMSG$}&
	generates a message containing appl.~layer data\\
\multicolumn{2}{|l@{\,}|}{$\rreqID\!\mathop{:}\!\NN\!\mathord\times \tRREQID \mathord\times \tIP \mathord\times \tSQN\mathord\times\tSQNK \mathord\times \tIP \mathord\times \tSQN \mathord\times \tIP\!\rightarrow\!\tMSG$}&
	generates a route request\\
\multicolumn{2}{|l@{\,}|}{$\rrepID\!\mathop{:}\!\NN\!\mathord\times \tIP \mathord\times \tSQN \mathord\times \tIP \mathord\times\tTIME \mathord\times \tIP\!\rightarrow \tMSG$}&
	generates a route reply\\
\multicolumn{2}{|l@{\,}|}{$\rerrID\mathop{:}(\tIP\rightharpoonup\tSQN) \mathord\times \tIP \rightarrow \tMSG$}&
	generates a route error message\\
\multicolumn{2}{|l@{\,}|}{$\fnexprt \mathop{:}\tRT\mathord\times\tTIME\mathord\times\tTIME\to \tRT$}&invalidates/removes expired routing table entries\\
\multicolumn{2}{|l@{\,}|}{$\fnexprreqs \mathop{:}\pow(\tIP\mathord\times\tRREQID\mathord\times\tTIME)\mathord\times\tTIME$}&removes expired route request identifiers\\
\multicolumn{2}{|l@{\,}|}{$\quad\to \pow(\tIP\mathord\times\tRREQID\mathord\times\tTIME)$}&\\
\multicolumn{2}{|l@{\,}|}{$\fnexpq \mathop{:}\tQUEUES\mathord\times\tTIME\to \tQUEUES$}&removes expired entries form a store\\

\multicolumn{2}{|l@{\,}|}{$\fnsettimert \mathop{:}\tRT\mathord\times\tIP\mathord\times\tTIME\to \tRT$}&updates expiration time for a routing table entry\\
\multicolumn{2}{|l@{\,}|}{$\fnsettimequeues \mathop{:}\tQUEUES\mathord\times\tIP\mathord\times\tTIME\to \tQUEUES$}&updates expiration time for an entry in the store\\
\hline
\end{tabular}
}
\end{center}
\vspace{-20mm}
\end{table}

\subsubsection{Modelling AODV}\label{ssubapp:ModelAODV}~

\noindent
Our model of AODV consists of $7$ processes, named $\AODV$, $\NEWPKT$, $\PKT$, $\RREQ$,
$\RREP$, $\RERR$ and $\QMSG$; their formal specifications are displayed as Processes~\ref{pro:aodv}--\ref{pro:queues}.
The red-coloured parts are those bits that differ from the specification given in~\cite{TR13,GHPT16}.
In this paper we only explain those parts, and refer to \cite[Section 6]{TR13} for a detailed explanation of all other
parts.

\paragraph{\rm\bf The Basic Routine.}
The basic process $\AODV$, depicted in \Pro{aodv}, either handles a message from the
corresponding queue, sends a queued data packet if a route to the destination has been established, 
or initiates a new route discovery process in case of queued
data packets with invalid or unknown routes.
This process maintains five data variables, {\ip}, {\sn}, {\rt}, {\rreqs} and
{\queues}, in which it stores its own identity, its own sequence number, its current routing
table, the list of route requests seen so far, and its current store
of queued data packets that await transmission.

With timers in place, the routing table needs regular updates. 
In particular, valid routing table entries have to be invalidated, and
invalid ones need to be erased 
when the expiration time of an entry has been reached.
Hence each time before we use information from the routing table $\rt$ maintained by a node, we prune expired
routes from the routing table, and invalidate routes that have not been used for a long time; this
happens for instance in Line~\ref{aodv:line1} of
Process~\ref{pro:aodv}, so that the updated routing table is used when we 
evaluate the guard of Line~\ref{aodv:line22}, checking that there is a valid route to $\dip$.

We again prune $\rt$ in Line~\ref{aodv:line3}, prior to for instance evaluating the guard in
Line~\ref{pkt2:line5} of Process~\ref{pro:pkt}---repeated pruning is needed
because time may have passed upon receiving the message in Line~\ref{aodv:line2} of Process~\ref{pro:aodv}.
A similar argument applies to Lines~\ref{aodv:line25} and~\ref{aodv:line29b}.

Likewise, before we consult the store of queued data packets (e.g.\ in
Lines~\ref{aodv:line22} and~\ref{aodv:line34}) we drop all packets
from those queues for which  $\retries$
unsuccessful attempts have been made to find a route to the destination (Line~\ref{aodv:line1q}).

Each time a routing table entry is updated (Lines~\ref{aodv:line10},~\ref{aodv:line14}
and~\ref{aodv:line18}) the lifetime of the entry is set to $\valtime$ (so that the expiration time
becomes%
  \algsetup{linenodelimiter=.,linenosize=\tiny}
  \begin{algorithm}[H]
    {\scriptsize
      \caption{The basic routine}
      \label{pro:aodv}
      \begin{algorithmic}[1]\scriptsize
\DEFPROCESS{\AODV}{\ip\comma\sn\comma\rt\comma\rreqs\comma\queues}
   \COMLINE{clean up routing table, and data storage}
	\highlightUPD{\rt:=\exprt{\rt}{\now}{\deltime}}      \label{aodv:line1}
	\highlightUPD{\queues:=\expq{\queues}{\now}}      \label{aodv:line1q}
	\PAR
	\COMLINE{node receives a message}
	\IFempty
		\receiveL{\msg}\ .																															\label{aodv:line2}
		\highlightUPD{\rt:=\exprt{\rt}{\now}{\deltime}}   \label{aodv:line3}
		\COMLINE{depending on the message, the node calls different processes}	
		\PAR
		\IF[new DATA packet]{$\msg = \newpkt{\data}{\dip}$}																	\label{aodv:line4}
			\newpktP{\data}{\dip}{\ip}{\sn}{\rt}{\rreqs}{\queues}																	\label{aodv:line5}
		\ELSIF[incoming DATA packet]{$\msg = \pkt{\data}{\dip}{\oip}$}  \label{aodv:line6}
			\pktP{\data}{\dip}{\oip}{\ip}{\sn}{\rt}{\rreqs}{\queues}																	\label{aodv:line7}
		\ELSIF[RREQ]{$\msg = \rreq{\hops}{\rreqid}{\dip}{\dsn}{\dsk}{\oip}{\osn}{\sip}$}								\label{aodv:line8}
			\COMLINE{update the route to \sip\ in \rt}																					\label{aodv:line9}
			\UPD{\rt:=\upd{\rt}{(\sip,0,\unkno,\val,1,\sip,\emptyset,\highlight{\now+\valtime})}}																\label{aodv:line10}
			\rreqP{\hops}{\rreqid}{\dip}{\dsn}{\dsk}{\oip}{\osn}{\sip}{\ip}{\sn}{\rt}{\rreqs}{\queues}
		\ELSIF[RREP]{$\msg = \rrep{\hops}{\dip}{\dsn}{\oip}{\highlight\lifetime}{\sip}$}											\label{aodv:line12}
			\COMLINE{update the route to \sip\ in \rt}																					\label{aodv:line13}
			\UPD{\rt:=\upd{\rt}{(\sip,0,\unkno,\val,1,\sip,\emptyset,\highlight{\now+\valtime})}}																\label{aodv:line14}
			\rrepP{\hops}{\dip}{\dsn}{\oip}{\highlight\lifetime}{\sip}{\ip}{\sn}{\rt}{\rreqs}{\queues}												\label{aodv:line15}
		\ELSIF[RERR]{$\msg = \rerr{\dests}{\sip}$}																					\label{aodv:line16}
			\COMLINE{update the route to \sip\ in \rt}																					\label{aodv:line17}
			\UPD{\rt:=\upd{\rt}{(\sip,0,\unkno,\val,1,\sip,\emptyset,\highlight{\now+\valtime})}}																\label{aodv:line18}
			\rerrP{\dests}{\sip}{\ip}{\sn}{\rt}{\rreqs}{\queues}																			\label{aodv:line19}
		\ENDIFii
		\ENDPAR																																		\label{aodv:line20}
		\ELSIF[send a queued packet if a valid route is known]{$\mbox{Let } \dip\mathbin\in\qD{\queues}\cap\akD{\rt}$}				\label{aodv:line22}
			\UPD{\data:=\head{\selq{\queues}{\dip}}}													\label{aodv:line23}
			\STARTPRIO
			 	\unicast{\nhop{\rt}{\dip}}{\pkt{\data}{\dip}{\ip}}\ . 											\label{aodv:line24}
                        	\highlightUPD{\rt:=\exprt{\rt}{\now}{\deltime}}      \label{aodv:line25}
				\UPD{\queues:=\drop{\dip}{\queues}}\COMMENT{drop {\data} from the {\queues} for {\dip}}													\label{aodv:line26}
				\highlightUPD{\rt:=\settimert{\rt}{\dip}{\now+\valtime}} \label{aodv:line26b}
				\highlightUPD{\rt:=\settimert{\rt}{\nhop{\rt}{\dip}}{\now+\valtime}}\label{aodv:line26c}
				\aodvL{\ip}{\sn}{\rt}{\rreqs}{\queues}			\label{aodv:line27}
			 \PRIO
				\COMspec{an error is produced and the routing table is updated}							\label{aodv:line29}
				\highlightUPD{\rt:=\exprt{\rt}{\now}{\deltime}} \label{aodv:line29b}
				\UPD{\dests:=\{(\rip,\inc{\sqn{\rt}{\rip}})\,|\,\rip\in\akD{\rt}\ans \nhop{\rt}{\rip}=\nhop{\rt}{\dip}\}}		\label{aodv:line30}
				\UPD{\rt:=\inv{\rt}{\dests}{\highlight{\now+\deltime}}}														\label{aodv:line32}
				\UPD{\queues:=\highlight{\setretries{\queues}{\dests}}}\label{aodv:line32a}
				\UPD{\pre:=\bigcup\{\precs{\rt}{\rip}\,|\,(\rip,*)\in\dests\}}									\label{aodv:line31}
				\UPD{\dests:=\{(\rip,\rsn)\,|\,(\rip,\rsn)\in\dests\ans \precs{\rt}{\rip}\not=\emptyset\}}				\label{aodv:line31a}
				\groupcast{\pre}{\rerr{\dests}{\ip}}\ .
				\aodv{\ip}{\sn}{\rt}{\rreqs}{\queues}													\label{aodv:line33}
 		 	\ENDPRIO
		\ELSIF[a route discovery process is initiated]{$\mbox{Let } \dip\mathbin\in\qD{\queues}{-}\akD{\rt}\wedge\highlight{\fD{\queues}{\dip} < \retries \wedge \selt{\queues}{\dip} \leq \now}$}	\label{aodv:line34}
			\UPD{\queues:=\highlight{\incretries{\queues}{\dip}}}\label{aodv:line34b}
			\highlightUPD{\queues:=\settimequeues{\queues}{\dip}{\now+2^{{\sigma_{\it retries}}(\queues,\dip)}\cdot\traversal}}			\label{aodv:line35}		
			\highlightUPD{\rt:=\settimert{\rt}{\dip}{\now+2\cdot\traversal}}\label{aodv:line35b}		
			\UPD{\sn:=\inc{\sn}}\COMMENT{increment own sequence number}								\label{aodv:line36}
			\COMLINE{update \rreqs\ by adding $(\ip,\nrreqid{\rreqs}{\ip})$}								\label{aodv:line37}
			\UPD{\rreqid:=\nrreqid{\rreqs}{\ip}}							\label{aodv:line38a}
			\UPD{\rreqs := \rreqs\cup\{(\ip,\rreqid, \highlight{\now+\pathdiscoverytime})\}}							\label{aodv:line38b}
			\broadcast{\rreq{$0$}{\rreqid}{\dip}{\sqn{\rt}{\dip}}{\sqnf{\rt}{\dip}}{\ip}{\sn}{\ip}}\ .					\label{aodv:line39}
			\aodvL{\ip}{\sn}{\rt}{\rreqs}{\queues}														\label{aodv:line40}			
	\ENDIFii
	\ENDPAR

	\end{algorithmic}
    }
  \end{algorithm}

\vspace{-4mm}
\noindent
 $\now+\valtime$), according to \cite[Section 6.2]{rfc3561}.
Likewise, after a route is used to forward a data packet (Line~\ref{aodv:line24}),
the lifetime of the routing table entries for the destination and for the next hop on the
path to the destination are updated in the same way (Line~\ref{aodv:line26b} and~\ref{aodv:line26c}),
again according to \cite[Section 6.2]{rfc3561}.
The lifetime parameter of route reply messages is simply passed on from the incoming message of
Line~\ref{aodv:line12} to the process $\RREP$ in Line~\ref{aodv:line15}.

When invalidating routing table entries in Line~\ref{aodv:line32}, the expiration time of the
invalidated entries is set to $\now+\deltime$, according to \cite[Section 6.11]{rfc3561}.
For each of the newly invalidated destinations, a fresh route discovery process needs to be
initiated. To this end, the number of pending route request for that destination is set to $0$,
and the time after which the next route request can be made to $\now$ (Line~\ref{aodv:line32a}).

If the guard of Line~\ref{aodv:line34} evaluates to \keyw{true}, a route discovery process for a
destination $\dip$ will be initiated. For this to happen, according to \cite[Section~6.3]{rfc3561}, the
number $\fD[\footnotesize]{\queues}{\dip}$ of pending route requests for $\dip$ needs to be smaller than the
parameter $\retries$. Moreover, the time we were instructed to wait for has been reached
($\selt[\footnotesize]{\queues}{\dip}\mathrel{\leq} \now$). When a new route request is being made, the
recorded number of pending route requests for $\dip$ is incremented (Line~\ref{aodv:line34b}),
and, again according to \cite[Section 6.3]{rfc3561}, an \pagebreak[3]instruction is processed to wait until time
$\now+2^{\fnfD[\tiny](\queues,\dip)}\cdot\traversal$ before issuing a new
route request for $\dip$ (Line~\ref{aodv:line35}). Furthermore, Line~\ref{aodv:line35b} says that
a routing table entry waiting for a route reply should not be expunged before time $\now + 2 \cdot
\traversal$ \cite[Section 6.4]{rfc3561}. Finally, Line~\ref{aodv:line38b} indicates that
``before broadcasting the RREQ, the originating node buffers the RREQ ID and the Originator IP
  address (its own address) of the RREQ for \pathdiscoverytime'' (\cite[Section 6.3]{rfc3561}).

\paragraph{\rm\bf Data Packet Handling.}
The process $\NEWPKT$ (\Pro{newpkt}), describing all actions performed by a node when a data packet is injected
by a client hooked up to the local node, is unchanged w.r.t.\ \cite{TR13,GHPT16}.

\vspace{-5mm}
  \algsetup{linenodelimiter=.,linenosize=\tiny}
  \begin{algorithm}[H]
    {\scriptsize
      \caption{Routine for handling a newly injected data packet}
      \label{pro:newpkt}
      \begin{algorithmic}[1]\scriptsize
\DEFPROCESS{\NEWPKT}{\data\comma\dip\,\comma\,\ip\comma\sn\comma\rt\comma\rreqs\comma\queues}
	\IF[the DATA packet is intended for this node]{$\dip=\ip$}															\label{newpkt:line2}
		\deliverL{\data}\ .
		\aodv{\ip}{\sn}{\rt}{\rreqs}{\queues}																						\label{newpkt:line3}
	\ELSIF[the DATA packet is not intended for this node]{$\dip\not=\ip$}										\label{newpkt:line4}
		\UPD{\queues:=\add{\data}{\dip}{\queues}}	\ .
		\aodv{\ip}{\sn}{\rt}{\rreqs}{\queues}																						\label{newpkt:line5}	
	\ENDIFii

	\end{algorithmic}
    }
  \end{algorithm}

\vspace{-4mm}

In the process $\PKT$ (\Pro{pkt}), dealing with data packets received via the protocol,
a data packet is forwarded to the next hop on the route to the destination in Line~\ref{pkt2:line6}.
According to \cite[Section 6.2]{rfc3561}, the expiration times of the routing table entries for the
destination, the next hop on the path to the destination, the source and the next hop on the
path to the source of the message are all set to $\now+\valtime$ (Lines~\ref{pkt2:line6a}--\ref{pkt2:line6d}).
The handling of an unsuccessful transmission is exactly as in Process~\ref{pro:aodv}.
Line ~\ref{pkt2:line19} says that ``if a data packet is received for an invalid
route, the lifetime field is updated to current time plus $\deltime$'' \cite[Section 6.11]{rfc3561}.

\paragraph{\rm\bf Receiving Route Requests.}
The process $\RREQ$ (\Pro{rreq}) models all events that may occur after a route
request has been received. In case the node itself is the intended destination of the RREQ message, 
the node generates a route reply (RREP) message, which is then sent along the established reverse route. 
A RREP message is also generated in case an intermediate node (a node that is\\[-8pt]
\mbox{}
\vspace{-11mm}
  \algsetup{linenodelimiter=.,linenosize=\tiny}
  \begin{algorithm}[H]
    {\scriptsize
      \caption{Routine for handling a received data packet}
      \label{pro:pkt}
      \begin{algorithmic}[1]\scriptsize
\DEFPROCESS{\PKT}{\data\comma\dip\comma\oip\,\comma\,\ip\comma\sn\comma\rt\comma\rreqs\comma\queues}
	\IF[the DATA packet is intended for this node]{$\dip=\ip$}																				\label{pkt2:line2}
		\deliverL{\data}\ .
		\aodv{\ip}{\sn}{\rt}{\rreqs}{\queues}																											\label{pkt2:line3}
	\ELSIF[the DATA packet is not intended for this node]{$\dip\not=\ip$}															\label{pkt2:line4}
	 	\PAR
			\IF[valid route to \dip]{$\dip\in\akD{\rt}$}																									\label{pkt2:line5}
				\COMLINE{forward packet}
				\STARTPRIO
					\unicast{\nhop{\rt}{\dip}}{\pkt{\data}{\dip}{\oip}}\ . \label{pkt2:line6}
                                     	\highlightUPD{\rt:=\exprt{\rt}{\now}{\deltime}}      \label{aodv:line6extra}
					\highlightUPD{\rt:=\settimert{\rt}{\dip}{\now+\valtime}}\label{pkt2:line6a}
					\highlightUPD{\rt:=\settimert{\rt}{\nhop{\rt}{\dip}}{\now+\valtime}}\label{pkt2:line6b}
					\highlightUPD{\rt:=\settimert{\rt}{\oip}{\now+\valtime}}\label{pkt2:line6c}
					\highlightUPD{\rt:=\settimert{\rt}{\nhop{\rt}{\oip}}{\now+\valtime}}\label{pkt2:line6d}
				\aodvL{\ip}{\sn}{\rt}{\rreqs}{\queues}						\label{pkt2:line7}
				\PRIO
					\COMspec{If the packet transmission is unsuccessful, a RERR message is generated}
					\highlightUPD{\rt:=\exprt{\rt}{\now}{\deltime}}		
					\UPD{\dests:=\{(\rip,\inc{\sqn{\rt}{\rip}})\,|\,\rip\in\akD{\rt}\ans \nhop{\rt}{\rip}=\nhop{\rt}{\dip}\}}			\label{pkt2:line9}
					\UPD{\rt:=\inv{\rt}{\dests}{\highlight{\now+\deltime}}}																												\label{pkt2:line10}
					\UPD{\queues:=\highlight{\setretries{\queues}{\dests}}}																			\label{pkt2:line11}
					\UPD{\pre:=\bigcup\{\precs{\rt}{\rip}\,|\,(\rip,*)\in\dests\}}																	\label{pkt2:line12}
					\UPD{\dests:=\{(\rip,\rsn)\,|\,(\rip,\rsn)\in\dests\ans \precs{\rt}{\rip}\not=\emptyset\}}							\label{pkt2:line13}
					\groupcast{\pre}{\rerr{\dests}{\ip}}\ . \aodv{\ip}{\sn}{\rt}{\rreqs}{\queues}											\label{pkt2:line14}
				\ENDPRIO
			\ELSIF[no valid route to \dip]{$\dip\not\in\akD{\rt}$}																				\label{pkt2:line15}
				\COMLINE{no local repair occurs; data is lost}																					\label{pkt2:line16}
				\PAR	
					\IF[invalid route to \dip]{$\dip\in\ikD{\rt}$}																						\label{pkt2:line18}
		\highlightUPD{\rt:=\settimert{\rt}{\dip}{\now+\deltime}}\label{pkt2:line19} 
						\COMLINE{if the route is invalid, a RERR is sent to the precursors}
						\groupcast{\precs{\rt}{\dip}}{\rerr{\{(\dip,\sqn{\rt}{\dip})\}}{\ip}}\ .											
						\aodvL{\ip}{\sn}{\rt}{\rreqs}{\queues}																							\label{pkt2:line20}
					\ELSIF[route not in \rt]{$\dip\not\in\ikD{\rt}$}																					\label{pkt2:line21}
						\aodvL{\ip}{\sn}{\rt}{\rreqs}{\queues}																							\label{pkt2:line22}
					\ENDIFii
				\ENDPAR																																			\label{pkt2:line23}									
			\ENDIFii
		\ENDPAR																																					\label{pkt2:line24}
	\ENDIFii

	\end{algorithmic}
    }
  \end{algorithm}

\noindent
 neither the destination nor the originator of the RREQ message) receives it and has knowledge about a valid and fresh enough route to the destination.

Just as in \Pro{aodv}, the process prunes expired
routes from the routing table before reading the routing table.
This happens in Lines~\ref{rreq:line15a} and~\ref{rreq:line27}.
Likewise, before consulting the list of already handled route requests
in Line~\ref{rreq:line2} the process
expunges expired entries from this list in Line~\ref{rreq:line1}.

In Line~\ref{rreq:line6}, the routing table for the originator $\oip$ of the received route request
is updated. According to \cite[Section 6.2]{rfc3561}, the lifetime of the entry is ``initialized to \valtime'', whereas according to \cite[Section 6.5]{rfc3561},
the expiration time ``is set to be the maximum of (ExistingLifetime, MinimalLifetime)'', where MinimalLifetime =
$$\now+2\cdot\traversal-2\cdot(\hops+1)\cdot\nodetraversal.$$
We implement both instructions, in Lines~\ref{rreq:line6} and~\ref{rreq:line7}, thereby taking the
maximum lifetime resulting from both instructions.

In Line~\ref{rreq:line8} we add the unique identifier $(\oip\comma\rreqid)$ for the current route
request as a new entry in the list of already handled route requests; its expiration time is set to
$\now+\pathdiscoverytime$, according to \cite[Section 6.5]{rfc3561}.

  \algsetup{linenodelimiter=.,linenosize=\tiny}
  \begin{algorithm}[H]
    {\scriptsize
      \caption{RREQ handling}
      \label{pro:rreq}
      \begin{algorithmic}[1]\scriptsize
\DEFPROCESS{\RREQ}{\hops\comma\rreqid\comma\dip\comma\dsn\comma\dsk\comma\oip\comma\osn\comma\sip\,\comma\,\ip\comma\sn\comma\rt\comma\rreqs\comma\queues}
\highlightUPD{\exprreqs{\rreqs}{\now}}\label{rreq:line1}
\PAR
	\IF[the RREQ has been received previously]{$(\oip\comma\rreqid\comma\highlight*)\in\rreqs$}																							\label{rreq:line2}
		\aodvL{\ip}{\sn}{\rt}{\rreqs}{\queues} \COM{silently ignore RREQ, i.e., do nothing}																			\label{rreq:line3}
	\ELSIF[the RREQ is new to this node]{$(\oip\comma\rreqid\comma\highlight*)\not\in\rreqs$}																								\label{rreq:line4}
		\UPD{\rt:=\upd{\rt}{(\oip,\osn,\kno,\val,\hops+1,\sip,\emptyset,\highlight{\now+\valtime})}}	
		\label{rreq:line6}
		\highlightUPD{\rt:=\settimert{\rt}{\oip}{\now+2\cdot\traversal-2\cdot(\hops+1)\cdot\nodetraversal}}\label{rreq:line7}%
		\UPD{\rreqs:=\rreqs\cup\{(\oip,\rreqid,\highlight{\now+\pathdiscoverytime})\}}		\COMMENT{update \rreqs}											\label{rreq:line8}
		\PAR																																															\label{rreq:line9}
		\IF[this node is the destination node]{$\dip=\ip$}																															\label{rreq:line10}
			\UPD{\sn:=\max(\sn,\dsn)}	\COMMENT{update the sqn of \ip}																									\label{rreq:line12}
			\COMLINE{unicast a RREP towards \oip\ of the RREQ}																											\label{rreq:line13}
			\STARTPRIO
					\unicast{\nhop{\rt}{\oip}}{{\rrep{$0$}{\dip}{\sn}{\oip}{\highlight{\myroute}}{\ip}}}\ . 																								\label{rreq:line14a}								
					\aodvL{\ip}{\sn}{\rt}{\rreqs}{\queues}																																	\label{rreq:line14}
				\PRIO
					\COMspec{If the transmission is unsuccessful, a RERR message is generated}																\label{rreq:line15}
					\highlightUPD{\rt:=\exprt{\rt}{\now}{\deltime}}\label{rreq:line15a}
					\UPD{\dests:=\{(\rip,\inc{\sqn{\rt}{\rip}})\,|\,\rip\in\akD{\rt}\wedge \nhop{\rt}{\rip}=\nhop{\rt}{\oip}\}}\label{rreq:line16}%
					\UPD{\rt:=\inv{\rt}{\dests}{\highlight{\now+\deltime}}}																																					\label{rreq:line18}			
					\UPD{\queues:=\highlight{\setretries{\queues}{\dests}}}																														\label{rreq:line18a}
					\UPD{\pre:=\bigcup\{\precs{\rt}{\rip}\,|\,(\rip,*)\in\dests\}}																										\label{rreq:line17}
					\UPD{\dests:=\{(\rip,\rsn)\,|\,(\rip,\rsn)\in\dests\ans \precs{\rt}{\rip}\not=\emptyset\}}																\label{rreq:line17a}
					\groupcast{\pre}{\rerr{\dests}{\ip}}\ .																																	
					\aodv{\ip}{\sn}{\rt}{\rreqs}{\queues}																																	\label{rreq:line19}
				\ENDPRIO	
		\ELSIF[this node is not the destination node]{$\dip\not=\ip$}																											\label{rreq:line20}
			\PAR																																															\label{rreq:line21}
			\COMLINE{valid route to \dip\ that is fresh enough}
			\IF{$\!\dip\mathbin\in\akD{\rt} \wedge \dsn \mathbin\leq  \sqn{\rt}{\!\dip} \wedge\sqnf{\rt}{\!\dip}\mathbin=\kno\!$}		\label{rreq:line22}
					\COMLINE{update \rt\ by adding precursors}																														\label{rreq:line23}
					\UPD{\rt := \addprecrt{\rt}{\dip}{\{\sip\}}}																																\label{rreq:line24}
					\UPD{\rt := \addprecrt{\rt}{\oip}{\{\nhop{\rt}{\dip}\}}}																												\label{rreq:line25}
				\COMLINE{unicast a RREP towards the \oip\ of the RREQ}
				\STARTPRIO
				    \textbf{unicast}(\nosp{\nhop{\rt}{\oip}}\comma
				    \\\hspace{23.2pt}\nosp{\rrep{\dhops{\rt}{\dip}}{\dip}{\sqn{\rt}{\dip}}{\oip}{\highlight{\selt{\rt}{\dip}-\now}}{\ip}}\,.
				    \label{rreq:line26}%
					\aodvL{\ip}{\sn}{\rt}{\rreqs}{\queues}																																	\label{rreq:line26a}
				\PRIO
					\COMspec{If the transmission is unsuccessful, a RERR message is generated}\label{rreq:line26b}
					\highlightUPD{\rt:=\exprt{\rt}{\now}{\deltime}}\label{rreq:line27}
					\STATE{\textcolor{brown}{\ensuremath{\mbox{\bf [\![}\dests:=\{(\rip,\inc{\sqn{\rt}{\rip}})\,|}}}\\
								{\textcolor{brown}{\ensuremath{\hspace{40pt}\,\rip\in\akD{\rt}\ans \nhop{\rt}{\rip}=\nhop{\rt}{\oip}\}\mbox{\bf ]\!]}}}}											\label{rreq:line28}
					\UPD{\rt:=\inv{\rt}{\dests}{\highlight{\now+\deltime}}}																																						\label{rreq:line30}			
					\UPD{\queues:=\highlight{\setretries{\queues}{\dests}}}																													\label{rreq:line30a}		
					\UPD{\pre:=\bigcup\{\precs{\rt}{\rip}\,|\,(\rip,*)\in\dests\}}																										\label{rreq:line29}
					\UPD{\dests:=\{(\rip,\rsn)\,|\,(\rip,\rsn)\in\dests\ans \precs{\rt}{\rip}\not=\emptyset\}}																\label{rreq:line29a}
					\groupcast{\pre}{\rerr{\dests}{\ip}}\ . 																																	\label{rreq:line31}
					\aodv{\ip}{\sn}{\rt}{\rreqs}{\queues}
				\ENDPRIO
			\ELSIF[$\!$no fresh route$\!$]{$\dip\mathbin{\not\in}\akD{\rt} \vee \sqn{\rt}{\!\dip} <  \dsn \vee\sqnf{\rt}{\!\dip}\mathbin=\unkno$}					\label{rreq:line32}
				\COMLINE{no further update of \rt}
				\broadcast{\rreq{$\hops+1$}{\rreqid}{\dip}{\max(\sqn{\rt}{\dip}\comma\dsn)}{\dsk}{\oip}{\osn}{\ip}}\ .										\label{rreq:line34}
				\aodvL{\ip}{\sn}{\rt}{\rreqs}{\queues}				\label{rreq:line35}
			\ENDIFii
			\ENDPAR																																													\label{rreq:line36}
		\ENDIFii
		\ENDPAR																																														\label{rreq:line37}
	\ENDIFii
	\ENDPAR

	\end{algorithmic}
    }
  \end{algorithm}

\vspace{-2mm}

In Line~\ref{rreq:line14a}, when sending a route reply in answer to the incoming route request
because the current node \emph{is} the destination of the request, ``the destination node copies the
value $\myroute$ [\dots] into the Lifetime field of the RREP'' \cite[Section 6.6.1]{rfc3561}. 
However, when sending a route reply ``as an intermediate hop along the path from the originator to the
destination'' (Line~\ref{rreq:line26}), ``the Lifetime field of the RREP is calculated by
subtracting the current time from the expiration time in its route table entry'' \cite[Section 6.6.2]{rfc3561}.

The treatment of an unsuccessful unicast (Lines~\ref{rreq:line15}--\ref{rreq:line19} and
Lines~\ref{rreq:line26b}--\ref{rreq:line31}) is exactly as in Process~\ref{pro:aodv}.

\paragraph{\rm\bf Receiving Route Replies.}

We handle a received route reply only if it would give rise to a genuine update to the routing table
entry for the destination $\dip$ of the original route request (Lines~\ref{rrep:line1}
and~\ref{rrep:line25}), not counting updates to the lifetime of that entry. When we do update the
routing table (Line~\ref{rrep:line2}), ``the expiry time is set to the current time plus the value
of the Lifetime in the RREP message'' \cite[Section 6.7]{rfc3561}.

%
  \algsetup{linenodelimiter=.,linenosize=\tiny}
  \begin{algorithm}[H]
    {\scriptsize
      \caption{RREP handling}
      \label{pro:rrep}
      \begin{algorithmic}[1]\scriptsize
\DEFPROCESS{\RREP}{\hops\comma\dip\comma\dsn\comma\oip\comma\highlight\lifetime\comma\sip\,\comma\,\ip\comma\sn\comma\rt\comma\rreqs\comma\queues}
	\IF[$\!$routing table has to be updated$\!$]{$\!\rt\mathop{\not=}\upd{\rt}{(\dip\comma\dsn\comma\kno\comma\val\comma\hops\mathop+1\comma\sip\comma\emptyset\comma\highlight{0})}\!$}						\label{rrep:line1}
		\UPD{\rt:=\upd{\rt}{(\dip\comma\dsn\comma\kno\comma\val\comma\hops+1\comma\sip\comma\emptyset\comma\highlight{\now+\lifetime})}}			\label{rrep:line2}
		\PAR\label{rrep:line3}
		\IF[this node is the originator of the corresponding RREQ]{$\oip = \ip$}																	\label{rrep:line4}
			\COMLINE{a packet may now be sent; this is done in the process \AODV}
			\aodvL{\ip}{\sn}{\rt}{\rreqs}{\queues}																													\label{rrep:line8}
		\ELSIF[this node is not the originator; forward RREP]{$\oip \not= \ip$}																	\label{rrep:line9}
			\PAR
				\IF[valid route to \oip]{$\oip\in\akD{\rt}$}																											\label{rrep:line11}
					\COMLINE{add next hop towards $\oip$ as precursor and forward the route reply}									\label{rrep:line12}									
					\UPD{\rt := \addprecrt{\rt}{\dip}{\{\nhop{\rt}{\oip}\}}}																					\label{rrep:line12a}
					\UPD{\rt := \addprecrt{\rt}{\nhop{\rt}{\dip}}{\{\nhop{\rt}{\oip}\}}}																	\label{rrep:line12b}
					\highlightUPD{\rt:=\settimert{\rt}{\oip}{\now+\valtime}}\label{rrep:line12d}
					\STARTPRIO
						\unicast{\nhop{\rt}{\oip}}{\rrep{$\hops+1$}{\dip}{\dsn}{\oip}{\highlight{\lifetime}}{\ip}}\ .															\label{rrep:line13}
						\aodvL{\ip}{\sn}{\rt}{\rreqs}{\queues}																										\label{rrep:line14}
					\PRIO
						\COMspec{If the transmission is unsuccessful, a RERR message is generated}
						\highlightUPD{\rt:=\exprt{\rt}{\now}{\deltime}}
					\STATE{\textcolor{brown}{\ensuremath{\mbox{\bf [\![}\dests:=\{(\rip,\inc{\sqn{\rt}{\rip}})\,|}}}\\
								{\textcolor{brown}{\ensuremath{\hspace{40pt}\,\rip\in\akD{\rt}\ans \nhop{\rt}{\rip}=\nhop{\rt}{\oip}\}\mbox{\bf ]\!]}}}}												
						\label{rrep:line16}
						\UPD{\rt:=\inv{\rt}{\dests}{\highlight{\now+\deltime}}}																														\label{rrep:line18}							
					\UPD{\queues:=\highlight{\setretries{\queues}{\dests}}}																		\label{rrep:line16a}
						\UPD{\pre:=\bigcup\{\precs{\rt}{\rip}\,|\,(\rip,*)\in\dests\}}																			\label{rrep:line17}
						\UPD{\dests:=\{(\rip,\rsn)\,|\,(\rip,\rsn)\in\dests\ans \precs{\rt}{\rip}\not=\emptyset\}}									\label{rrep:line17a}
						\groupcast{\pre}{\rerr{\dests}{\ip}}\ .\ \aodv{\ip}{\sn}{\rt}{\rreqs}{\queues} 												\label{rrep:line20}
					\ENDPRIO
				\ELSIF[no valid route to \oip]{$\oip\not\in\akD{\rt}$}																						\label{rrep:line21}
					\aodvL{\ip}{\sn}{\rt}{\rreqs}{\queues}																											\label{rrep:line21a}
				\ENDIFii																																						\label{rrep:line22}
			\ENDPAR																																							\label{rrep:line23}
		\ENDIFii	
		\ENDPAR																																								\label{rrep:line23a}
	\ELSIF[$\!$routing table is not updated$\!$]{$\rt\mathop=\upd{\rt}{(\dip\comma\dsn\comma\kno\comma\val\comma\hops\mathop+1\comma\sip\comma\emptyset\comma\highlight{0})}$}							\label{rrep:line25}
		\aodvL{\ip}{\sn}{\rt}{\rreqs}{\queues}																														\label{rrep:line26}
	\ENDIFii

	\end{algorithmic}
    }
  \end{algorithm}

As implemented in Line~\ref{rrep:line12d}, ``the (reverse) route used to forward a RREP has its
lifetime changed to be the maximum of (existing-lifetime, (current time + \valtime)'') \cite[Section 6.7]{rfc3561}.
This literal reading of the RFC seems a bit weird, since the route to {\oip} is not updated otherwise. 
Although not specified in the RFC, it would make sense to also add a precursor to the reverse
route by 
$\assignment{\rt := \addprecrt{\rt}{\oip}{\{\nhop{\rt}{\dip}\}}}$\label{pg:small_precursor_improvement}.
Inserting this line, would not change the results and proofs presented in this paper.

\paragraph{\rm\bf Receiving Route Errors.}
The process $\RERR$ models the part of AODV that handles error messages.
An error message consists of a set $\dests$ of pairs of
an unreachable destination IP address $\rip$ and
the corresponding unreachable destination sequence number $\rsn$.
The adaptations to this process are just as the ones discussed earlier.

%
  \algsetup{linenodelimiter=.,linenosize=\tiny}
  \begin{algorithm}[H]
    {\scriptsize
      \caption{RERR handling}
      \label{pro:rerr}
      \begin{algorithmic}[1]\scriptsize
\DEFPROCESS{\RERR}{\dests\comma\sip\,\comma\,\ip\comma\sn\comma\rt\comma\rreqs\comma\queues}
		\COMLINE{invalidate broken routes}												\label{rerr:line4}
		\UPD{\dests:=\{(\rip,\rsn)\,|\,(\rip,\rsn)\mathbin\in\dests\wedge\rip\mathbin\in\akD{\rt}\wedge \nhop{\rt}{\rip}\mathbin=\sip\wedge \sqn{\rt}{\rip}\mathbin<\rsn\}}\label{rerr:line2}%
		\UPD{\rt:=\inv{\rt}{\dests}{\highlight{\now+\deltime}}}										\label{rerr:line5}
		\UPD{\queues:=\highlight{\setretries{\queues}{\dests}}}\label{rerr:line5a}
		\COMLINE{forward the RERR to all precursors for \rt\ entries for broken connections}			\label{rerr:line1}
		\UPD{\pre:=\bigcup\{\precs{\rt}{\rip}\,|\,(\rip,*)\in\dests\}}									\label{rerr:line3}
		\UPD{\dests:=\{(\rip,\rsn)\,|\,(\rip,\rsn)\in\dests\ans \precs{\rt}{\rip}\not=\emptyset\}}				\label{rerr:line3a}
		\groupcast{\pre}{\rerr{\dests}{\ip}}\ . \aodv{\ip}{\sn}{\rt}{\rreqs}{\queues}						\label{rerr:line6}

	\end{algorithmic}
    }
  \end{algorithm}

\vspace{-7mm}

\paragraph{\rm\bf The Message Queue.}
Since we have to guarantee input-enabledness of all network nodes, 
a node \dval{ip} must always be able to perform a receive action, regardless of which state
it is in. For this reason we
introduce a process $\QMSG$, modelling a message queue,
that runs in parallel with {\AODV} or any other process that might be called.
This process is unchanged w.r.t.\ \cite{TR13,GHPT16}.


\vspace{-2mm}
  \algsetup{linenodelimiter=.,linenosize=\tiny}
  \begin{algorithm}[H]
    {\scriptsize
      \caption{Message queue}
      \label{pro:queues}
      \begin{algorithmic}[1]\scriptsize
\DEFPROCESS{\QMSG}{\msgs}
	\IFempty
		\COMLINE{store incoming message at the end of \msgs}			\label{queues:line1}
		\receiveL{\msg}\ . 			
		\Qmsg{\append{\msg}{\msgs}}								\label{queues:line2}
	\ELSIF[the queue is not empty]{$\msgs\not=[\,]$}					\label{queues:line3}
		\PAR
		\COMLINE{pop top message and send it to another sequential process}								\label{queues:line4}
		\sendL{\head{\msgs}}\ .\ \Qmsg{\tail{\msgs}}					\label{queues:line5}
		\COMLINE{or receive and store an incoming message}				\label{queues:line7}
		\STATE $+$\,  \receive{\msg}\ . \Qmsg{\append{\msg}{\msgs}}                     \label{queues:line8}
		\ENDPAR
	\ENDIFii

	\end{algorithmic}
    }
  \end{algorithm}

\vspace{-4mm}

 
\subsection{Invariants}\label{subapp:invariants}
\spnewtheorem{assumption}{Assumption}{\bfseries}{\rm}
\setcounter{proposition}{0}
\renewcommand{\theproposition}{\thesection.\arabic{proposition}}
\setcounter{corollary}{0}
\renewcommand{\thecorollary}{\thesection.\arabic{corollary}}
\setcounter{theorem}{0}
\renewcommand{\thetheorem}{\thesection.\arabic{theorem}}
\setcounter{definition}{0}
\renewcommand{\thedefinition}{\thesection.\arabic{definition}}
\setcounter{remark}{0}
\renewcommand{\theremark}{\thesection.\arabic{remark}}

We now analyse our timed version of AODV.
We will go through the propositions proved for AODV without time in
\cite{TR13,GHPT16}---up to the proof of loop freedom---and
check whether they still hold. Most propositions still hold and the
proofs are, mutatis mutandis, the same as the proofs for AODV without time.
Changes mainly concern line numbers, and the changes triggered by the
introduction of the new functions that can modify the routing table, namely
\hyperlink{settimert}{$\fnsettimert$} and
\hyperlink{exprt}{$\fnexprt$}, as well as the function that can modify the
set of route request identifiers, namely
\hyperlink{exprreqs}{$\fnexprreqs$}.  The modification of the other
functions and of data types are mainly to include the role of time; they
do not modify their roles.

A transition $N\ar{\tau}N'$ between two network expressions may arise from a transition
$\colonact{R}{\starcastP{m}}$ performed by a network node \dval{ip},
synchronising with receive actions of all nodes $\dval{dip}\in R$
in transmission range.
In this case, we write $N\ar{R:\starcastP{\dval{m}}}_\dval{ip}N'$.
This means that $N=[M]$ and $N'=[M']$ are network expressions such
that \plat{$M\ar{R:\starcastP{m}}M'$}, and the cast action is performed by
node \dval{ip}.  This transition stems ultimately from an action
$\broadcastP{\dexp{ms}}$, $\groupcastP{\dexp{dests}}{\dexp{ms}}$, or
$\unicast{\dexp{dest}}{\dexp{ms}}$ (cf.\ \Sect{process_algebra}). Each
such action can be identified by a line number in one of the processes
of \SSubApp{ModelAODV}.

With $\xiN{\dval{ip}}(\keyw{var})$ we denote the evaluation
$\xi(\keyw{var})$ of the variable $\keyw{var}$ maintained by node
$\dval{ip}$ when AODV is in state $N$---see \cite[Section 7.2]{TR13}
or \cite[Section 6.2]{GHPT16} for further explanation.

In \ref{ssubapp:data} we have defined functions that work on evaluated
routing tables $\xiN{\dval{ip}}(\rt)$, such as $\fnnhop$.
To ease readability, we abbreviate
\plat{$\nhop{\xiN{\dval{ip}}(\rt)}{\!\dval{dip}}$}~by \plat{$\nhp{\dval{ip}}$}.
Similarly, we use 
\plat{$\sq{\dval{ip}}\hspace{-.5pt}$}, \plat{$\dhp{\dval{ip}}\hspace{-.5pt}$}, \plat{$\sta{\dval{ip}}\hspace{-.5pt}$},
\plat{$\lti{\dval{ip}}\hspace{-.5pt}$},
\plat{$\kd{\dval{ip}}$}, \plat{$\akd{\dval{ip}}$}
and \plat{$\ikd{\dval{ip}}$} for
\plat{$\sqn{\xiN{\dval{ip}}(\rt)}{\dval{dip}}$}, \plat{$\dhops{\xiN{\dval{ip}}(\rt)}{\dval{dip}}$},
\plat{$\status{\xiN{\dval{ip}}(\rt)}{\dval{dip}}$,}
\plat{$\ltime{\xiN{\dval{ip}}(\rt)}{\dval{ip}}$},
\plat{$\kD{\xiN{\dval{ip}}(\rt)}$}, \plat{$\akD{\xiN{\dval{ip}}(\rt)}$} and \plat{$\ikD{\xiN{\dval{ip}}(\rt)}$}, respectively.

\subsubsection{Basic Properties}

\begin{proposition}\label{prop:preliminaries}\rm\cite[Proposition~7.1]{TR13}
  \begin{enumerate}[(a)]
  \item\label{it:preliminariesi} With the exception of new packets
    that are submitted to a node by a client of AODV, every message
    received and handled by the main routine of AODV has to be sent by
    some node before.\label{before} More formally, we consider an
    arbitrary path $N_0\ar{\ell_1}N_1\ar{\ell_2} \ldots \ar{\ell_k}
    N_k$ with $N_0$ an initial state in our model of AODV\@. If the
    transition $N_{k-1}\ar{\ell_k} N_k$ results from a synchronisation
    involving the action $\receive{\msg}$ from Line~\ref{aodv:line2}
    of Pro.~\ref{pro:aodv}---performed by the node \dval{ip}---, where
    the variable {\msg} is assigned the value $m$, then either
    $m=\newpkt{\dval{d}}{\dval{dip}}$ or one of the $\ell_i$ with
    $i<k$ stems from an action $\starcastP{m}$ of a node $\dval{ip}'$
    of the network.
  \item\label{it:preliminariesii} No node can receive a message
    directly from itself.  Using the formalisation above, we must have
    $\dval{ip}\neq\dval{ip}'$.
  \end{enumerate}
\end{proposition}

\begin{proof}
  Exactly as in \cite{TR13}.  (The process $\QMSG$ has not been changed.)
\qed
\end{proof}

\begin{proposition}\rm\cite[Proposition~7.2]{TR13}
  \label{prop:invarianti_itemiii}
  The sequence number of any given node $\dval{ip}$ increases
  monotonically, i.e., never decreases, and is never unknown.  That
  is, for $\dval{ip}\mathbin\in\IP$, if $N \ar{\ell} N'$ then
  $1\leq\xiN{\dval{ip}}(\sn)\leq\xiN[N']{\dval{ip}}(\sn)$.
\end{proposition}

\begin{proof}
  Exactly as in \cite{TR13}. There are no additional methods to modify
  a node's own sequence number.
\qed
\end{proof}

\begin{remark}\label{rem:remark}
  Most of the forthcoming proofs can be done by showing the statement
  for each initial state and then checking all locations in the
  processes where the validity of the invariant is possibly changed.
  Note that routing table entries are only changed by the functions
  \hyperlink{update}{\fnupd},
  \hyperlink{invalidate}{$\fninv$},
  $\fnaddprecrt$,
  \hyperlink{settimert}{$\fnsettimert$} or
  \hyperlink{exprt}{$\fnexprt$}.\footnote{The functions \hyperlink{settimert}{$\fnsettimert$} or
  \hyperlink{exprt}{$\fnexprt$} are added w.r.t.\ \cite[Remark~7.3]{TR13}.}
  Thus we have to show that an
  invariant dealing with routing tables is satisfied after the
  execution of these functions if it was valid before.  In our proofs,
  we go through all occurrences of these functions.  In case the
  invariant does not make statements about precursors, the function
  $\fnaddprecrt$ need not be considered.
\end{remark}

Proposition 7.4 in \cite{TR13} says that
  the set of known destinations of a node increases monotonically.
  That is, for $\dval{ip}\in\IP$, if \plat{$N \ar{\ell} N'$}
  then $\kd{\dval{ip}}\subseteq\kd[N']{\dval{ip}}$.
This proposition no longer holds since \hyperlink{exprt}{$\fnexprt$} can remove routing table
entries.

Proposition 7.5 in \cite{TR13} says that
  the set of already seen route requests of a node increases
  monotonically.  That is, for $\dval{ip}\mathbin\in\IP$, if \plat{$N \ar{\ell} N'$}\vspace{-1pt} then
  $\xiN{\dval{ip}}(\rreqs)\subseteq\xiN[N']{\dval{ip}}(\rreqs)$.
This proposition no longer holds since the function
\hyperlink{exprreqs}{$\fnexprreqs$} prunes the list of route request
seen by a node.

Proposition 7.6 in \cite{TR13} says that
  in each node's routing table, the sequence number for any given
  destination increases monotonically, i.e., never decreases.
This proposition no longer holds since routing table entries can be
removed and recreated with an inferior sequence number.
However, we have the following weakening. 

\begin{proposition}\label{prop:dsn increase}\rm
  In each node's routing table, the sequence number for any given destination, as long as it is not
  deleted, increases monotonically.
  That is, for $\dval{ip},\dval{dip}\mathbin\in\IP$, if $N \ar{\ell} N'$ and
  $\dval{dip}\mathbin\in\kd{\dval{ip}} \cap \kd[N']{\dval{ip}}$ then
  $\sq{\dval{ip}}\mathbin\leq\fnsqn_{N'}^\dval{ip}(\dval{dip})$.
\end{proposition}
\begin{proof}
  Identical to the proof of Proposition~7.6 in \cite{TR13}.
\qed
\end{proof}

The next invariant tells that each node is correctly informed about
its own identity.
\begin{proposition}\rm\label{prop:self-identification}\cite[Proposition~7.7]{TR13}
  For each $\dval{ip}\in\IP$ and each reachable state $N$ we have
  $\xiN{\dval{ip}}(\ip)=\dval{ip}$.
\end{proposition}

\begin{proof}
  Exactly as in \cite{TR13}; there are no modifications of the variable \ip.
\qed
\end{proof}

\begin{proposition}\label{prop:ip=ipc}\rm\cite[Proposition~7.8]{TR13}
  If an AODV control message is sent by node $\dval{ip}\in\IP$, the
  node sending this message identifies itself correctly:
  \begin{equation*}
    N\ar{R:\starcastP{m}}_{\dval{ip}}N' \ims \dval{ip}=\ipc\ ,
  \end{equation*}
  where the message $m$ is either $\rreq{*}{*}{*}{*}{*}{*}{*}{\ipc}$,
  $\rrep{*}{*}{*}{*}{*}{\ipc}$, or $\rerr{*}{\ipc}$.
\end{proposition}

\begin{proof}
  Exactly as in \cite{TR13}, using \Prop{self-identification}.
\qed
\end{proof}

\begin{corollary}\rm\label{cor:sipnotip}\cite[Corollary~7.9]{TR13}
  At no point will the variable $\sip$ maintained by node \dval{ip}
  have the value \dval{ip}.
  \begin{equation*}
    \xiN{\dval{ip}}(\sip)\neq\dval{ip}
  \end{equation*}
\end{corollary}

\begin{proof}
  The same proof as in \cite{TR13}, mutatis mutandis (different line numbers).
\end{proof}

\begin{proposition}\rm\label{prop:positive hopcount}\cite[Proposition~7.10]{TR13}
  All routing table entries have a hop count greater than or equal to $1$.
  \begin{equation}\label{eq:inv_length}
    (*,*,*,*,\dval{hops},*,*,*)\in\xiN{\dval{ip}}(\rt) \ims \dval{hops}\geq1
  \end{equation}
\end{proposition}

\begin{proof}
  Essentially the same proof as in \cite{TR13}, following \Rem{remark}.
  We have to consider the new functions
  \hyperlink{settimert}{$\fnsettimert$} and \hyperlink{exprt}{$\fnexprt$}
  and change the line numbers. $\fnsettimert$ does not
  modify the hop count. $\fnexprt$ either removes the entry or leaves the
  hop count unchanged. In both cases, the invariant is preserved.
\qed
\end{proof}

\begin{proposition}\label{prop:starcastNew}\rm\cite[Proposition~7.11]{TR13}
  \begin{enumerate}[(a)]
  \item If a route request with hop count $0$ is sent by a node
    $\ipc\in\IP$, the sender must be the originator.
    \begin{equation}\label{inv:starcast_i}
      N\ar{R:\starcastP{\rreq{0}{*}{*}{*}{*}{\oipc}{*}{\ipc}}}_{\dval{ip}}N' \ims \oipc=\ipc(=\dval{ip})
    \end{equation}
  \item If a route reply with hop count $0$ is sent by a node
    $\ipc\in\IP$, the sender must be the destination.
    \begin{equation}\label{inv:starcast_i_rrep}
      N\ar{R:\starcastP{\rrep{0}{\dipc}{*}{*}{*}{\ipc}}}_{\dval{ip}}N' \ims \dipc=\ipc(=\dval{ip})    
    \end{equation}
  \end{enumerate}
\end{proposition}

\begin{proof}
  The same proof as in \cite{TR13}, mutatis mutandis; it uses \Prop{positive hopcount}.
\qed
\end{proof}

\begin{proposition}\rm\label{prop:rte}\cite[Proposition~7.12]{TR13}
  \begin{enumerate}[(a)]
  \item\label{it_a} Each routing table entry with $0$ as its
    destination sequence number has a sequence-number-status flag
    valued unknown.
    \begin{equation}\label{eq:unk_sqn}
      (\dval{dip},0,\dval{f},*,*,*,*,*)\in\xiN{\dval{ip}}(\rt) \ims \dval{f}=\unkno
    \end{equation}
  \item\label{it_d} Unknown sequence numbers can only occur at $1$-hop
    connections.
    \begin{equation}\label{eq:inv_viia}
      (*,*,\unkno,*,\dval{hops},*,*,*)\in\xiN{\dval{ip}}(\rt) \ims \dval{hops}=1
    \end{equation}
  \item\label{it_c} $1$-hop connections must contain the destination
    as next hop.
    \begin{equation}\label{eq:inv_vii}
      (\dval{dip},*,*,*,1,\dval{nhip},*,*)\in\xiN{\dval{ip}}(\rt)\ims \dval{dip}=\dval{nhip}
    \end{equation}
  \item\label{it_e} If the sequence number $0$ occurs within a routing
    table entry, the hop count as well as the next hop can be
    determined.
    \begin{equation}\label{eq:inv_viib}
      (\dval{dip},0,\dval{f},*,\dval{hops},\dval{nhip},*,*)\in\xiN{\dval{ip}}(\rt)\ims \dval{f}{=}\unkno \wedge \dval{hops}{=}1\wedge \dval{dip}{=}\dval{nhip}
    \end{equation}
  \end{enumerate}
\end{proposition}

\begin{proof}
  At the initial states all routing tables are empty. Since
  \hyperlink{settimert}{$\fnsettimert$}, \hyperlink{exprt}{$\fnexprt$} and \hyperlink{addprert}{$\fnaddprecrt$}
neither decrease the sequence number nor change
  the sequence-number-status flag,  the next hop or the hop count of a routing
  table entry, they cannot invalidate any of the above invariants. The function \hyperlink{invalidate}{$\fninv$}
  changes neither the sequence-number-status flag, nor the next hop or
  the hop count, but could decrease the sequence number of an entry.
  The proof in \cite{TR13} points to \cite[Proposition 7.6]{TR13} to
  show that this cannot happen. Here this follows from \Prop{dsn increase}.
  For this reason, we still can disregard applications of $\fninv$.
  Hence, as in \cite{TR13}, one only has to look at the application calls of \hyperlink{update}{\fnupd}.
  The remainder of the proof follows \cite{TR13}, mutatis mutandis.
  It uses \Prop{starcastNew}.
\qed
\end{proof}

\begin{proposition}\label{prop:msgsendingii}\rm\cite[Proposition~7.13]{TR13}
  \begin{enumerate}[(a)]
  \item Whenever an originator sequence number is sent as part of a
    route request message, it is known, i.e., it is greater than or equal to $1$.
    \begin{equation}\label{inv:starcast_sqni}
      N\ar{R:\starcastP{\rreq{*}{*}{*}{*}{*}{*}{\osnc}{*}}}_{\dval{ip}}N' \ims \osnc \geq1
    \end{equation}
  \item Whenever a destination sequence number is sent as part of a
    route reply message, it is known, i.e., it is greater than or equal to $1$.
    \begin{equation}\label{inv:starcast_sqnii}
      N\ar{R:\starcastP{\rrep{*}{*}{\dsnc}{*}{*}{*}}}_{\dval{ip}}N'\ims\dsnc\geq1
    \end{equation}
  \end{enumerate}
\end{proposition}

\begin{proof}\!
 {Just} as in \cite{TR13}, mutatis mutandis,
  using Propositions~\ref{prop:preliminaries},~\ref{prop:invarianti_itemiii} and~\ref{prop:rte}.\!\!\!\!\!\qed
\end{proof}
\begin{proposition}\label{prop:msgsending}\rm\cite[Proposition~7.14]{TR13}
  \begin{enumerate}[(a)]
  \item\label{prop:msgsendingRREQ} If a route request is sent
    (forwarded) by a node $\ipc$ different from the originator of the
    request then the content of $\ipc$'s routing table must be fresher
    or at least as good as the information inside the message.
    \begin{equation}\label{inv:starcast_ii}
      \begin{array}{rcl}
        &&N\ar{R:\starcastP{\rreq{\hopsc}{*}{*}{*}{*}{\oipc}{\osnc}{\ipc}}}_{\dval{ip}}N'
        \ans\ipc\neq\oipc\\
        &\Rightarrow&
        \oipc\in\kd{\ipc}
        \ans\big(\sq[\oipc]{\ipc}>\osnc\ors (\sq[\oipc]{\ipc}=\osnc 
        \\
        &&\ans \dhp[\oipc]{\ipc}\leq\hopsc\ans \sta[\oipc]{\ipc}=\val)\big)
      \end{array}
    \end{equation}
  \item If a route reply is sent by a node $\ipc$, different from the
    destination of the route, then the content of $\ipc$'s routing
    table must be consistent with the information inside the message.
    \begin{equation}\label{inv:starcast_iv}
      \begin{array}{@{}rcl@{}}
        &&N\ar{R:\starcastP{\rrep{\hopsc}{\dipc}{\dsnc}{*}{*}{\ipc}}}_{\dval{ip}}N'
        \ans\ipc\neq\dipc\\
        &\Rightarrow&
        \dipc\in\kd{\ipc}
        \ans \sq[\dipc]{\ipc} = \dsnc\\
        &&\ans  \dhp[\dipc]{\ipc}=\hopsc\ans \sta[\dipc]{\ipc}=\val
      \end{array}
    \end{equation}
  \end{enumerate}
\end{proposition}
\begin{proof}
  The same proof as in \cite{TR13}, mutatis mutandis, using \Prop{self-identification}.
\end{proof}
\Prop{msgsending} states facts about the network state at the end of a \textbf{*cast}-action. 
Since the evaluation function of a node does not change while transmitting a message (taking 
a $\tau$-action stemming from rules \tableref{bc}, \tableref{gc} and \tableref{uc} of \Tab{sos sequential},
an $\Rw\w$-action or an $R:\textbf{*cast}$-action), a similar proposition can be shown for all 
such actions.
%
If we consider the start of a transmission (the $\tau$-action), we can even strengthen the proposition. 
We show this only for the case of unicasting a RREP message (the strengthened version of 
\Prop{msgsending}(b)).

\begin{proposition}\label{prop:msgsendingstartuni}\rm
If the sending of a route reply is initiated by a node $\ipc$, different from the
    destination of the route, then the content of $\ipc$'s routing
    table must be consistent with the information inside the message, 
    including the  lifetime field. Moreover the sequence number is known (the sequence-number-status flag is set to \kno).
    \begin{equation}\label{inv:starcast_lifetime}
      \begin{array}{@{}rcl@{}}
        &&N\ar{\unicast{*}{\rrep{\hopsc}{\dipc}{\dsnc}{*}{*}{\ipc}}}_{\dval{ip}}N'
        \ans\ipc\neq\dipc\\
        &\Rightarrow&
        \dipc\in\kd{\ipc}
        \ans \sq[\dipc]{\ipc} = \dsnc\\
        &&\ans  \dhp[\dipc]{\ipc}=\hopsc\ans \sta[\dipc]{\ipc}=\val\\
        &&\ans \sqnf{\xiN{\ipc}(\rt)}{\dipc}=\kno \ans \lti[\dipc]{\ipc} > \xiN{\ipc}(\now)\ ,
 \end{array}
    \end{equation}
    where the label is a new notation indicating that a transition stemming from rule \tableref{uc}
    with $\dval{ms}=\rrep{\hopsc}{\dipc}{\dsnc}{*}{*}{\ipc}$
    is taken by node \dval{ip}.
\end{proposition}

\begin{proof}
  The same proof as for Proposition 7.14(b) in \cite{TR13}, mutatis mutandis.
  The last line, however, did not occur in \cite{TR13}.
  We extend the proof to justify this addition.
    \begin{description}
       \item[Pro.~\ref{pro:rreq}, Line~\ref{rreq:line14a}:] As in \cite{TR13} a new route reply with
		$\ipc:=\xi(\ip)=\dval{ip}$ is initiated.
		Moreover, by Line~\ref{rreq:line10}, $\dipc:=\xi(\dip)=\xi(\ip)=\dval{ip}$ and 
		thus $\ipc=\dipc$.
		Hence, the antecedent of \eqref{inv:starcast_lifetime} is not satisfied.
      \item[Pro.~\ref{pro:rreq}, Line~\ref{rreq:line26}:]
        That $\sqnf{\xiN{\ipc}(\rt)}{\dipc}=\kno$ follows from Line~\ref{rreq:line22};
        that $\lti[\dipc]{\ipc} > \xiN{\ipc}(\now)$ follows from Line~\ref{aodv:line3} of
        Pro.~\ref{pro:aodv}, which is always executed prior to Line~\ref{rreq:line26} or Pro.~\ref{pro:rreq},
        in the same time slice.
      \item[Pro.~\ref{pro:rrep}, Line~\ref{rrep:line13}:]
        That $\sqnf{\xiN{\ipc}(\rt)}{\dipc}\mathbin=\kno$ and $\lti[\dipc]{\ipc} \mathbin> \xiN{\ipc}(\now)$
        follows from the $\fnupd$ done in Line~\ref{rrep:line2}, with Line~\ref{rrep:line1} ensuring its effect.
      \qed
     \end{description}
\end{proof}

\begin{proposition}\rm\label{prop:starcastrerr}\cite[Proposition~7.15]{TR13}
  Any sequence number appearing in a route error message
  stems from an invalid destination and is equal to the sequence
  number for that destination in the sender's routing table at the
  time of sending.
  \begin{equation}\label{inv:starcast_rerr}
        \begin{array}{@{}rcl@{}}
		    &&N\ar{R:\starcastP{\rerr{\destsc}{\ipc}}}_{\dval{ip}}N'\ans (\ripc,\rsnc)\in\destsc\\
    &\Rightarrow&
    \ripc\in\ikd{\dval{ip}} \ans \rsnc = \sq[\ripc]{\dval{ip}}
    \end{array}
  \end{equation}
\end{proposition}

\begin{proof}
  Same proof as in \cite{TR13}, mutatis mutandis.
\qed
\end{proof}

Propositions 7.16--7.25 in \cite{TR13} show that all partial functions
used in the specification of AODV are always defined when
they occur outside of an atomic formula (when an undefined function
call occurs in an atomic formula, that formula evaluates to \keyw{false}---cf.~Footnote~\ref{fn:undefvalues}).
The proofs, which use Propositions \ref{prop:preliminaries} and \ref{prop:msgsendingii}, apply to our timed model of AODV as well. Moreover, the
arguments for the new partial functions $\fnfD[\footnotesize]$ and $\fnseltime[\footnotesize]$ are
identical to the argument for $\sigma_{\mbox{\footnotesize\it p-flag}}$ in \cite[Proposition 7.25]{TR13}.

\subsubsection{The Quality of Routing Table Entries}\label{ssec:quality}~

\noindent In \cite[Section 7.5]{TR13} the \emph{net sequence number} of a route to a
destination $\dval{dip}$ in a routing table $\dval{rt}$ is defined by
\hypertarget{nsqn}{
\[\begin{array}{r@{\hspace{0.5em}}r@{\hspace{0.5em}}l}
  \fnnsqn : \tRT\times\tIP&\to& \tSQN\\
  \nsqn{\dval{rt}}{\dval{dip}}&:=&
    \left\{
      \begin{array}{ll}
        \sqn{\dval{rt}}{\dval{dip}}&\mbox{if }\status{\dval{rt}}{\dval{dip}}=\val \ors \sqn{\dval{rt}}{\dval{dip}}=0\\
        \sqn{\dval{rt}}{\dval{dip}}-1&\mbox{otherwise}\ .
      \end{array}
    \right.
  \end{array}\]
}
If two routing tables $\dval{rt}$ and $\dval{rt}'$ have a routing table entry to destination
$\dval{dip}$, i.e., $\dval{dip}\in\kD{\dval{rt}}\cap\kD{\dval{rt}'}$, 
they can be compared w.r.t.\ their \emph{quality} for that destination \cite{TR13}:
\[\begin{array}{r@{\hspace{0.5em}}r@{\hspace{0.5em}}l}
\dval{rt} \rtord \dval{rt}' &:\Leftrightarrow
&\nsqn{\dval{rt}}{\dval{dip}} < \nsqn{\dval{rt}'}{\dval{dip}} \ors\\
&&\big(\nsqn{\dval{rt}}{\dval{dip}} = \nsqn{\dval{rt}'}{\dval{dip}} \wedge
\dhops{\dval{rt}}{\dval{dip}} \geq \dhops{\dval{rt}'}{\dval{dip}}\big)
\end{array}
\]
\noindent For all destinations $\dval{dip}\in\IP$, 
the relation $\rtord$ on routing tables with an entry for \dval{dip} is a total preorder.
The equivalence relation induced by $\rtord$ is denoted by~$\rtequiv$.

\begin{proposition}\label{prop:qual}\rm\cite[Proposition~7.26]{TR13}
Assume a routing table $\dval{rt}\in\tRT$ with $\dval{dip}\in\kD{\dval{rt}}$.\vspace{-1ex}
\begin{enumerate}[(a)]
\item\label{it:qual_upd}
An \hyperlink{update}{\fnupd} of $\dval{rt}$ can only increase the quality of the routing table.
That is, for all routes \dval{r} such that \upd{\dval{rt}}{\dval{r}} is defined
(i.e., $\pi_{4}(\dval{r})=\val$,
 $\pi_{2}(\dval{r})=0\Leftrightarrow\pi_{3}(\dval{r})=\unkno$ and
 $\pi_{3}(\dval{r})=\unkno\Rightarrow\pi_{5}(\dval{r})=1$)
we have\vspace{-1ex}\vspace{-3pt}
  \begin{equation}\label{eq:qual_upd}
    \dval{rt}\rtord\upd{\dval{rt}}{\dval{r}}\ .
\vspace{-2pt}
  \end{equation}
\item\label{it:qual_inv}
An \hyperlink{invalidate}{$\fninv$} on $\dval{rt}$ does not change the
quality of the routing table if, for each $(\dval{rip},\dval{rsn})\in\dval{dests}$,
\dval{rt} has a valid entry for \dval{rip}, and
  \begin{itemize}
  \item \dval{rsn} is the by one incremented sequence number from the routing table, or 
  \item both \dval{rsn} and the sequence number in the routing table are $0$.
  \end{itemize}
That is, for all partial functions $\dval{dests}$ (subsets of $\tIP\times\tSQN$)
\vspace{-2pt}
  \begin{equation}\label{eq:qual_inv}
  \begin{array}{rcl}
  &&\big((\dval{rip},\dval{rsn})\in\dval{dests} \ims \dval{rip}\in\akD{\dval{rt}} \ans
  \dval{rsn}=\inc{\sqn{\dval{rt}}{\dval{rip}}}\big) \\
  &\Rightarrow&
    \dval{rt}\rtequiv\inv{\dval{rt}}{\dval{dests}}{*}\ .
  \end{array}
\vspace{-2pt}
  \end{equation}
\item\label{it:qual_addpre}
If precursors are added to an entry of $\dval{rt}$, the quality of the routing table
does not change.
 That is, for all $\dval{dip}\in\IP$ and sets of precursors $\dval{npre}\in\pow(\IP)$,
\vspace{-2pt}
  \begin{equation}\label{eq:qual_addpre}
    \dval{rt}\rtequiv\addprecrt{\dval{rt}}{\dval{dip}}{\dval{npre}}\ .
\vspace{-2pt}
  \end{equation}
\end{enumerate}
\end{proposition}

\begin{proof}
The same as in \cite{TR13}, using \Prop{positive hopcount}.
\qed
\end{proof}
Further, we have to prove that the applications of
\fnsettimert\ do not
decrease the quality of a routing table entry, nor do applications of \fnexprt\ that do not delete the
entry to \dval{dip}.
The first is straightforward since \fnsettimert\ only modifies time components.
The second follows since such applications leave both $\nsqn{\dval{rt}}{\dval{dip}}$ and
$\dhops{\dval{rt}}{\dval{dip}}$ invariant.

\newcommand{\nowN}[1][N]{\dval{now}_{#1}}

Theorem 7.27 of \cite{TR13} says that the quality of routing table
entries can never decrease. This result does not hold any longer, as the
entry may expire and reemerge with a lower quality. However, we do have
the following weakening of this result.

\begin{proposition}\label{prop:state_quality}\rm
As long as a routing table entry is not deleted, its quality can only be increased, never decreased.

Assume $N \ar{\ell}N'$ and $\dval{ip},\dval{dip}\mathbin\in\IP$.
If $\dval{dip}\in\kd{\dval{ip}} \cap \kd[N']{\dval{ip}}$ then
\[\xiN{\dval{ip}}(\rt)\rtord\xiN[N']{\dval{ip}}(\rt)\ .\]
\end{proposition}

\begin{proof}
By \Prop{qual}, and the remark following it, the quality of routing table entries, as long as they
are not deleted, cannot decrease due to applications of $\fnupd$, $\fnaddprecrt$, $\fnsettimert$ and
$\fnexprt$. Hence we only need to check all applications of $\fninv$. That proceeds exactly as 
the proof of Proposition 7.27 in \cite{TR13}.
\qed
\end{proof}
\Prop{state_quality} states in particular that if $N \ar{\ell}N'$ and
$\dval{dip}\in\kd{\dval{ip}} \cap \kd[N']{\dval{ip}}$ then
$\nsq{\dval{ip}} \leq \fnnsqn_{N'}^{\dval{ip}}(\dval{dip})$.

Proposition 7.28 and Theorem 7.30 of \cite{TR13} state relations
between the routing tables of different nodes. They are the key
results in establishing loop freedom. Both results do not hold here,
at least not unconditionally. However, a weakening of Proposition 7.28
and the full Theorem 7.30 hold if we assume that premature route
expiration does not occur. We formalise this assumption in two parts
as Assumptions~\ref{ass:pre} and~\ref{ass:underway} below.

For the value of the  variable $\now$ in state $N$, we write $\nowN$.
Assuming that in an initial state of AODV the clocks of all nodes
have the same value, this will continue to be the case throughout the
life of the protocol, since AODV does not modify this variable. 
Hence we do not need a superscript \dval{ip} to indicate which node's variable $\now$ is meant.
Moreover, $\nowN$ increases monotonically.

A route to \dval{dip} may be marked as \emph{valid} in the routing table of a node \dval{ip}, but
if $\lti{\dval{ip}}\leq\nowN$ or \hyperlink{1hoplife}{$\neg\keyw{1hoplife}(\nhp{\dval{ip}},\nowN)$}
its validity is questionable, and, following the RFC~\cite{rfc3561},
the routing table entry ought to be marked as \emph{invalid}.
Hence, before the routing table is
consulted, the function \fnexprt\ is always applied, making 
all valid routing table entries invalid that are timed out themselves, 
or have a timed-out routing table entry to the next hop.
We define the set of \emph{intrinsically valid} routing table entries of node \dval{ip} in state $N$ as
$\VD{\dval{ip}} := \{\dval{dip}\in\akd{\dval{ip}} \mid \lti{\dval{ip}}>\nowN \wedge \keyw{1hoplife}(\nhp{\dval{ip}},\nowN)
\} =\xiN{\dval{ip}}(\akD{\exprt{\rt}{\now}{\deltime}})$.

\begin{assumption}\label{ass:pre}
If a node has an intrinsically valid routing table entry to a destination \dval{dip}, then the next hop, if not
\dval{dip} itself, has a known route to \dval{dip}.
\begin{equation}\label{eq:pre}
\dval{dip}\in\VD{\dval{ip}} \ans \dval{nhip} := \nhp{\dval{ip}} \neq \dval{dip}
\ims \dval{dip}\mathbin\in\kd{\dval{nhip}}
\end{equation}
\end{assumption}
\hypertarget{underway}{To formalise the second part of the assumption that premature route
expiration does not occur, we first define what we mean by a message being \emph{underway}.
A message starts being underway when its transmission is initiated.
For a RREQ or RREP message this is between the states $N$ and $N'$ for which\vspace{-1pt}
\[
        N\ar{\broadcastP{\rreq********}}_{\dval{sip}}N'
\quad \mbox{or} \quad
        N\ar{\unicast{*}{\rrep******}}_{\dval{sip}}N'\vspace{-3pt}
\]
in the notation of \Prop{msgsendingstartuni}, with \dval{sip} being the sending node.
For a RREQ message, this indicates a transition stemming from Rule~\tableref{bc} in \Tab{sos sequential},
resulting\linebreak[3] from the execution of \Pro{aodv}, Line~\ref{aodv:line39} or \Pro{rreq}, Line~\ref{rreq:line34}.
When a message leaves the incoming message queue of the receiving node we still treat it as
underway until the receiving node makes 
``sufficient'' updates to its routing table triggered by the
receipt of the message, or when it becomes clear that an update is not going to happen:\footnote{%
By sufficient we mean enough changes for the invariants presented later to hold.}
a PKT message is underway until Line~\ref{pkt2:line2},~\ref{pkt2:line5},~\ref{pkt2:line19} or~\ref{pkt2:line21} of \Pro{pkt} is executed;
a RREQ message is underway until Line~\ref{rreq:line2} or~\ref{rreq:line6} of \Pro{rreq} is executed;
a RREP message is underway until Line~\ref{rrep:line2} or~\ref{rrep:line25} of \Pro{rrep} is executed;
and a RERR message until Line~\ref{rerr:line5} of \Pro{rerr} is executed}.%
\footnote{NEWPKT messages are not considered since they are not stored in the message queue.}
In each case exactly one of these lines will in fact be executed, and this happens in the same time
slice in which the message leaves the incoming message queue.

\begin{assumption}\label{ass:underway}
If a RREP  message with destination \dval{dip} or a RREQ message with originator \dval{dip}, sent
by a node $\dval{sip}\neq\dval{dip}$, is underway to a node \dval{ip}, then $\dval{dip} \in \kd{\dval{sip}}$.
%
\end{assumption}

\begin{proposition}\rm\label{prop:inv_nsqn}
Assume that premature route expiration does not occur (Assumptions~\ref{ass:pre} and~\ref{ass:underway}).
If, in a reachable network expression $N$, a node $\dval{ip}\mathop\in\IP$ has an intrinsically
valid routing table entry to $\dval{dip}$, then also the next hop
\dval{nhip} towards \dval{dip}, if not \dval{dip} itself, has a
routing table entry to $\dval{dip}$, and the net sequence number of
the latter entry is at least as large as that of the former.%
\begin{equation}\label{eq:inv_ix}
\dval{dip}\in\VD{\dval{ip}}\ans\dval{nhip}\not=\dval{dip}
\ims	\dval{dip}\in\kd{\dval{nhip}} \mathrel{\wedge} \nsq{\dval{ip}}\leq \nsq{\dval{nhip}}\ ,
\end{equation}
where $\dval{nhip}:=\nhp{\dval{ip}}$ is the IP address of the next hop.
\end{proposition}
Apart from its reliance on Assumptions~\ref{ass:pre} and~\ref{ass:underway}, this proposition weakens 
\cite[Proposition~7.28]{TR13} by assuming $\dval{dip}\in\VD{\dval{ip}}$ instead of $\dval{dip}\in\kd{\dval{ip}}$.

\begin{proof}
We can forget about the conclusion $\dval{dip}\in\kd{\dval{nhip}}$, since this follows by Assumption~\ref{ass:pre}.
For an initial network expression the invariant holds since all routing tables are empty.
We need to make sure that the invariant is maintained under all modifications to $\xiN{\dval{ip}}(\rt)$
or $\xiN{\dval{nhip}}(\rt)$, and under progress of time.

Progress of time cannot invalidate the invariant; at most it can invalidate the antecedent.

A modification of \plat{$\xiN{\dval{nhip}}(\rt)$} is harmless, as it
can only increase $\nsq{\dval{nhip}}$ (cf.\ \Prop{state_quality}).
The antecedent of the proposition ($\dval{dip}\in\kd{\dval{ip}} \cap \kd[N']{\dval{ip}}$)
follows from Assumption~\ref{ass:pre} and the fact that
\plat{$\dval{dip}\in\VD{\dval{ip}}$ and $\dval{nhip}\not=\dval{dip}$} holds both before
and after the modification of \plat{$\xiN{\dval{nhip}}(\rt)$}.

Applications of $\fnaddprecrt$ have no effect on the invariant.
Applications of \hyperlink{exprt}{$\fnexprt$} have no effect on the invariant either,
since $\fnexprt$ is idempotent, and net sequence numbers are not affected by $\fnexprt$.
Applications of \hyperlink{invalidate}{$\fninv$}
cannot invalidate the invariant; at most they can invalidate the antecedent.
Whenever $\fnsettimert$ is applied, $\fnexprt$ has been applied before in the same time slice.
For this reason, we have $\lti{\dval{ip}}>\nowN \wedge
\mbox{\hyperlink{1hoplife}{$\keyw{1hoplife}(\nhp{\dval{ip}},\nowN)$}}$, so $\fnsettimert$ cannot chance the condition
\plat{$\dval{dip}\in\VD{\dval{ip}}$} from false to true. It also has no further effect on the invariant.

Hence, it suffices to check all applications of \hyperlink{update}{\fnupd} that actually change a
routing table entry, beyond its precursors. This proceeds as in the proof of \cite[Proposition~7.28]{TR13},
using Propositions~\ref{prop:preliminaries},~\ref{prop:ip=ipc} and~\ref{prop:msgsending},
but with one refinement.

  In cases Pro.~\ref{pro:rreq}, Line~\ref{rreq:line6} and Pro.~\ref{pro:rrep}, Line~\ref{rrep:line2}
  we handle in a state $N$ a RREQ message with originator \dval{dip}, or a RREP message with
  destination \dval{dip}, that was sent by a node $\dval{nhip}\neq\dval{dip}$ in a state $N^\dagger$. The proof of
  \cite[Proposition~7.28]{TR13} then calls \cite[Theorem~7.27]{TR13} to infer that
  $\nsq{\dval{nhip}}\geq\fnnsqn_{N^\dagger}^{\dval{nhip}}(\dval{dip})$.
  Here we can instead use \Prop{state_quality}, but only under Assumption~\ref{ass:underway}.
\qed
\end{proof}
To prove loop freedom we will show that on any route established by AODV the quality 
of routing table entries increases when going from one node to the next hop. Here, the 
preorder is not sufficient, since we need a strict increase in quality. 
Therefore, on routing tables $\dval{rt}$ and $\dval{rt}'$ that both have an entry to $\dval{dip}$, i.e., $\dval{dip}\in\kD{\dval{rt}}\cap\kD{\dval{rt}'}$,
 we define a relation $\rtsord$ by
\[
\dval{rt} \rtsord \dval{rt}'\ :\Leftrightarrow\  \dval{rt} \rtord \dval{rt}' \ans \dval{rt} \not\rtequiv \dval{rt}'\ .
\]

\begin{corollary}\label{cor:strictord}\rm{}\hspace{-3.10028pt}\cite[\hspace{-1.5pt}Corollary 7.29]{TR13}\hspace{-1pt}
The relation $\rtsord$ is irreflexive and \mbox{transitive}. 
\end{corollary}

\begin{proposition}\rm\cite[Theorem~7.30]{TR13}
\label{prop:inv_a}
Assume that premature route expiration does not occur (Assumptions~\ref{ass:pre} and~\ref{ass:underway}).
The quality of the routing table entries for a destination \dval{dip} is strictly increasing
along a route towards \dval{dip},
 until it reaches either \dval{dip} or a node with an invalid routing table entry to \dval{dip}.
\begin{equation}\label{eq:inv_x}
\dval{dip}\in\akd{\dval{ip}}\cap \akd{\dval{nhip}} \ans\dval{nhip}\not=\dval{dip}
\ims \xiN{\dval{ip}}(\rt)\rtsord \xiN{\dval{nhip}}(\rt)\ ,
\end{equation}
where $N$ is a reachable network expression and $\dval{nhip}:=\nhp{\dval{ip}}$ is the IP address of the next hop.
\end{proposition}
\begin{proof}
For an initial network expression the invariant holds since all routing tables are empty.
We need to make sure that the invariant is maintained under all modifications to $\xiN{\dval{ip}}(\rt)$
or $\xiN{\dval{nhip}}(\rt)$. 
Applications of $\fnaddprecrt$ and
\hyperlink{settimert}{$\fnsettimert$} have no effect on the invariant.
Applications of \hyperlink{invalidate}{$\fninv$} and \hyperlink{exprt}{$\fnexprt$} cannot
invalidate the invariant; at most they can invalidate the ante\-cedent. Hence, it suffices to check
all applications of \hyperlink{update}{\fnupd} that change~a routing table entry, beyond its precursors,
just as in the proof of \cite[Theorem~7.30]{TR13}.

The argument that the invariant is maintained under updates of $\xiN{\dval{nhip}}(\rt)$ is unchanged
w.r.t.\ the proof of \cite[Theorem~7.30]{TR13}. It uses \Prop{rte}. 
At two occasions
this proof refers to \cite[Proposition~7.28]{TR13} (addressed as ``Invariant (20)''), and in both
cases a reference to \Prop{inv_nsqn} suffices as well.  Here, we use that each $\fnupd$
that is handled in a state $N$ is preceded by an application of $\fnexprt$ in the same time slice.
Hence $\dval{dip}\in\akd{\dval{ip}}$ implies $\dval{dip}\in\VD{\dval{ip}}$.

The argument that the invariant is maintained under updates of $\xiN{\dval{ip}}(\rt)$ is almost 
unchanged w.r.t.\ the proof of \cite[Theorem~7.30]{TR13}.
It uses Propositions~\ref{prop:preliminaries} and~\ref{prop:ip=ipc} and
Invariants~\eqref{inv:starcast_ii} and~\eqref{inv:starcast_iv}.
However, there are two occasions where the argument needs to be refined.

\begin{itemize}
\vspace{-1ex}
\item In the case Pro.~\ref{pro:rreq}, Line~\ref{rreq:line6} we handle in a state $N$ a RREQ message with
  originator \dval{dip} that was sent by a node $\dval{nhip}\neq\dval{dip}$ in a state $N^\dagger$. The proof of
  \cite[Theorem~7.30]{TR13} then calls \cite[Proposition~7.6]{TR13} to infer that\vspace{-2pt}
  $\sq{\dval{nhip}} \geq \fnsqn_{N^\dagger}^{\dval{nhip}}(\dval{dip})$.
  Here we can use \Prop{dsn increase}, but only under Assumption~\ref{ass:underway}.
\item In cases Pro.~\ref{pro:rreq}, Line~\ref{rreq:line6} and Pro.~\ref{pro:rrep}, Line~\ref{rrep:line2}
  we handle in a state $N$ a RREQ message with originator \dval{dip}, or a RREP message with
  destination \dval{dip}, that was sent by a node $\dval{nhip}\neq\dval{dip}$ in a state $N^\dagger$. The proof of
  \cite[Theorem~7.30]{TR13} then calls \cite[Theorem~7.27]{TR13} to infer that
  $\xiN[N^\dagger]{\dval{nhip}}(\rt) \rtord \xiN[N]{\dval{nhip}}(\rt)$.
  Here we can instead use \Prop{state_quality}, but only under Assumption~\ref{ass:underway}.
\qed
\vspace{-1ex}
\end{itemize}
\end{proof}

\noindent

\newcommand{\RG}[2]{\mathcal{R}_{#1}(#2)}
\begin{definition}\rm\cite{TR13}
The \emph{routing graph} of network expression $N$ with respect to
$\dval{dip}\in\tIP$ is $\RG{N}{\dval{dip}}\mathop{:=}\linebreak[1](\IP,E)$, where
all nodes of the network form the vertices and there is an
arc $({\dval{ip}},{\dval{ip}}')\mathbin\in E$ iff $\dval{ip}\mathop{\not=}\dval{dip}$ and
$
(\dval{dip}\comma\nosp{*}\comma\nosp{*}\comma\nosp{\val}\comma\nosp{*}\comma\nosp{\dval{ip}'}\comma\nosp{*}\comma\nosp{*})\mathop{\in}\xiN{\dval{ip}}(\rt).
$

We say that a network expression $N$ is \emph{loop free} if the
corresponding routing graphs $\RG{N}{\dval{dip}}$ are loop free, for
all $\dval{dip}\mathop{\in}\IP$. A routing protocol, such as AODV, is
\emph{loop free} iff all reachable network expressions are loop free.
\end{definition}
An arc in a routing graph states that $\dval{ip}'$ is the next hop on
a valid route to $\dval{dip}$ known by $\dval{ip}$; a path in a routing
graph describes a route towards $\dval{dip}$ discovered by AODV\@.

Using this definition of a routing graph, \Prop{inv_a} states that 
along a path towards a destination \dval{dip} in the routing
graph of a reachable network expression $N$, until it reaches either
\dval{dip} or a node with an invalided routing table entry to dip,
the quality of the routing table entries for \dval{dip} is strictly increasing.
From this, we can immediately conclude

\begin{theorem}\rm\label{thm:loop free}
Assume that premature route expiration does not occur (Assumptions~\ref{ass:pre} and~\ref{ass:underway}).
Then the specification of AODV given in \SubApp{spec} is loop free.
\qed
\end{theorem}


\subsection{Premature Route Expiration}\label{subapp:D}
By \Thm{loop free}, to establish loop freedom for AODV it suffices to show that premature route
expiration cannot occur. In view of the counterexample to loop freedom sketched in \Fig{loop1},
this condition appears necessary as well.
In this appendix we do an attempt to prove an invariant that implies that premature route
expiration, and hence routing loops, do not occur in AODV\@.
In this process we formalise postulates on the real-time behaviour of the protocol that need to be
made in order to have any chance on success. Even when assuming these, our invariant turns out not
to be preserved by 5 lines of the AODV specification. As documented in
\SubApp{loops}, each of these violations gives rise to premature route expiration, and consequently to routing loops. Additionally, for our proof to go
through, we need to make an assumption (Assumption~\ref{ass:rreqvalid} below) that does not hold
for AODV\@. The key to modifying AODV into a
loop free variant is (1) to make a small change that validates Assumption~\ref{ass:rreqvalid}, and
(2) to change the above-mentioned 5 lines in such a way that the intended invariant is maintained.

\begin{assumption}\label{ass:rreqvalid}
When a RREQ message with originator \dval{oip} is sent by a node $\dval{sip}\neq\dval{oip}$,
the node \dval{sip} has a valid routing table entry to \dval{oip}.
\end{assumption}
We will mark results that depend on this assumption by \highlightspec{(A3)}.
When applying Assumption~\ref{ass:rreqvalid} we will also use that in the state $N^\dagger$ where
the transmission of the above RREQ message commences (by the execution of transition \tableref{bc} of \Tab{sos sequential})
the valid routing table entry to \dval{oip} satisfies \plat{$\fnltime_{N^\dagger}^{\dval{sip}}(\dval{oip}) >\nowN[N^\dagger]$}.
This follows because the forwarding of the RREQ message is always preceded by an application of
$\fnexprt$ in the same time slice (Line~\ref{aodv:line3} of \Pro{aodv}).

\begin{proposition}\label{prop:starcast}\rm\cite[Proposition~7.36c]{TR13}
  The sequence number of an originator appearing in a route request can never be greater than the
  originator's own sequence number.
\begin{equation}\label{inv:starcast_v}
	  N\ar{R:\starcastP{\rreq{*}{*}{*}{*}{*}{\oipc}{\osnc}{*}}}_{\dval{ip}}N'
          \ims \osnc\leq\xiN{\oipc}(\sn)
\end{equation}
\end{proposition}
\begin{proof}
Exactly as in \cite{TR13}, using Propositions~\ref{prop:preliminaries},~\ref{prop:invarianti_itemiii}
and~\ref{prop:self-identification}.
\qed
\end{proof}

\begin{proposition}\rm\label{prop:dsn}\cite[Proposition~7.37]{TR13}
\begin{enumerate}[(a)]
\item The sequence number of a destination appearing in a route reply can never be greater than the
  destination's own sequence number.
  \begin{equation}\label{eq:dsn_rrep}
    N\ar{R:\starcastP{\rrep{*}{\dipc}{\dsnc}{*}{*}}}_{\dval{ip}}N'
    \ims \dsnc\leq\xiN{\dipc}(\sn)
  \end{equation}
\item A known destination sequence number of a valid routing table entry can never be greater than the
  destination's own sequence number.
  \begin{equation}\label{eq:dsn}
    (\dval{dip},\dval{dsn},\kno,\val,*,*,*)\in\xiN{\dval{ip}}(\rt)
    \ims \dval{dsn}\leq\xiN{\dval{dip}}(\sn)
  \end{equation}
\end{enumerate}
\end{proposition}

\begin{proof}
Exactly as in \cite{TR13}, using Propositions~\ref{prop:preliminaries},~\ref{prop:invarianti_itemiii},
\ref{prop:self-identification} and \ref{prop:starcast}.
\qed
\end{proof}

\begin{proposition}\rm\label{prop:nhip lifetime}\highlightspec{(A3)}
Let $N^\ddagger$ be a state in which the own sequence number maintained by node \dval{dip} is
incremented to the value \dval{dsn}, and let $N$ be a state in which a node \dval{ip} has a
\emph{valid} routing table entry to \dval{dip} with next hop $\dval{nhip} \neq \dval{dip}$ and a
destination sequence number $\dval{dsn}'\geq\dval{dsn}$. Then $\nowN \geq \nowN[N^\ddagger]$ and
\begin{align}\label{eq:lifetime1}
  \dval{dip}\in\akd{\dval{nhip}} &\ims \lti{\dval{nhip}} \geq \nowN[N^\ddagger]\;, \\
  \dval{dip}\in\ikd{\dval{nhip}} &\ims \lti{\dval{nhip}} \geq \nowN[N^\ddagger]+\deltime\;.
\end{align}
\end{proposition}

\begin{proof}
Using proof by contradiction, we show that  the sequence number $\dval{dsn}'$ is known, i.e.,\
$\sqnf{\xiN{\dval{ip}}(\rt)}{\dval{dip}}\linebreak[2]=\kno$. If we were to assume
$\sqnf{\xiN{\dval{ip}}(\rt)}{\dval{dip}}=\unkno$, then $\dhops{\xiN{\dval{ip}}(\rt)}{\dval{dip}} = 1$ and hence 
$\dval{nhip}=\nhop{\xiN{\dval{ip}}(\rt)}{\dval{dip}} = \dval{dip}$, both by \Prop{rte};
a contradiction to the assumption $\dval{nhip} \neq \dval{dip}$.
Hence
\[
  \xiN{\dval{dip}}(\sn)\geq \dval{dsn}' \geq \dval{dsn} = \xiN[N^{\ddagger}]{\dval{dip}}(\sn)\;,
\]
where the first step follows from \Eq{dsn}.
Since sequence numbers increase over time (\Prop{invarianti_itemiii}) and $N^{\ddagger}$ is the state where 
$\dval{dsn}$ is set, we get $\nowN \geq \nowN[N^\ddagger]$.

The invariants hold in initial states, as all routing tables are empty.
Applications of $\fnaddprecrt$ and \hyperlink{settimert}{$\fnsettimert$}
cannot invalidate the invariants. Neither can applications of \hyperlink{invalidate}{\fninv} or \hyperlink{exprt}{$\fnexprt$}
to the routing table of \dval{ip}, or an \hyperlink{update}{\fnupd} to the routing table of \dval{nhip}.
An application of $\fnexprt$ to the routing table of \dval{nhip} that invalidates the
antecedent $\dval{dip}\in\akd{\dval{nhip}}$ but validates $\dval{dip}\in\ikd{\dval{nhip}}$ always
results in a state where the lifetime of the routing table entry is extended by \deltime.
An application of \fninv\ to the routing table of \dval{nhip} that invalidates the
antecedent $\dval{dip}\in\akd{\dval{nhip}}$ always results in a state where
\plat{$\lti{\dval{nhip}} = \nowN+\deltime \geq \nowN[N^\ddagger]+\deltime$}.
It remains to examine all applications of {\fnupd} to the routing table of \dval{ip},
restricting attention to updates that change more than precursors.\vspace{-1ex}
\begin{description}
 \item[Pro.~\ref{pro:aodv}, Lines~\ref{aodv:line10},~\ref{aodv:line14},~\ref{aodv:line18}:]
   After these updates the condition $\dval{nhip} \neq \dval{dip}$ is no longer met.
\item[Pro.~\ref{pro:rreq}, Line~\ref{rreq:line6}:]
   If this update results in a change to the routing table, beyond the addition of precursors,
   afterwards\vspace{-2pt} $\dval{nhip}:=\fnnhop_N^{\dval{ip}}(\dval{dip})=\xiN{\dval{ip}}(\sip)\neq\dval{dip}:=\xiN{\dval{ip}}(\oip)$
   and $\dval{dsn}':=\fnsqn_N^{\dval{ip}}(\dval{dip})=\xiN{\dval{ip}}(\osn)$
   are taken from the sender and sequence-number fields of the incoming RREQ message that is being processed here.
   (The inequation of \dval{nhip} and \dval{dip} is an assumption.)
   Let $N^\#$ be the state in which node \dval{dip} initiated this route request, and thus
   incremented its own sequence number to the value $\dval{dsn}'\geq \dval{dsn}$.
   Then $\nowN[N^\#] \mathbin\geq \nowN[N^\ddagger]$ by \Prop{invarianti_itemiii}.
   By Propositions~\ref{prop:preliminaries}(\ref{it:preliminariesi}) and~\ref{prop:ip=ipc}, the RREQ
   message must have been forwarded by $\dval{nhip}$. Let $N^\dagger$ be the state in which the
   transmission of the forwarded RREQ message by node \dval{nhip} commenced. Obviously,
   $\nowN[N^\dagger] \mathbin\geq \nowN[N^\#]$.
   Assumption~\ref{ass:rreqvalid} yields \plat{$\dval{dip}\in\akd[N^\dagger]{\dval{nhip}}$}.
   Before node \dval{nhip} forwarded the route request (by executing Line~\ref{rreq:line34} of
   Pro.~\ref{pro:rreq}), and in the same time slice,
   it must have executed Line~\ref{aodv:line3} of Pro.~\ref{pro:aodv}, so that
   \plat{$\fnltime_{N^\dagger}^{\dval{nhip}}(\dval{dip}) \mathbin> \nowN[N^\dagger]$}.
   Hence $\fnltime_{N^\dagger}^{\dval{nhip}}(\dval{dip}) \mathbin> \nowN[N^\ddagger]$.
   Further modifications to the routing table of \dval{nhip} (by $\fnaddprecrt$, $\fnsettimert$,
   $\fnupd$, $\fnexprt$ and $\fninv$) between states $N^\dagger$ and $N$ preserve the invariant in the ways surveyed above.
 \item[Pro~\ref{pro:rrep}, Line~\ref{rrep:line2}:]
   If this update results in a change to the routing table, beyond the addition of precursors,
   afterwards $\dval{nhip}:=\fnnhop_N^{\dval{ip}}(\dval{dip})=\xiN{\dval{ip}}(\sip)\neq\dval{dip}=\xiN{\dval{ip}}(\dip)$ 
   and $\dval{dsn}':=\fnsqn_N^{\dval{ip}}(\dval{dip})=\xiN{\dval{ip}}(\dsn)$
   are taken from the sender and sequence-number fields of the incoming RREP message that is being processed here.
   By Propositions~\ref{prop:preliminaries}(\ref{it:preliminariesi}) and~\ref{prop:ip=ipc}, this RREP
   message must have been sent before by $\dval{nhip}$; say its transmission started in state $N^\dagger$.
   \Prop{msgsendingstartuni} yields \plat{$\dval{dip}\in\akd[N^\dagger]{\dval{nhip}}$} and\vspace{-1ex}
   \[\fnsqn_{N^\dagger}^{\dval{nhip}}(\dval{dip}) \mathbin= \dval{dsn}' \wedge\vspace{-1ex}
   \sqnf{\xiN[N^\dagger]{\dval{nhip}}(\rt)}{\dval{dip}}\mathbin=\kno \wedge
   \fnltime_{N^\dagger}^{\dval{nhip}}(\dval{dip}) \mathbin> \nowN[N^\dagger].\]
   Hence, by \Prop{dsn}(b), $\nowN[N^\dagger] \mathbin\geq \nowN[N^\ddagger]$.
   So $\fnltime_{N^\dagger}^{\dval{nhip}}(\dval{dip}) \mathbin> \nowN[N^\ddagger]$.
   Further modifications to the routing table of \dval{nhip} (by $\fnaddprecrt$, $\fnsettimert$,
   $\fnupd$, $\fnexprt$ and $\fninv$) between states $N^\dagger$ and $N$ preserve the invariant in the ways surveyed above.
\qed
\end{description}
\end{proof}

We can only show the absence of premature route expiration
under further assumptions. In particular, we postulate the
following relations between time constants.
Henceforth the timing parameters $\nodetraversal$ and $\traversal$ will be abbreviated by
{\renewcommand{\nodetraversal}{\keyw{NODE\_TT}}%
\renewcommand{\traversal}{\keyw{NET\_TT}}%
$\nodetraversal$ and $\traversal$.
\begin{align}
0 \leq \nodetraversal & \leq  \traversal                      \label{eq:node-net}\stepcounter{Hequation}\\
0 \leq \valtime & <  \deltime - \nodetraversal - \traversal   \label{eq:valtime} \stepcounter{Hequation}\\
0 \leq \myroute & <  \deltime - \nodetraversal - \traversal   \label{eq:myroute} \stepcounter{Hequation}\\
3\cdot\traversal & < \deltime + \nodetraversal        \label{eq:traversal}\stepcounter{Hequation}
\end{align}}%
These conditions are in line with the RFC:
\cite[Section 10]{rfc3561} recommends:\vspace{-2ex}
\begin{center}
\begin{tabular}{l@{\qquad}r}
\nodetraversal & $40\,ms$\phantom{${}^{{26}}$}\\
\traversal & $2\cdot \nodetraversal\cdot \keyw{NET\_DIAMETER}$\footnotemark\\
\valtime   & $10.000\,ms\footnotemark$\\
\myroute   & $2  \cdot \valtime$\phantom{${}^{{26}}$}\\
\deltime   & $5 \cdot \valtime$\phantom{${}^{{26}}$}
\end{tabular}
\end{center}
\addtocounter{footnote}{-1}
\footnotetext{The default value of \keyw{NET\_DIAMETER} is 35,
    yielding a {\nodetraversal} of 2800\,ms.}
\addtocounter{footnote}{1}
\footnotetext{When link-layer indications are used to detect link breakages
(rather than Hello messages) \cite[Section 10]{rfc3561}, as we assume here; otherwise 3000\,ms.}


\begin{proposition}\rm\label{prop:B}~\vspace{-1ex}
  \begin{enumerate}[(a)]
  \item The expiration time of a valid route is always smaller than\\ $\nowN + \deltime - \nodetraversal - \traversal$.
\begin{align}\label{eq:vD_leq_time}\stepcounter{Hequation}
\renewcommand{\nodetraversal}{\keyw{NODE\_TT}}%
\renewcommand{\traversal}{\keyw{NET\_TT}}%
\begin{array}{cl}
     &\dval{dip}\mathbin\in\akd{\dval{ip}}\\
      \ims& \lti{\dval{ip}} < \nowN {+} \deltime - \nodetraversal - \traversal
    \end{array}   
\end{align}
  \item The lifetime recorded in a route reply message is always smaller than\\ $\deltime - \nodetraversal - \traversal$.
\begin{align}\label{eq:ltime_leq_time}\stepcounter{Hequation}
\renewcommand{\nodetraversal}{\keyw{NODE\_TT}}%
\renewcommand{\traversal}{\keyw{NET\_TT}}%
 \begin{array}{cl}
     & N\ar{R:\starcastP{\rrep{*}{*}{*}{*}{\dval{ltime}}{*}}}_{\dval{ip}}N'\\
     \ims & \dval{ltime} < \deltime - \nodetraversal - \traversal
       \end{array}
    \end{align}
  \end{enumerate}
\end{proposition}
\renewcommand{\nodetraversal}{\keyw{NODE\_TT}}%
\renewcommand{\traversal}{\keyw{NET\_TT}}%

\begin{proof}
We prove the two statements by simultaneous induction.
\vspace{-1ex}
  \begin{enumerate}[(a)]
  \item The invariant holds in the initial states, as all routing tables are empty.
    The functions \hyperlink{invalidate}{$\fninv$}, $\fnaddprecrt$,
    and \hyperlink{exprt}{$\fnexprt$} cannot increase the lifetime of a valid routing table entry, without invalidating the entry.
    It therefore suffices to check whether the invariant is preserved under
    the applications of \hyperlink{update}{\fnupd} and \hyperlink{settimert}{$\fnsettimert$} in Processes~\ref{pro:aodv}--\ref{pro:rerr}.
    (\Pro{queues} does not use these functions.)
\begin{description}
\item[%
  Pro.~\ref{pro:aodv}, Lines~\ref{aodv:line10},~\ref{aodv:line14},~\ref{aodv:line18},~\ref{aodv:line26b},~\ref{aodv:line26c};
  Pro.~\ref{pro:pkt}, Lines~\ref{pkt2:line6a},~\ref{pkt2:line6b},~\ref{pkt2:line6c},~\ref{pkt2:line6d};]~\\
{\bf
  Pro.~\ref{pro:rreq}, Line~\ref{rreq:line6};
  Pro.~\ref{pro:rrep}, Line~\ref{rrep:line12d}:}
If these potential changes to routing table entries increase  the lifetime of a node at all, the expiration time is set to
   $\nowN+\valtime$. That the invariant is preserved follows from \Eq{valtime}.
 \item[Pro.~\ref{pro:aodv}, Line~\ref{aodv:line35b}:]
   As this affects an invalid route ($\dval{dip}\mathbin\in\qD{\dval{store}}{-}\akD{\dval{rt}}$, by Line~\ref{aodv:line34}), the invariant is preserved.
 \item[Pro.~\ref{pro:pkt}, Line~\ref{pkt2:line19}:]
    As this affects an invalid route, the invariant is preserved.
 \item[Pro.~\ref{pro:rreq}, Line~\ref{rreq:line7}:]
   Here $\lti{\dval{ip}}$ is set to
   $\nowN+2\cdot\traversal\linebreak[1]- 2\cdot(\dval{hops}{+}1)\cdot\nodetraversal$.
   Since $\dval{hops}\in\NN$, the result follows from \Eq{traversal}.
 \item[Pro~\ref{pro:rrep}, Line~\ref{rrep:line2}:]
   Here $\lti{\dval{ip}}$ is set to $\nowN+\dval{ltime}$, where the value \dval{ltime} stems from an
   incoming RREP message (Pro.~\ref{pro:aodv}, Line~\ref{aodv:line2}).
   By \Prop{preliminaries}(\ref{it:preliminariesi}), this RREP message must have be sent before by
   some node. By induction, using \Eq{ltime_leq_time}, the invariant holds.
\end{description}
  \item We check all occasions in Processes~\ref{pro:aodv}--\ref{pro:queues} where a route reply is sent.
    \begin{description}
      \item[\Pro{rreq}, Line~\ref{rreq:line14a}:]
        Here \dval{ltime} is set to $\myroute$, so the result follows from \Eq{myroute}.
      \item[Pro.~\ref{pro:rreq}, Line~\ref{rreq:line26}:]
        \dval{ltime} is set to $\lti{\dval{ip}}-\nowN$.
        Hence the invariant holds by induction, using statement (a) of the lemma.
      \item[Pro.~\ref{pro:rrep}, Line~\ref{rrep:line13}:]
        Here the value \dval{ltime} is taken from an incoming RREP message.
        By \Prop{preliminaries}(\ref{it:preliminariesi}), this RREP message must have be sent before by
        some node. Hence the statement follows by induction.
      \qed
     \end{description}
  \end{enumerate}
\end{proof}


\newcommand{\prepkt}[1][\dval{nhip}]{\keyw{pkt}_{N}^{#1}(\dval{dip})}
\newcommand{\atime}[1][\dval{nhip}]{\keyw{atime}_N^{#1}(\dval{dip})}

As indicated in \SSect{lf}, we now capture realistic network scenarios by assuming
that the transmission time of a message plus the period it spends in the queue of incoming messages
of the receiving node is bounded by \nodetraversal.
Since {\nodetraversal} ``is a conservative estimate
   of the average one hop traversal time for packets and should include
   queuing delays, interrupt processing times and transfer times''~\cite[Sect.~10]{rfc3561}, 
 the following postulate makes sense.

\spnewtheorem{postulate}{Postulate}{\bfseries}{\rm}
\begin{postulate}\label{post:nodetraversal}
Let $N^\dagger$ be a state in which the transmission of a message
 to \dval{ip} is initiated,
and let $N$ be the state in which the message leaves the queue of incoming messages of node \dval{ip}.
Then $\nowN \leq \dval{now}_{N^\dagger}+\nodetraversal$.
\end{postulate}
Likewise, we assume that the period a route request travels through the network is bounded by \traversal.
\begin{postulate}\label{post:nettraversal}
Let $N^\ddagger$ be a state in which a route request is initiated, and
$N$ a state in which a corresponding RREQ message
leaves the queue of incoming messages of an arbitrary node \dval{ip}.
Then $\nowN \leq \dval{now}_{N^\ddagger}+\traversal$.
\end{postulate}
Together with Assumption~\ref{ass:rreqvalid}, these postulates are strong enough to ensure the validity of Assumption~\ref{ass:underway}.

A similar statement as \Post{nettraversal} could be set up for route replies; 
it is, however, not needed for the current analysis.

\begin{theorem}\label{thm:A3impliesA2}\rm\highlightspec{(A3)}
Assumption~\ref{ass:underway} holds.
\end{theorem}

\begin{proof}
Suppose in state $N$ a RREP message that establishes a route to \dval{dip}, sent
by a node $\dval{sip}\neq\dval{dip}$, is underway to a node \dval{ip}.
Let $N^\dagger$ be the state in which the transmission of the message was initiated.
By \Post{nodetraversal}, $\nowN \leq \dval{now}_{N^\dagger}+\nodetraversal$.
   In state~$N^\dagger$ node $\dval{sip}$ had a valid routing table entry to $\dval{dip}$, with an
   expiration time larger than $\nowN[N^{\dagger}]$, by \Prop{msgsendingstartuni},
   i.e., $\dval{dip}\in\akd[N^{\dagger}]{\dval{sip}}$ and  $\ltiN[\dval{dip}]{\dval{sip}}{N^{\dagger}} > \nowN[N^\dagger]$.
   Upon invalidation of an entry, the expiration time is always either set to $\now+\deltime$ or
   extended by $\deltime$.
   Since $\nodetraversal < \deltime$, by \Eq{valtime},
   the routing table entry for \dval{dip} cannot have expired in state $N$.

The case for a RREQ message proceeds likewise, but using Assumption~\ref{ass:rreqvalid} and the remark following it, instead of \Prop{msgsendingstartuni}.
\qed
\end{proof}

Write $\prepkt$ if a data packet for destination \dval{dip} is \hyperlink{underway}{underway} 
(from some node \dval{sip}) to
node \dval{nhip} conform the definition given prior to Assumption~\ref{ass:underway}.
Moreover, if $\prepkt$, write $\atime$ 
for the latest possible time the next data packet destined to
\dval{dip} will arrive at node \dval{nhip} confirm the prediction of \Post{nodetraversal}.
It follows that\vspace{-1ex}
\setcounter{Hequation}{29}
\begin{equation}\label{eq:atime}
\prepkt
\ims \nowN \leq \atime \leq \nowN {+} \nodetraversal\;.
\end{equation}
The following ``intended theorem'' ensures that also Assumption~\ref{ass:pre} holds.
This is a trivial corollary of the two invariants proposed below.
Hence, if the intended theorem would hold, loop freedom follows.
However, the invariant turns out not to be preserved under 5 lines of the AODV specification, as
made clear by the last line in the following ``intended proof''.

\spnewtheorem{wtheorem}[theorem]{Intended Theorem}{\bfseries}{\rm}
\begin{wtheorem}\label{thm:intended}\rm~\vspace{-1ex}\highlightspec{(A3)}\label{main}
  \begin{enumerate}[(a)]
  \item
If a data packet destined for \dval{dip} is underway to node \dval{nhip},
then \dval{nhip} has a routing table entry to \dval{dip}
that will not expire before (or upon) arrival of that (first) data packet.
\begin{align}\label{eq:32}\stepcounter{Hequation}
\hspace{-18pt}
\begin{array}{@{}cl}
    &  \prepkt[\dval{nhip}] \wedge \dval{nhip}\neq\dval{dip}\\
\Rightarrow & (\dval{dip}\mathbin\in\akd{\dval{nhip}}\wedge \lti{\dval{nhip}}>\atime - \deltime)\\
    & \mbox{} \vee ( \dval{dip}\mathbin\in\ikd{\dval{nhip}}\wedge \lti{\dval{nhip}}>\atime )
\end{array}
\end{align}
  \item
If a node \dval{ip} has a valid routing table entry to a destination
\dval{dip} with expiration time $\dval{ltime} > \nowN$,
and no data packet is underway to \dval{nhip}---the next hop towards \dval{dip}---then \dval{nhip},
if not \dval{dip} itself, has a valid entry to \dval{dip}
with expiration time $> \dval{ltime} + \nodetraversal - \deltime$,
or an invalid one with expiration time $> \dval{ltime} + \nodetraversal$.
\begin{align}\label{eq:33}\stepcounter{Hequation}
\hspace{-18pt}
\renewcommand{\deltime}{\keyw{DELETE\_PRD.}}
\begin{array}{@{}cl}
     & \dval{dip}\in\akd{\dval{ip}} \ans \lti{\dval{ip}}>\nowN \ans \dval{nhip}\neq\dval{dip} \ans \neg \prepkt[\dval{nhip}]\\
\Rightarrow & (\dval{dip}\mathbin\in\akd{\dval{nhip}}\wedge \lti{\dval{nhip}}\mathbin>\lti{\dval{ip}} {+} \nodetraversal {-}\deltime) \\
     & \mbox{} \vee ( \dval{dip}\mathbin\in\ikd{\dval{nhip}}\wedge \lti{\dval{nhip}}\mathbin>\lti{\dval{ip}} {+} \nodetraversal )
\end{array}
\end{align}
where $\dval{nhip}:=\nhp{\dval{ip}}$ is the IP address of the next hop
towards \dval{dip}.
\end{enumerate}
\end{wtheorem}

\begin{proof}
We prove the two statements by simultaneous induction.
Both invariants hold in initial states, as no packet is underway and all routing tables are
empty.\vspace{-1ex}%
  \begin{enumerate}[(a)]
  \item We have to check that the invariant is preserved under
(i) changes that validate the condition \plat{$\prepkt[\dval{nhip}]$},
(ii) changes to the routing table of node \dval{nhip},
and (iii) changes that increase the value of \plat{$\atime$}.

\begin{enumerate}[\hspace{-15pt}\bf(i)]
\item
Let $\dval{nhip}\neq\dval{dip}$.
The only way the condition \plat{$\prepkt[\dval{nhip}]$} can turn valid
is when a node \dval{ip} executes Line~\ref{aodv:line24} of Pro.~\ref{pro:aodv}
or Line~\ref{pkt2:line6} of Pro.~\ref{pro:pkt},
with $\xiN{\dval{ip}}(\dip)\mathbin=\dval{dip}$ and $\nhp{\dval{ip}}=\dval{nhip}$,
and \plat{$\prepkt[\dval{nhip}]$} did not hold before.
Right beforehand, node \dval{ip} must have executed Line~\ref{aodv:line22} of Pro.~\ref{pro:aodv}
or Line~\ref{pkt2:line5} of Pro.~\ref{pro:pkt}; hence $\dval{dip}\in\akd{\dval{ip}}$.
Before that, \dval{ip} must have executed Line~\ref{aodv:line1} of Pro.~\ref{pro:aodv}, so that
\plat{$\lti{\dval{ip}}>\nowN$}.
By induction, invariant \Eq{33} yields\vspace{-1ex}
$$\begin{array}{l}
\renewcommand{\deltime}{\keyw{DELETE\_PRD.}}
(\dval{dip}\mathbin\in\akd{\dval{nhip}}\wedge \lti{\dval{nhip}}\mathbin>\lti{\dval{ip}} {+} \nodetraversal {-}\deltime) \\
 \ \vee\ ( \dval{dip}\mathbin\in\ikd{\dval{nhip}}\wedge \lti{\dval{nhip}}\mathbin>\lti{\dval{ip}} {+} \nodetraversal ).
\end{array}$$
This holds just before \plat{$\prepkt[\dval{nhip}]$} turned valid, and hence also just after.
By \Eq{atime}, $\lti{\dval{ip}} + \nodetraversal > \nowN + \nodetraversal\linebreak[1] \geq \atime$,
so the invariant is maintained.
\item
We now examine changes to the routing table of node \dval{nhip}.
These could be made by the functions 
\hyperlink{update}{\fnupd},
\hyperlink{invalidate}{$\fninv$},
$\fnaddprecrt$,
\hyperlink{settimert}{$\fnsettimert$} or
\hyperlink{exprt}{$\fnexprt$}.
An $\fnupd$ cannot make a valid entry invalid, erase an invalid
entry, or shorten the lifetime of an entry. For this reason, the invariant is maintained under
applications of $\fnupd$. The same applies to applications of $\fnsettimert$.
Applications of $\fnaddprecrt$ have no impact on the invariant.

\hspace{1em}
If the routing table entry to \dval{dip} is invalidated by $\fninv$, the expiration
time of the entry is always set to $\nowN+\deltime$.
Assuming that $\prepkt[\dval{nhip}] \wedge \dval{nhip}\neq\dval{dip}$,
by (\ref{eq:atime},\ref{eq:valtime}), $\atime \leq \nowN + \nodetraversal < \nowN+\deltime = \lti{\dval{nhip}}$.

\hspace{1em}
If the routing table entry to \dval{dip} is invalidated by $\fnexprt$, the expiration
time of the entry is extended by \deltime. This preserves the invariant.

\hspace{1em}
Finally consider the erasure of an entry by $\fnexprt$.
Suppose that right afterwards, and thus also right before, we have 
$\prepkt[\dval{nhip}] \wedge \dval{nhip}\neq\dval{dip}$.\vspace{1pt}
Then, by induction and \Eq{atime}, \plat{$\lti{\dval{nhip}}>\atime\geq \nowN$} holds
when $\fnexprt$ is applied to an invalid route to $\dval{dip}$, or
\begin{align*}
\lti{\dval{nhip}} &>\atime-\deltime\\
&\geq \nowN-\deltime
\end{align*}
when it is applied to an valid one,
so the route is not deleted by $\fnexprt$.

\item 
The only event that can increase the value of $\atime$ is the arrival at node \dval{nhip} of a
data packet destined for \dval{dip}, when another data packet for \dval{dip} is already underway.
When this happens, first Lines~\ref{aodv:line1},~\ref{aodv:line2} and~\ref{aodv:line6} of Pro.~\ref{pro:aodv}
are executed, with $\xi(\dip)=\dval{dip}$.
Then either $\dval{nhip} = \dval{dip}$, so that the invariant remains satisfied,
or Line~\ref{pkt2:line4} of Pro.~\ref{pro:pkt}, with $\xi(\ip)=\dval{nhip}$, is executed in the same
time slice. Assume the latter. In case \dval{nhip} has a valid routing table entry for \dval{dip},
since $\fnexprt$ has been executed in the same time slice, 
$\renewcommand{\deltime}{\keyw{DELETE\_PRD}}
\lti{\dval{nhip}}>\nowN \geq \atime - \nodetraversal > \atime - \deltime$
by (\ref{eq:atime},\ref{eq:valtime}), so the invariant remains satisfied.

\hspace{1em}
Otherwise, Line~\ref{pkt2:line15} of Pro.~\ref{pro:pkt} is executed
in the same time slice. In that case, applying induction on the state
just after this line, $\dval{dip}\in\ikd{\dval{nhip}}$, so Lines~\ref{pkt2:line18} and~\ref{pkt2:line19} of
Pro.~\ref{pro:pkt} are executed in the same time slice.
This results in a state where \plat{$\lti{\dval{nhip}}$} has at least the value $\nowN+\deltime$.
The execution of Line~\ref{pkt2:line19} marks the arrival of the data packet, and thus the state
chance where the value of \plat{$\atime$} is increased. Right afterwards,
$\lti{\dval{nhip}} \geq \nowN+\deltime > \nowN+\nodetraversal \geq \atime$, using
\Eq{valtime} and \Eq{atime}. Hence the invariant is maintained.

\end{enumerate}
  \item We have to check that the invariant is preserved under
(i) changes that validate the condition \plat{$\neg\prepkt[\dval{nhip}]$},
(ii) changes to the routing table of node \dval{nhip},
and (iii) changes to the routing table of node \dval{ip}.
As \plat{$\nowN$} is monotonically increasing, changes to $\nowN$ cannot invalidate the invariant.

\begin{enumerate}[\hspace{-15pt}\bf(i)]
\item Starting with $\prepkt[\dval{nhip}]$,
suppose that $\dval{dip}\mathbin\in\akd{\dval{ip}} \wedge \lti{\dval{ip}}\mathbin>\nowN$,
and a data packet destined for \dval{dip} is handled by node \dval{nhip},
in the sense that Lines~\ref{aodv:line2} and~\ref{aodv:line6} of Pro.~\ref{pro:aodv} are
executed with $\xi(\dip)\mathbin=\dval{dip}$.
Then either $\dval{nhip} = \dval{dip}$, so that the invariant remains satisfied,
or Line~\ref{pkt2:line4} of Pro.~\ref{pro:pkt}, with $\xi(\ip)=\dval{nhip}$, is executed in the same time slice.
Assuming the latter, in case \dval{nhip} has a valid routing table entry for \dval{dip},
since $\fnexprt$ has been executed in the same time slice, 
$\lti{\dval{nhip}}>\nowN > \lti{\dval{ip}} + \nodetraversal -\deltime$,
by \Prop{B}, so the invariant remains satisfied.

\hspace{1em}
Otherwise ($\dval{dip}\not\in\akd{\dval{nhip}}$), Line~\ref{pkt2:line15} of Pro.~\ref{pro:pkt} is executed in the same time slice.
In that case, applying induction on the state just before the data packet arrived,
invariant \Eq{32} yields $\dval{dip}\mathbin\in\ikd{\dval{nhip}}\!\!$,
so Lines~\ref{pkt2:line18} and~\ref{pkt2:line19} of Pro.~\ref{pro:pkt} are executed in the same time slice.
This results in a state where $\lti{\dval{nhip}}$ has at least the value $\nowN+\deltime$.
The execution of Line~\ref{pkt2:line19} marks the arrival of the data packet, and thus the state
chance where the the condition \plat{$\neg\prepkt[\dval{nhip}]$} becomes valid. Right afterwards,
$\lti{\dval{ip}} + \nodetraversal < \nowN + \deltime$, by \Prop{B}.
Hence the invariant is maintained.

\item We now examine changes to the routing table of node \dval{nhip}.
These could be made by the functions 
\hyperlink{update}{\fnupd},
\hyperlink{invalidate}{$\fninv$},
$\fnaddprecrt$,
\hyperlink{settimert}{$\fnsettimert$} or
\hyperlink{exprt}{$\fnexprt$}.
An $\fnupd$ cannot make a valid entry invalid, erase an invalid
entry, or shorten the lifetime of an entry. For this reason, the invariant is maintained under
applications of $\fnupd$. The same applies to applications of $\fnsettimert$.
Applications of $\fnaddprecrt$ have no impact on the invariant.

\hspace{1em}
If the routing table entry to \dval{dip} is invalidated by $\fninv$, its expiration time
$\lti{\dval{nhip}}$ is always set to $\nowN+\deltime$.
Using the assumption $\dval{dip}\in\akd{\dval{ip}}$ and Equation \Eq{vD_leq_time}, we get $\lti{\dval{ip}} + \nodetraversal < \nowN + \deltime$.
Hence the invariant is maintained.

\hspace{1em}
If the routing table entry to \dval{dip} is invalidated by $\fnexprt$, the expiration
time of the entry is extended by \deltime. This preserves the invariant.

\hspace{1em}
Finally consider the erasure of an entry by $\fnexprt$.
Suppose that right afterwards, and thus also right before,
the antecedent of the invariant holds.
Then, by induction,\vspace{3pt}

$\begin{array}{l}
\phantom{\ \vee\ }(\dval{dip}\mathbin\in\akd{\dval{nhip}}\wedge \lti{\dval{nhip}}>\lti{\dval{ip}} - \deltime) \\
\ \vee\ ( \dval{dip}\mathbin\in\ikd{\dval{nhip}}\wedge \lti{\dval{nhip}}>\lti{\dval{ip}} ).
\end{array}$\\[2pt]
Using that $\lti{\dval{ip}} > \nowN$, the route is not deleted by $\fnexprt$.

\item We conclude with changes to the routing table of node \dval{ip}.
Clearly the invariant is maintained under applications of $\fninv$, $\fnaddprecrt$ and $\fnexprt$.
We now go though all occurrences of $\fnupd$ and $\fnsettimert$ in Processes~\ref{pro:aodv}--\ref{pro:queues}.
\begin{description}
 \item[Pro.~\ref{pro:aodv}, Lines~\ref{aodv:line10},~\ref{aodv:line14},~\ref{aodv:line18}:]
   These entries create or update a routing table entry with $\dval{nhip}=\dval{dip}$, so the antecedent of
   the invariant is not met\journalonly{\chP{, in case the routing table is updated; otherwise the invariant remains unchanged}}.
 \item[Pro~\ref{pro:rreq}, Line~\ref{rreq:line6}:]
   If this update results in a change to the routing table, beyond the addition of precursors,
   afterwards \plat{$\dval{oip}\in\akd{\dval{ip}}$} and
   $\dval{nhip}:=\fnnhop_N^{\dval{ip}}(\dval{oip})=\xiN{\dval{ip}}(\sip)$ is the sender of the incoming RREQ
   message that is being processed here.
   We may assume that $\dval{nhip}\neq\dval{oip}$, as otherwise the invariant is maintained.
   By \Prop{preliminaries}(\ref{it:preliminariesi}), this RREQ
   message must have been sent before by $\dval{nhip}$.
   Let  $N^\dagger$ be the state in which the transmission of the message was initiated (by the
   execution of transition \tableref{bc} of \Tab{sos sequential}).
   By \Post{nodetraversal} and \Eq{node-net}
   $\dval{now}_{N^\dagger} \leq \nowN \leq \dval{now}_{N^\dagger}+\nodetraversal \leq\linebreak[3] \dval{now}_{N^\dagger}+\traversal$.
   In state $N^\dagger$, node $\dval{nhip}$ had a valid routing table entry to $\dval{oip}$, with a
   positive
   remaining lifetime, i.e., $\dval{oip}\in\fnakD^{\dval{nhip}}_{N^\dagger}(\dval{oip})$ and
   $\fnltime^{\dval{nhip}}_{N^\dagger}(\dval{oip}) > \nowN[N^{\dagger}]$, by Assumption~\ref{ass:rreqvalid}   \highlightspec{(A3)} 
   and the remark following it. By \Eq{vD_leq_time}, using that $\dval{oip}\in\akd{\dval{ip}}$
   and $\dval{now}_{N^\dagger} \geq \nowN - \traversal$, it follows that the condition\vspace{-2pt}
\begin{align}\label{eq:34}\stepcounter{Hequation}
\renewcommand{\deltime}{\keyw{DEL\_PRD.}}
\begin{array}{l}
     (\dval{oip}\mathbin\in\akd[N^\#]{\dval{nhip}}\wedge \ltiN[\dval{oip}]{\dval{nhip}}{N^\#}\mathbin>\lti[\dval{oip}]{\dval{ip}} {+} \nodetraversal {-}\deltime) \\
     \mbox{} \vee ( \dval{oip}\mathbin\in\ikd[N^\#]{\dval{nhip}}\wedge \ltiN[\dval{oip}]{\dval{nhip}}{N^\#}\mathbin>\lti[\dval{oip}]{\dval{ip}} {+} \nodetraversal )\vspace{-11pt}
\end{array}
\end{align}
   holds in state $N^\#:=N^\dagger$. To see that it still holds in state $N^\#:=N$, we argue that it
   is preserved under changes to the routing table of node \dval{nhip} between states $N^\dagger$ and $N$.
   Since the state $N$ is fixed in \Eq{34}, the value
   $\lti[\dval{oip}]{\dval{ip}}$ does not change, so 
   only changes to $\dval{oip}\mathbin\in\akd[N^\#]{\dval{nhip}}$, $\dval{oip}\mathbin\in\ikd[N^\#]{\dval{nhip}}$ and 
   $\ltiN[\dval{oip}]{\dval{nhip}}{N^\#}$ need to be considered.
These could be made by the functions 
\hyperlink{update}{\fnupd},
\hyperlink{invalidate}{$\fninv$},
$\fnaddprecrt$,
\hyperlink{settimert}{$\fnsettimert$} or
\hyperlink{exprt}{$\fnexprt$}.
An $\fnupd$ cannot make a valid entry invalid, erase an invalid
entry, or shorten the lifetime of an entry. 
For this reason, \Eq{34}
is maintained under applications of $\fnupd$. 
The same applies to applications of $\fnsettimert$.
Applications of $\fnaddprecrt$ have no impact on \Eq{34} either.

\hspace{1em}
If the routing table entry to \dval{oip} is invalidated by $\fninv$, its expiration time
$\ltiN[\dval{oip}]{\dval{nhip}}{N^\#}$ is set to $\nowN[N^\#]+\deltime$.
So Equation \Eq{vD_leq_time}, applied to $\dval{oip}\mathop\in\akd{\dval{ip}}$, yields
$\ltiN[\dval{oip}]{\dval{nhip}}{N^\#} = \nowN[N^\#]+\deltime \geq \nowN[N^\dagger]+\deltime >
\nowN +\deltime\linebreak[1] {-}\traversal \mathbin>\lti{\dval{ip}} {+} \nodetraversal$.
Hence invariant \Eq{34} is maintained.

\hspace{1em}
Using the antecedent of \Eq{33}, $\lti[\dval{oip}]{\dval{ip}} \geq \nowN \geq \nowN[N^\#]$,
so \Eq{34} implies that the routing table entry to \dval{oip} cannot be deleted by $\fnexprt$.
If the routing table entry to \dval{oip} is invalidated by $\fnexprt$, its expiration
time $\ltiN[\dval{oip}]{\dval{nhip}}{N^\#}$ is extended by \deltime. This preserves \Eq{34}.
 \item[Pro~\ref{pro:rrep}, Line~\ref{rrep:line2}:]
   The argument is exactly as in the previous case, but using RREP instead of RREQ and dip instead of oip.
   Moreover, we call \Prop{msgsendingstartuni} instead of Assumption~\ref{ass:rreqvalid}.
 \item[Pro.~\ref{pro:aodv}, Line~\ref{aodv:line26b}; Pro.~\ref{pro:pkt}, Line~\ref{pkt2:line6a}:]
   When this instruction is executed, a data
   packet is underway to $\dval{nhip}:=\fnnhop_N^{\dval{ip}}(\dval{dip})$ (Pro.~\ref{pro:aodv},
   Line~\ref{aodv:line24}, or Pro.~\ref{pro:pkt}, Line~\ref{pkt2:line6}, resp.),
   so the antecedent of the invariant is not satisfied.
 \item[Pro.~\ref{pro:aodv}, Line~\ref{aodv:line35b}; Pro.~\ref{pro:pkt}, Line~\ref{pkt2:line19}:]
   As this affects an invalid route, the invariant is preserved.
 \item[Pro.~\ref{pro:rreq}, Line~\ref{rreq:line7}:]
   Let $\dval{nhip}:=\fnnhop_N^{\dval{ip}}(\dval{oip})$ be the next hop to \dval{oip} (before
   and after the the call of {\fnsettimert}), and let $\dval{osn}:=\fnsqn_N^{\dval{ip}}(\dval{oip})$ be the
   destination sequence number of this route. Then $\dval{osn}\geq\xiN{\dval{ip}}(\osn)$,
   where \plat{$\xiN{\dval{ip}}(\osn)$} is the sequence number for \dval{oip} carried in the route request.
   Let $N^\ddagger$ be the state in which node
   $\dval{oip}$ initiated the route request, and thus incremented its own sequence number to \plat{$\xiN{\dval{ip}}(\osn)$}.
   By \Post{nettraversal}, $\nowN \leq \nowN[N^\ddagger] + \traversal$.
   Assume $\dval{oip}\in\akd{\dval{ip}} \wedge \fnltime^{\dval{ip}}_{N}(\dval{oip}) >\nowN \wedge 
   \dval{nhip}\neq\dval{oip} \wedge \neg \keyw{pkt}_{N}^{\dval{nhip}}(\dval{oip})$
   as otherwise the\linebreak[1] invariant is maintained.
   Then, by induction, we have \plat{$\dval{oip}\mathbin\in\kd{\dval{nhip}}$}
   right before the update, so also right afterwards.

\hspace{1em}
   Suppose $\dval{oip}\in\akd{\dval{nhip}}$. Then, by \Prop{nhip
     lifetime} \highlightspec{(A3)} and the above calculation,
   $\lti[\dval{oip}]{\dval{nhip}} \geq \nowN[N^\ddagger]\geq \nowN - \traversal$.
   Using that $\dval{oip}\in\akd{\dval{ip}}$ in combination with \Eq{vD_leq_time}
   we have $\renewcommand{\deltime}{\keyw{DEL\_PRD}}
   \nowN  - \traversal > \lti{\dval{ip}} + \nodetraversal- \deltime $; hence the invariant 
   is maintained.

\hspace{1em}
   Suppose $\dval{oip}\in\ikd{\dval{nhip}}$. Then
   $\renewcommand{\deltime}{\keyw{DELETE\_PRD}}\lti[\dval{oip}]{\dval{nhip}} \geq \nowN[N^\ddagger]+\deltime$ by \Prop{nhip lifetime}
   \highlightspec{(A3)}, so
   $\lti[\dval{oip}]{\dval{nhip}} \geq \nowN+\deltime\linebreak[1]-\traversal > \lti[\dval{oip}]{\dval{ip}} + \nodetraversal$,
   using \Eq{vD_leq_time} and that $\dval{oip}\in\akd{\dval{ip}}$.

 \item[Pro.~\ref{pro:aodv}, Line~\ref{aodv:line26c}; Pro.~\ref{pro:pkt},
   Lines~\ref{pkt2:line6b},~\ref{pkt2:line6c},~\ref{pkt2:line6d}; Pro.~\ref{pro:rrep},
   Line~\ref{rrep:line12d}:] \mbox{}\\
   \highlightspec{WRONG}.
\qed
\end{description}
\end{enumerate}
\end{enumerate}
\end{proof}

We have shown that, under Assumption~\ref{ass:rreqvalid}, only $5$ lines of the AODV specification 
invalidate Invariant \Eq{33} of the Intended Theorem~\ref{thm:intended}, namely 
   Line~\ref{aodv:line26c} of Pro.~\ref{pro:aodv}, 
   Lines~\ref{pkt2:line6b},~\ref{pkt2:line6c},~\ref{pkt2:line6d} of Pro.~\ref{pro:pkt} and 
   Line~\ref{rrep:line12d} of Pro.~\ref{pro:rrep}.
If these lines were to be changed in a way that 
preserves this invariant, then 
Assumption~\ref{ass:pre} holds as well. 
Assumption~\ref{ass:underway} holds as a consequence of Assumption~\ref{ass:rreqvalid} (\Thm{A3impliesA2}). 
Hence, if A3 can be ensured and \ref{thm:intended} can be repaired then AODV becomes loop free 
(\Thm{loop free}). In the next section we will show how this can be achieved.

\subsection{Six Routing Loops and their Repair}\label{subapp:loops}

The loop freedom proof above broke down in six places:
the unwarranted Assumption~\ref{ass:rreqvalid} we needed to make,
and the five lines that do not preserve our main invariant.
Below we sketch scenarios showing that each of these flaws actually leads to a case of
premature route expiration, and consequently a routing loop.
We also describe how the protocol could be fixed to avoid these loops.

\subsubsection*{Assumption~\ref{ass:rreqvalid}.}
Assume a 5-node linear topology $C{-}B{-}A{-}E{-}D$.
It can happen that $B$ has an invalid routing table entry to $D$ with a sequence
number that is substantially larger than $D$'s own sequence
number.\footnote{This can happen when $B$ maintains a valid route to
  $D$ and the link $D{-}B$ breaks down and reappears multiple times:
  each time a link break occurs, $B$ invalidates the route to $D$,
  thereby incrementing its destination sequence number; and each time
  the link reappears, $D$ forwards a message to $B$, causing $B$ to
  validate its route to $D$ (Pro.~\ref{pro:aodv},
  Lines~\ref{aodv:line10},~\ref{aodv:line14} or~\ref{aodv:line18})
  without changing its destination sequence number.}
By waiting sufficiently long, it can moreover happen that this routing
table entry has almost reached the end of its lifespan.
Assume that in such a state $D$ initiates two route requests, say
RREQ$_{DA}$ with destination $A$ and RREQ$_{DE}$ for $E$. Right after RREQ$_{DA}$
is sent on its way via $E$ to $A$, the link $D{-}C$ emerges, so that
RREQ$_{DE}$ travels via $C$ and $B$ towards $A$.
When RREQ$_{DE}$ is forwarded by node $B$, $B$ will not update its route
to $D$, because it already has an (invalid) route to $D$ with a higher
sequence number than the one carried by the route request. Yet, it extends the
expiration time of its (invalid) route to $D$ to the value
\[\now+2\cdot\traversal-2\cdot(2+1)\cdot\nodetraversal\]
following Pro.~\ref{pro:rreq}, Line~\ref{rreq:line7}.
When the message reaches $A$, $A$ creates a routing table entry for $D$,
with next hop $B$ and the sequence number carried by RREQ$_{DE}$.

Although RREQ$_{DA}$ has to travel only two hops before reaching $A$, it is possible that it arrives
there after RREQ$_{DE}$. When this happens, Line~\ref{rreq:line6} of Pro.~\ref{pro:rreq} does not give
rise to an update of the routing table entry at $A$ for $D$, since that entry has already a higher
destination sequence number than the one carried by RREQ$_{DA}$. Yet, by Pro.~\ref{pro:rreq}, Line~\ref{rreq:line7}
the entry to $D$ (with next hop $B$) has its expiration time extended to at least
\[\now+2\cdot\traversal-2\cdot(2+1)\cdot\nodetraversal\;,\]
which by now is strictly past the expiration time of the route to $D$ maintained by $B$.
A case of premature route expiration, and a possible routing loop,
results.

The repair of this loop is already sketched in Footnote~\ref{repair}: simply make the broadcast
forwarding the route request (Pro.~\ref{pro:rreq}, Line~\ref{rreq:line34}) conditional on the
existence of a valid route to $\oip$. This assures that Assumption~\ref{ass:rreqvalid} is met.
An even better solution is to make execution of all of Pro.~\ref{pro:rreq},
Lines~\ref{rreq:line9}--\ref{rreq:line37} conditional on $\oip\in\akD{\rt}$.
Besides preventing the routing loop indicated above, this is a strict improvement of AODV on all
counts, as none of the actions taken in Pro.~\ref{pro:rreq}, Lines~\ref{rreq:line9}--\ref{rreq:line37}
makes any sense if $\oip\notin\akD{\rt}$.

\subsubsection*{Pro.~\ref{pro:aodv}, Line~\ref{aodv:line26c}.}
We now consider the topology $A\wideparen{{-}B{-}}C{-}D$.
It can happen that a node $A$ has a valid routing table entry to a
destination $D$ with next hop $C$, but 
the routing table entry for $C$ at node $A$ has a hop count 
strictly larger than $1$, and next hop $B\neq C$.
One of the ways this can happen is when the link $A{-}C$ breaks down
(after the route $A{-}C{-}D$ has been established)
and $C$ initiates a route request that reaches $A$ via $B$;
right afterwards, the link $A{-}C$ is restored, before the link break
could impede any unicast.\footnote{The described situation can also
  arise without any link breaks; we leave the ``how'' as a puzzle for the reader.}
In such a situation, whenever $A$ sends a data packet to $D$, via $C$,
Line~\ref{aodv:line26c} extends the lifetime of $A$'s routing table
entry to $C$. Yet the route from $B$ to $C$ is never used and will
eventually expire and be deleted. This gives rise a case of premature route expiration.

The resulting routing loop can be avoided by changing the AODV
specification in such a way that Line~\ref{aodv:line26c} is executed
only when the hop count of the route to $\nhop{\rt}{\dip}$ is 1.
This is the only situation where there is a rationale for this
instruction in the first place. The invariant of the Intended
Theorem~\ref{main} is then preserved by Line~\ref{aodv:line26c}
because afterwards the antecedent $\dval{nhip}\neq\dval{dip}$ is not met.
An alternative is to skip this line altogether; of course this could
yield shorter lifetimes of certain routing table entries.

\subsubsection*{Pro.~\ref{pro:pkt}, Line~\ref{pkt2:line6b}.}
The routing loop arises just as in the previous case, except that the
data packets from $A$ to $D$ are now forwarded by $A$ rather than
originating from $A$. The repair is the same as well.

\subsubsection*{Pro.\,\ref{pro:pkt}, Line\,\ref{pkt2:line6c}.}
Let us now have a look at the following topology: $A{<}{\substack{\textstyle B_{1}\\\textstyle B_{2}}}{>}C{-}{D}$.
It can happen that first a route $A{-}B_1{-}C{-}D$ is established, and
later $C$ finds a fresher route to $A$ via node $B_2$.
A stream of data packets from $A$ to $D$, via $B_1$ and $C$,
will cause node $C$, at Pro.~\ref{pro:pkt}, Line~\ref{pkt2:line6c} to
extend the lifetime of the route to $A$ (via $B_2$).
But the routing table entry for $A$ at $B_2$ will eventually expire
and disappear, giving rise to a case of premature route expiration.

A possible repair is to simply skip Line~\ref{pkt2:line6c}.
For routes need not be bidirectional: if the route to $\oip$ is not
used, a stream of data packets \emph{from} $\oip$ is no reason to keep
the route in the direction $\oip$ alive.

\subsubsection*{Pro.~\ref{pro:pkt}, Line~\ref{pkt2:line6d}.}
This routing loop arises by combining the scenarios of
Lines~\ref{pkt2:line6b} and~\ref{pkt2:line6c}.
The repair is to simply delete this line.

\newcommand{\topo}{%
   \hspace{55pt}%
   \wideparen{{-}B{-}}C\hspace{-32pt}\textcolor{white}{\rule[-1.5pt]{16pt}{10.5pt}}\hspace{-7.6pt}\overline{\phantom{B}\rule{0pt}{7.3pt}}%
   \hspace{-70.5pt}%
   B_{2}\wideparen{{-}B_{1}}\hspace{-12pt}\textcolor{white}{\rule[-1.5pt]{15pt}{10pt}}\hspace{-15pt}\overline{B_{1}\rule{0pt}{7.3pt}\cdots A{-}B}%
   \hspace{16pt}%
   {-}D
   }
   
\subsubsection*{Pro.~\ref{pro:rrep}, Line~\ref{rrep:line12d}.}
Assume a topology $\topo$ in which the
link $B_1{-}A$ recently broke, and suppose a route request
from $A$ looking for node $D$ travels via $B$ and $C$. Afterwards the link $B{-}C$ breaks
and during that time $C$ searches for a new route to $A$. Node $B_2$ answers this route
request, and a route $C{-}B_2{-}B_1{-}A$ is established. This can
happen even if the routing table entry for $A$ at $B_1$ is invalid,
namely when $B_2$ missed all RERR messages announcing this. Only
afterwards comes the route reply from $D$, passing through $C$ and
$B_2$ on its way to $A$. At $B_2$ Pro.~\ref{pro:rrep}, Line~\ref{rrep:line12d}
causes the lifetime of the route to $A$ to be extended. However, it
could be arbitrary long ago that the route from $B_1$ to $A$ was ever
used, and this invalid route may be about to expire. By possibly repeating
this scenario various times (since $A$ never finds a route to $D$), one obtains
a valid routing table entry from $B_2$ to $A$ via $B_1$, while the
routing table entry for $A$ at $B_1$ is deleted.

This case of premature route expiration can be prevented by simply
omitting Line~\ref{rrep:line12d}.  The argument for doing this is again
that routes need not be bidirectional: if the route to $\oip$ is not
used, a route reply that is intended to establish a route
\emph{from} $\oip$ is no reason to keep the route in the direction $\oip$ alive.

\end{document}